\documentclass[journal]{IEEEtran}

\ifCLASSINFOpdf
 
\else
  
\fi

\usepackage{geometry}
\usepackage{color}
\usepackage{setspace}
\usepackage[ruled,vlined]{algorithm2e}
\usepackage{cite}
\usepackage{graphicx}
\usepackage{subfig}
\usepackage{amsfonts}
\usepackage{amsmath}
\usepackage{amsthm}
\usepackage{mathdots}
\usepackage{bm}

\newtheorem{proposition}{Proposition}
\newtheorem{remark}{Remark}
\newtheorem{lemma}{Lemma}

\begin{document}
\title{Location Information Assisted Beamforming Design for Reconfigurable Intelligent Surface Aided Communication Systems}

\author{Zhe Xing,~\IEEEmembership{Graduate Student Member,~IEEE,}
        Rui Wang,~\IEEEmembership{Senior Member,~IEEE,}
        Xiaojun Yuan,~\IEEEmembership{Senior Member,~IEEE,}
        and Jun Wu,~\IEEEmembership{Senior Member,~IEEE}
        \vspace*{-16.8pt}
\thanks{

Z. Xing and R. Wang are with the College of Electronics and Information Engineering, Tongji University, Shanghai 201804, China. R. Wang is also with the Shanghai Institute of Intelligent Science and Technology, Tongji University, Shanghai 201804, China (e-mail: zxing@tongji.edu.cn; ruiwang@tongji.edu.cn).

X. Yuan is with the National Key Laboratory of Science and Technology on Communications, University of Electronic Science and Technology of China, Chengdu, 610000, China (e-mail: xjyuan@uestc.edu.cn).

J. Wu is with the School of Computer Science, Fudan University, Shanghai
200433, China (e-mail: wujun@fudan.edu.cn).
}
}


\markboth{}%
{\MakeLowercase{\textit{et al.}}: Location Information Assisted Beamforming Design for Reconfigurable Intelligent Surface Aided Communication Systems}

\maketitle

\maketitle

\begin{abstract}
In reconfigurable intelligent surface (RIS) aided millimeter-wave (mmWave) communication systems, in order to overcome the limitation of the conventional channel state information (CSI) acquisition techniques, this paper proposes a location information assisted beamforming design without the requirement of the conventional channel training process. First, we establish the geometrical relation between the channel model and the user location, based on which we derive an approximate CSI error bound based on the user location error by means of Taylor approximation, triangle and power mean inequalities, and semidefinite relaxation (SDR). Second, for combating the uncertainty of the location error, we formulate a worst-case robust beamforming optimization problem. To solve the problem efficiently, we develop a novel iterative algorithm by utilizing various optimization tools such as Lagrange multiplier, matrix inversion lemma, SDR, as well as branch-and-bound (BnB). Additionally, {\color{black}we provide sufficient conditions for the SDR to output rank-one solutions, and} modify the BnB algorithm to acquire the phase shift solution under an arbitrary constraint of possible phase shift values. Finally, we analyse the algorithm {\color{black} convergence and} complexity, and carry out simulations to validate the theoretical derivation of the CSI error bound and the robustness of the proposed algorithm. Compared with the existing non-robust approach and the robust beamforming techniques based on S-procedure and penalty convex-concave procedure (CCP), our method can converge more quickly and achieve better performance in terms of the worst-case signal-to-noise ratio (SNR) at the receiver. 
\end{abstract}


\IEEEpeerreviewmaketitle


\section{Introduction}

For decades, the rapid development of telecommunication technologies has been concomitant with a surge of mobile data traffic along with a sharp increase in the number of mobile terminals, which sparks off a burning issue of spectrum scarcity. To deal with this issue, several key enabling techniques, such as millimeter-wave (mmWave), massive multiple-input-multiple-output (MIMO) and ultra-dense network (UDN), have been incorporated into the fifth-generation (5G) wireless communication network \cite{5G Key Technologies}. Although these techniques are validated to be advantageous in terms of improving the spectral efficiency and providing reliable connectivity, they are still unable to adequately address the problems of high energy consumption (EC) and high hardware/deployment cost (HDC). These problems will become even more serious in the future sixth-generation (6G) wireless communication network.

Recently, the urgent demand for coping with the problems of high EC and HDC in 5G/6G, has promoted the emergence and development of a new concept, termed reconfigurable intelligent surface (RIS), which aims to make the wireless communication environment smarter and more controllable for combating the undesirable propagation conditions \cite{C.Liaskos-CM2018, S. Hu-TSP2018}. An RIS is generally an artificial metasurface consisting of a large quantity of near-passive reflecting units, each of which is independently controlled in a software-defined manner to adjust the physical properties, such as phase shifts, of the impinging electromagnetic waves, so as to reflect the waves to a desired receiver \cite{M.D.Renzo-JSAC2020, S.Gong-CST2020, Q.Wu-CM2020, C.Huang-WC2020}. As the RIS can generally be fabricated with low-cost simple electronic components (e.g. varactor diodes \cite{J.Y.Lau-TAP2012} and positive intrinsic-negative (PIN) diodes \cite{L.Dai-Access2020}) and does not need power-consuming radio-frequency (RF) chains to perform active signal retransmission, it has attracted considerable attention and has been envisioned as one of the most promising candidate technologies in 5G/6G \cite{N.Rajatheva-6G2020}. 

Summarized from the existing researches, the RIS is mostly employed to assist the wireless communication or user localization when the line-of-sight (LoS) link is weak or even unavailable, and to improve the system performance including (but not limited to) channel capacity \cite{E. Bjornson-WCL2020}, physical-layer security \cite{M.Cui-WCL2019}, outage probability \cite{C.Guo-CL2020}, robustness \cite{Our-Previous}, spectral/energy efficiency \cite{C.Huang-TWC2019, Q.Wu-TWC2019, Y.Liu-Globecom2020}, achievable data rate\cite{Z.Xing-TWC2021, C.Huang-ICASSP2018, X.Hu-JSAC2020}, potential positioning accuracy \cite{S.Hu-TSP2018, J.He-VTC2020}, etc., under either perfect \cite{E. Bjornson-WCL2020, M.Cui-WCL2019, C.Guo-CL2020, Our-Previous, C.Huang-TWC2019, Q.Wu-TWC2019, S.Hu-TSP2018, J.He-VTC2020, C.Huang-ICASSP2018, X.Hu-JSAC2020} or non-ideal hardware conditions \cite{Y.Liu-Globecom2020, Z.Xing-TWC2021}. To achieve these goals, the phase shift design/optimization, also known as passive beamforming, needs to be performed at the RIS. This implies that the RIS controller should first acquire adequate instantaneous channel state information (CSI) from the base station (BS). 
Although the instantaneous CSI can generally be obtained by various channel estimation techniques developed for the RIS-aided wireless communication \cite{B.Zheng-WCL2020, L.Wei-TCOM2021, Z.He-WCL2020}, several problems still exist during this procedure. First, when the number of the RIS reflecting units is large, conventional CSI estimation techniques require a huge overhead for channel training. Second, when the BS acquires the CSI, the CSI needs to be delivered to the RIS controller via an additional link, which incurs an extra communication overhead \cite{Zhong-TCOM2020}. Third, since the CSI is generally time-variant in real applications, the update of passive beamforming may become hysteretic due to training and/or communication delay, which seriously compromises the potential gain achieved by the use of RIS.

In view of the aforementioned issues, a novel RIS-aided communication scheme based on user location information has been proposed recently \cite{Zhong-TCOM2020}. Specifically, in practice, the positions of the BS and RIS are generally fixed and known after their deployments, while the position of a user can be easily acquired in many ways. {\color{black}For instance, the user position is available when 1) the RIS is designed to support the user positioning, or 2) the user is covered by the external global positioning system (GPS) or ultra-wide band (UWB) signals in the outdoor or indoor environments. The first case has been addressed by several prior works \cite{RIS-L-1,RIS-L-2,RIS-L-3,RIS-L-4}, where the RIS was leveraged to localize the user by means of the maximum likelihood estimation (MLE), the machine learning, as well as the codebook design, and the phase shift configurations and the Cramér-Rao Lower Bound of the potential localization errors were also investigated. The second case is relatively straightforward in the presence of the cooperation of the existing localization system and the RIS-aided communication system \cite{Zhong-TCOM2020}. Once the user position is obtained, according to the spatial geometric relations between the coordinates of the BS, the RIS, the user, and the array responses at each terminal side, the channel matrices or vectors can be directly reconstructed without the channel training process. The reconstructed channels are further dedicated to the subsequent transmit or passive beamforming design/optimization. Therefore, utilizing the location information to acquire the CSI can overcome the drawbacks of the conventional CSI estimation described in the previous paragraph}.

However, the user position from the localization systems is generally inaccurate. Due to the unavoidable localization error caused by inherent accuracy limitation or hardware imperfection of the localization devices, the reconstructed channels suffer from CSI uncertainty. In such a situation, in order to maintain a good system performance, the transmit or passive beamforming should be designed to be robust to the CSI uncertainty. Up to now, there have already been several important prior works which studied the robust beamforming in the RIS-aided communication systems. For instance, G. Zhou, \textit{et al.} \cite{G.Zhou-TSP2020, G.Zhou-WCL2020}, considered both the bounded CSI error and the statistical CSI error, and designed robust transmit and passive beamforming optimization algorithms based on S-procedure and penalty convex-concave procedure (CCP); J. Zhang, \textit{et al.} \cite{J.Zhang-CL2020}, and M. Zhao, \textit{et al.} \cite{M.Zhao-TSP2021}, focused on minimizing the average mean squared error (MSE) of the data symbols, or minimizing the transmit power with outage probability constraints, in consideration of the statistical CSI error. However, these works did not investigate the CSI error arising from the user location uncertainty, and might suffer from slow convergence when the number of transmit antennas or reflecting elements became large. To the best of our knowledge, how to achieve an efficient robust transmit and passive beamforming optimization approach based on the location information has not been explored yet, which motivates our research herein. {\color{black}Although in our previous conference paper \cite{Our-Previous}, we addressed this remaining issue by relating the CSI error bound to the user location error bound and designing a non-iterative worst-case robust beamforming optimization scheme based on the saddle point theory, the proposed design in \cite{Our-Previous} was only suitable for a special far-field communication scenario, where the investigated BS-RIS channel was characterized by a deterministic rank-one LoS model. Instead, in this work, we propose a novel iterative worst-case robust beamforming optimization approach for the RIS-aided location information assisted wireless communication system, which applies to a more general communication scenario in the presence of Rician fading}. In such a case, compared to \cite{Our-Previous}, both the derivation of the CSI error bound and the optimization process of the transmit and passive beamforming are essentially different and more challenging. 

The contributions of this paper are summarized as follows.

\begin{itemize}
\item[•] \textbf{Derivation of the CSI error bound}: First, we consider a three-dimensional (3D) RIS-aided mmWave communication system, where the locations of the BS, the RIS and the user are adopted to acquire the CSI of the BS-RIS-user link. Under the assumption that the user location error is restricted in a spherical region \cite{Zhong-TCOM2020}, {\color{black}we combine the Taylor expansion, the triangle inequality, the power mean inequality, and the semidefinite relaxation (SDR) to derive an approximate CSI error bound according to the user location error bound. This new bound has not been obtained by the previous studies to the best of our knowledge, but is essential to our subsequent robust transmit and passive beamforming design}. The CSI error bound is empirically verified to be tight when the user location uncertainty is moderate, or when the user is far away from the RIS.

\item[•] \textbf{Robust transmit and passive beamforming}: After the CSI error bound is derived, a worst-case robust transmit and passive beamforming optimization problem is formulated. Since the original problem is non-convex and difficult to be solved efficiently, {\color{black}a novel iterative optimization approach is proposed to acquire a suboptimal solution. Specifically, the inner minimization is first conducted by introducing a Lagrange multiplier, and the outer maximization is then iteratively accomplished to obtain the optimal solutions of the Lagrange dual variable and the phase-shift matrix}. Owing to the constant-modulus property of the reflectors, the outer maximization problem is {\color{black}decomposed} into a feasibility-check problem and a constrained QCQP problem with the aid of the matrix inversion lemma, where the QCQP problem is solved by the SDR when the phase shift arguments belong to $[0,2\pi]$, or by the branch-and-bound (BnB) algorithm when the phase shift arguments are arbitrarily constrained. {\color{black}Regarding the SDR for the RIS phase shift optimization, our work makes the first attempt to rigorously show that the SDR yields rank-one solutions when certain regularity condition holds.}

\item[•] \textbf{Algorithm design and performance evaluation}: Based on the proposed optimization approach, the overall algorithm is finally built, and its {\color{black}convergence behaviour} and computational complexity are analysed. Afterwards, the theoretical derivation of the CSI error bound is verified, and the performance of the developed algorithm is evaluated numerically. Compared with the conventional non-robust approach, the proposed algorithm shows strong robustness against the CSI uncertainty. Compared with the existing worst-case robust beamforming methods based on S-procedure and penalty CCP, our proposed approach can perform better and converge more quickly in terms of the worst-case signal-to-noise ratio (SNR) at the receiver. Besides, our proposed algorithm with BnB has the advantages of being able to provide near-optimal solutions and deal with arbitrary phase shift argument sets.
\end{itemize}

The rest of this paper is organized as follows. Section II describes the system model and the problem formulation for the robust transmit and passive beamforming optimization. Section III derives the CSI error bound based on the user location error bound. Section IV proposes an iterative optimization approach to acquire the solutions of the transmit beamforming vector and the reflective phase-shift matrix. Section V carries out the simulations and comparisons to evaluate the algorithm performance. Section VI draws the final conclusions and prospects.

\textit{Notations:} $\left[\mathbf{v}\right]_i$ and $\left[\mathbf{M}\right]_{(i,i)}$ represent the $i$-th and $(i,i)$-th element in vector $\mathbf{v}$ and matrix $\mathbf{M}$. $\mathbf{M}^{\mathrm{T}}$, $\mathbf{M}^*$ and $\mathbf{M}^{\mathrm{H}}$ denote the transpose, conjugate and conjugate transpose of $\mathbf{M}$. $\mathbf{M}\in\mathbb{C}^{a\times b}$ means that $\mathbf{M}$ is an $a\times b$ sized complex-value matrix. $\|.\|_2$ symbolizes the $\ell_2$-norm. $\mathbb{E}\{\cdot\}$, $tr(\cdot)$ and $\arg\{\cdot\}$ denote the expectation, trace and argument, respectively. $\mathrm{diag}(x_1,x_2,...,x_n)$ represents a diagonal matrix whose diagonal elements are $(x_1,x_2,...,x_n)$. $\mathfrak{Re}\{x\}$ and $\mathfrak{Im}\{x\}$ denote the real part and the imaginary part of $x$, respectively. $\mathbf{I}_{a\times a}$ is an $a\times a$ sized identity matrix.

\section{System Model and Problem Formulation}

\begin{figure}[!t]
\includegraphics[width=3in]{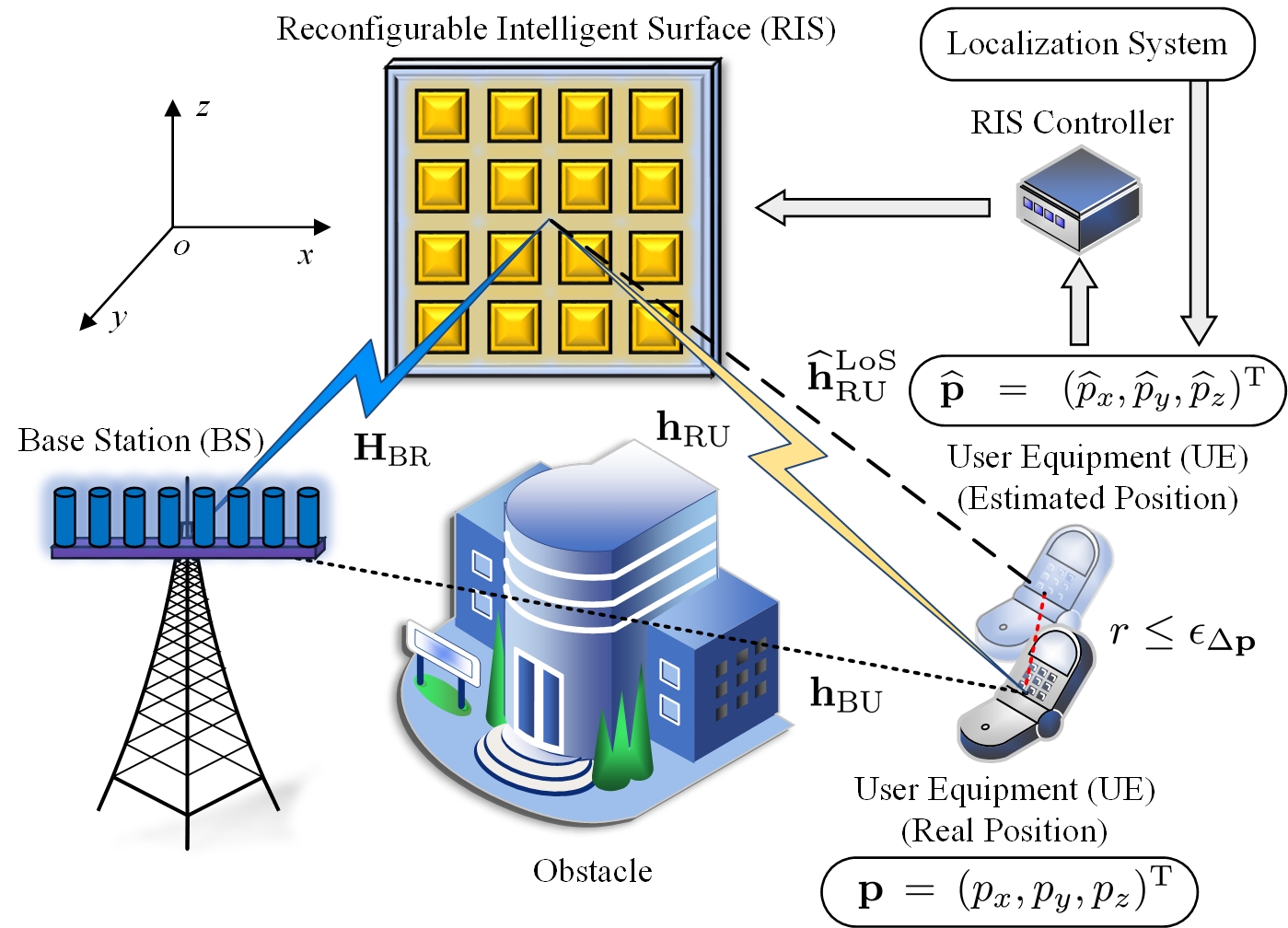}
\hfil
\centering
\caption{\color{black}The considered RIS-aided wireless communication system in a 3D propagation environment. 
The user location $\widehat{\mathbf{p}}$, acquired from the existing localization systems, is adopted to reconstruct the LoS channel from the RIS to the UE. The actual user location $\mathbf{p}$ is assumed to be within a spherical region with radius of $\epsilon_{\Delta\mathbf{p}}$ and center of $\widehat{\mathbf{p}}$.
}
\label{IRS localization--blockage-system}
\end{figure}

In this paper, an RIS-aided mmWave communication system in a 3D propagation environment, as shown in Fig. \ref{IRS localization--blockage-system}, is investigated. The system is composed of a BS with a uniform linear array (ULA) consisting of $M$ antennas, an RIS with a uniform square planar array (USPA) consisting of $L\times L=N$ reflecting elements with $L$ being the number of rows or columns of the USPA, a {\color{black}user equipment (UE)} with one antenna, and an RIS controller. {\color{black}The UE is assumed moving slowly or nearly static}. The antenna spacing of the BS is $d_{BS}$, while the element spacing of the RIS is $d_{RIS}$. To guarantee a good communication performance, the RIS is deployed near the BS.
For convenience of system description, a 3D Cartesian coordinate system, as shown in the top-left corner of Fig. 1, is established to specify the positions of the BS, the RIS and the {\color{black}UE}. With the aid of coordinate indication, the position of the {\color{black}UE} is represented by $\mathbf{p}=(p_x,p_y,p_z)^{\mathrm{T}}$. Assume that the ULA on the BS and the USPA on the RIS are deployed parallel to the $x$-axis and the $x$-$o$-$z$ plane, respectively. Then, let $\mathbf{q}_1=(q_{x,1},q_{y,1},q_{z,1})^{\mathrm{T}}$ denote the coordinates of the left-end BS antenna, and $\mathbf{v}_1=(v_{x,1},v_{y,1},v_{z,1})^{\mathrm{T}}$ denote the coordinates of the reflecting element in the bottom-left corner of the RIS. Hence, the positions of the $i$-th antenna on the BS, for $i=1,2,3,...,M$, and the $(\ell+(k-1)L)$-th reflecting element on the RIS, for $\ell=1,2,...,L$ and $k=1,2,...,L$, are represented by
\begin{equation}\label{q_i}
\mathbf{q}_i = \mathbf{q}_1 + \bm{\delta}_{\mathbf{q}}(i),
\end{equation}
\begin{equation}\label{v_(ell+(k-1)L)}
\mathbf{v}_{\ell+(k-1)L} = \mathbf{v}_1 + \bm{\delta}_{\mathbf{v}}(\ell+(k-1)L),
\end{equation}
where $\bm{\delta}_{\mathbf{q}}(i)=((i-1)d_{BS},0,0)^{\mathrm{T}}$ and $\bm{\delta}_{\mathbf{v}}(\ell+(k-1)L)=((\ell-1)d_{RIS},0,(k-1)d_{RIS})^{\mathrm{T}}$. 

Based on this system setup, the remainder of this section will illustrate the channel and received signal models, CSI error model, and problem formulation.

\subsection{Channel and Received Signal Models}

In the considered communication system, {\color{black} an obstruction, e.g. an edifice or building, exists on the direct path between the BS and the UE, whereas no object blocks the LoS links between the BS/RIS and the RIS/UE. In this situation, the baseband equivalent BS-UE channel, denoted by $\mathbf{h}_{\mathrm{BU}}$, is characterized by the Rayleigh fading model, whose elements follow zero mean complex Gaussian distributions \cite{C.Pan-JSAC2020}}. The baseband equivalent BS-RIS and RIS-UE channels, denoted respectively by $\mathbf{H}_{\mathrm{BR}}$ and $\mathbf{h}_{\mathrm{RU}}$, are modelled as \cite{Q.Wu-TWC2019}
\begin{equation}
\mathbf{H}_{\mathrm{BR}}=
\sqrt{\frac{\kappa_R}{1+\kappa_R}}\mathbf{H}_{\mathrm{BR}}^{\mathrm{LoS}} + \sqrt{\frac{1}{1+\kappa_R}}\mathbf{H}_{\mathrm{BR}}^{\mathrm{NLoS}},
\end{equation}
\begin{equation}
\mathbf{h}_{\mathrm{RU}}=
\sqrt{\frac{\kappa_R}{1+\kappa_R}}\mathbf{h}_{\mathrm{RU}}^{\mathrm{LoS}} + \sqrt{\frac{1}{1+\kappa_R}}\mathbf{h}_{\mathrm{RU}}^{\mathrm{NLoS}},
\end{equation}
where $\kappa_R$ represents the Rician factor;
$\mathbf{H}_{\mathrm{BR}}^{\mathrm{NLoS}}\in\mathbb{C}^{N\times M}$ and $\mathbf{h}_{\mathrm{RU}}^{\mathrm{NLoS}}\in\mathbb{C}^{N\times 1}$ are the non-LoS Rayleigh fading components \cite{C.Pan-JSAC2020};
$\mathbf{H}_{\mathrm{BR}}^{\mathrm{LoS}}\in\mathbb{C}^{N\times M}$ and $\mathbf{h}_{\mathrm{RU}}^{\mathrm{LoS}}\in\mathbb{C}^{N\times 1}$ are the LoS components given by (\ref{H_B-R}) and (\ref{h_R-M}) on the top of the next page \cite{Z. A. Shaban-Arxiv2021},
\begin{figure*}[!t]
\normalsize
\begin{scriptsize}
\begin{equation}\label{H_B-R}
\begin{split}
\mathbf{H}_{\mathrm{BR}}^{\mathrm{LoS}}=
\left(\begin{matrix}
\sqrt{\rho_{Loss}^{B_1-R_1}} e^{j\frac{2\pi}{\lambda}\vartheta_{1,1}}& \sqrt{\rho_{Loss}^{B_2-R_1}} e^{j\frac{2\pi}{\lambda}\vartheta_{1,2}}&\cdots&\sqrt{\rho_{Loss}^{B_M-R_1}} e^{j\frac{2\pi}{\lambda}\vartheta_{1,M}}\\
\sqrt{\rho_{Loss}^{B_1-R_2}} e^{j\frac{2\pi}{\lambda}\vartheta_{2,1}}& \sqrt{\rho_{Loss}^{B_2-R_2}} e^{j\frac{2\pi}{\lambda}\vartheta_{2,2}}&\cdots&\sqrt{\rho_{Loss}^{B_M-R_2}} e^{j\frac{2\pi}{\lambda}\vartheta_{2,M}}\\
\vdots&\vdots&\ddots&\vdots\\
\sqrt{\rho_{Loss}^{B_1-R_N}} e^{j\frac{2\pi}{\lambda}\vartheta_{N,1}}& \sqrt{\rho_{Loss}^{B_2-R_N}} e^{j\frac{2\pi}{\lambda}\vartheta_{N,2}}&\cdots&\sqrt{\rho_{Loss}^{B_M-R_N}} e^{j\frac{2\pi}{\lambda}\vartheta_{N,M}}\\
\end{matrix}\right),
\end{split}
\end{equation}
\end{scriptsize}
\begin{scriptsize}
\begin{equation}\label{h_R-M}
\mathbf{h}_{\mathrm{RU}}^{\mathrm{LoS}}=\left(\sqrt{\rho_{Loss}^{R_1-U}} e^{j\frac{2\pi}{\lambda}\varpi_1},\ \sqrt{\rho_{Loss}^{R_2-U}}e^{j\frac{2\pi}{\lambda}\varpi_2},\ \cdots,\ \sqrt{\rho_{Loss}^{R_N-U}}e^{j\frac{2\pi}{\lambda}\varpi_{N}}\right)^{\mathrm{T}}.
\end{equation}
\vspace*{-10pt}
\end{scriptsize}
\hrulefill
\vspace*{-8pt}
\end{figure*}
where $\lambda$ is the signal wavelength; $\vartheta_{\ell+(k-1)L,i}$ and $\varpi_{\ell+(k-1)L}$, for $\ell=1,2,...,L$, $k=1,2,...,L$ and $i=1,2,...,M$, are expressed as
\begin{equation}\label{In_H_B-R}
\vartheta_{\ell+(k-1)L,i}=\|\mathbf{v}_1-\mathbf{q}_1\|_2 - \|\mathbf{v}_{\ell+(k-1)L}-\mathbf{q}_i\|_2,
\end{equation}
\begin{equation}\label{In_h_R-M}
\varpi_{\ell+(k-1)L}=\|\mathbf{v}_{\ell+(k-1)L}-\mathbf{p}\|_2-\|\mathbf{v}_1-\mathbf{p}\|_2.
\end{equation}
$\rho_{Loss}^{B_i-R_{\ell+(k-1)L}}$ and $\rho_{Loss}^{R_{\ell+(k-1)L}-U}$ are, respectively, the large-scale path loss coefficients from the $i$-th BS antenna to the $(\ell+(k-1)L)$-th RIS element and from the $(\ell+(k-1)L)$-th RIS element to the {\color{black}UE}, given by \cite{Q.Wu-TWC2019}
\begin{equation}\label{pathloss_definition_BR}
\rho_{Loss}^{B_i-R_{\ell+(k-1)L}}=\zeta_0\left(\frac{d_{B_i-R_{\ell+(k-1)L}}}{d_0}\right)^{-\alpha},
\end{equation}
\begin{equation}\label{pathloss_definition_RM}
\rho_{Loss}^{R_{\ell+(k-1)L}-U}=\zeta_0\left(\frac{d_{R_{\ell+(k-1)L}-U}}{d_0}\right)^{-\alpha},
\end{equation}
where $d_0$ is the reference distance; $\zeta_0$ denotes the path loss at the reference distance of $d_0=1$ m; $\alpha$ represents the path loss exponent; $d_{B_i-R_{\ell+(k-1)L}}$ and $d_{R_{\ell+(k-1)L}-U}$ are given by
\begin{equation}\nonumber
d_{B_i-R_{\ell+(k-1)L}}=\|\mathbf{v}_{\ell+(k-1)L}-\mathbf{q}_i\|_2,
\end{equation}
\begin{equation}\nonumber
d_{R_{\ell+(k-1)L}-U}=\|\mathbf{v}_{\ell+(k-1)L}-\mathbf{p}\|_2. 
\end{equation}




Based on the above channel models, the signal received by the {\color{black}UE} is modelled as 
{\color{black}\begin{equation}
y(t)=\sqrt{P_T}\left(\mathbf{h}_{\mathrm{RU}}^{\mathrm{H}}\mathbf{\Theta}\mathbf{H}_{\mathrm{BR}} + e_\mathrm{BU} \mathbf{h}_{\mathrm{BU}}^{\mathrm{H}}\right)\mathbf{w}x(t)+n(t),
\end{equation}
where $e_\mathrm{BU}=\{0,1\}$ is a binary variable representing the eventuality of the BS-UE channel, and in this paper, we consider the existence of $\mathbf{h}_{\mathrm{BU}}$ with $e_\mathrm{BU}=1$;} $P_T$ denotes the total transmit power; 
$x(t)$ represents the transmit symbol which satisfies $\mathbb{E}\{x^*(t)x(t)\}=1$; $n(t)\sim\mathcal{CN}(0,\sigma_n^2)$ denotes the receiver noise; $\mathbf{w}$ is the unit-norm transmit beamforming vector; $\mathbf{\Theta}$ is the phase-shift matrix of the RIS, given by $\mathbf{\Theta}=\mathrm{diag}\left(\beta_1 e^{j\theta_1},\beta_2 e^{j\theta_2},\cdots,\beta_N e^{j\theta_N}\right)$, where $\beta_{\ell+(k-1)L}$ and $\theta_{\ell+(k-1)L}$ denote the amplitude and phase shift argument of the $(\ell+(k-1)L)$-th reflecting element, respectively. As the RIS is a near-passive reflecting apparatus, {\color{black}we assume that the reflection amplitudes satisfy $\beta_1=\beta_2=\cdots=\beta_N=\beta$ \cite{Z.Xing-TWC2021} without loss of generality, where $0<\beta\leq 1$}. The phase shift arguments belong to a set $\mathcal{S}$, where $\mathcal{S}$ can be one of the following categories: 

\!\!\!\!\!\!\!\! 1) $\mathcal{S}=[0,2\pi]$: the simplest and most universal argument set.

\!\!\!\!\!\!\!\! 2) $\mathcal{S}=[\ell_l,\ell_u]$: a general argument set restricted within an arbitrary interval, where $\ell_l$ and $\ell_u$ are real values satisfying $\ell_l<\ell_u$.

\!\!\!\!\!\!\!\! 3) $\mathcal{S}=\{0,\Delta\theta,\cdots,(W_{\theta}-1)\Delta\theta\}$: a discrete argument set with $W_{\theta}$ phase shift levels, where $\Delta\theta=2\pi/W_{\theta}$ \cite{Y.Liu-Globecom2020}.

\subsection{CSI Error Model}

Considering that the locations of BS and RIS are stationary, the channel condition between BS and RIS hardly changes {\color{black}over} time. Moreover, because the RIS is close to the BS and the mmWave channel has limited scattering, the LoS component in the BS-RIS channel is dominant whereas the non-LoS component is substantially weak, resulting in $\mathbf{H}_{\mathrm{BR}}\approx\mathbf{H}_{\mathrm{BR}}^{\mathrm{LoS}}$ \cite{K.Zhi-WCL2021, Z.Xing-TWC2021}. 
Therefore, {\color{black}by assuming that the exact locations of each antenna on the BS and each reflecting element on the RIS are fixed and known after their deployments, $\mathbf{H}_{\mathrm{BR}}$ can be perfectly obtained by (\ref{H_B-R}), (\ref{In_H_B-R}) and (\ref{pathloss_definition_BR}).}
{\color{black}The UE is considered within the coverage areas of the external localization systems, e.g. the GPS or UWB, which can provide the UE location for the RIS-aided communication system to} acquire (reconstruct) $\mathbf{h}_{\mathrm{RU}}^{\mathrm{LoS}}$ based on (\ref{h_R-M}), (\ref{In_h_R-M}) and (\ref{pathloss_definition_RM}).\footnote{\color{black}It is remarkable that the UE may occasionally appear in the blind zone of the existing localization systems, where the UE location is hardly available. Under this circumstance, conventional cascaded channel estimation techniques, which have been widely investigated in, e.g. \cite{B.Zheng-WCL2020, L.Wei-TCOM2021, Z.He-WCL2020}, are responsible for CSI acquisition. }
 {\color{black}The acquired RIS-UE channel, denoted by $\widehat{\mathbf{h}}_{\mathrm{RU}}^{\mathrm{LoS}}$ herein, contains an unavoidable error caused primarily by two factors:}

•\textbf{User location error}: In view of the inherent hardware imperfection, limited precision of measurements or some other unfavorable aspects, 
the user location obtained, e.g. from GPS or UWB, denoted by $\mathbf{\widehat{p}}=(\widehat{p}_x,\widehat{p}_y,\widehat{p}_z)^{\mathrm{T}}$, is generally inaccurate. By referring to \cite{Zhong-TCOM2020}, $\mathbf{\widehat{p}}$ satisfies
$\mathbf{p}=\mathbf{\widehat{p}}+\Delta\mathbf{p}$,
where $\Delta\mathbf{p}=(\Delta p_x,\Delta p_y, \Delta p_z)^{\mathrm{T}}$ is the user location error, bounded by $\|\Delta\mathbf{p}\|_2=r \leq \epsilon_{\Delta\mathbf{p}}$, where $\epsilon_{\Delta\mathbf{p}}$ is a known small positive constant depending on the localization accuracy \cite{Zhong-TCOM2020}. 

•\textbf{Non-LoS component}: The non-LoS Rayleigh fading component may not be neglected when the {\color{black}UE stands} far away from the RIS. Although $\mathbf{h}_{\mathrm{RU}}^{\mathrm{NLoS}}$ is random, its $\ell_2$-norm within a certain communication duration can be determined by the existing channel norm feedback techniques based on pilot data sequence \cite{D.Hammarwall-TSP2008, E. Bjornson-ICASSP2008} or on channel correlation matrix \cite{E. Bjornson-TSP2009}. When the {\color{black}UE} moves slowly, the $\ell_2$-norm of $\mathbf{h}_{\mathrm{RU}}^{\mathrm{NLoS}}$ can be regarded approximately as a fixed term in a short period, since the variation of channel condition between RIS and {\color{black}UE} is slow. Thus, we can assume that $\|\mathbf{h}_{\mathrm{RU}}^{\mathrm{NLoS}}\|_2$ is known and given by $\|\mathbf{h}_{\mathrm{RU}}^{\mathrm{NLoS}}\|_2=\delta_{\mathrm{RU}}^{\mathrm{NLoS}}$. {\color{black} Similarly, we can also assume that $\|\mathbf{h}_{\mathrm{BU}}\|_2$ is known and given by $\|\mathbf{h}_{\mathrm{BU}}\|_2=\delta_{\mathrm{BU}}$ for the same reason, although $\mathbf{h}_{\mathrm{BU}}$ is stochastic and cannot be reconstructed by $\mathbf{\widehat{p}}$ either. These values are conducive to facilitating the subsequent derivation of the CSI error bound and the beamforming optimization process.}

{\color{black}Therefore, when $\mathbf{\widehat{p}}$ is used to acquire (reconstruct) the LoS part of the RIS-UE channel, the overall CSI error between the actual and reconstructed channels can be expressed as
\begin{equation}\label{Overall_CSI_Error}
\begin{split}
\Delta\mathbf{h}_{\mathrm{RU}}
= & \mathbf{h}_{\mathrm{RU}} - \widehat{\mathbf{h}}_{\mathrm{RU}}^{\mathrm{LoS}}\\ = &
\sqrt{\frac{\kappa_R}{1+\kappa_R}}\mathbf{h}_{\mathrm{RU}}^{\mathrm{LoS}} + \sqrt{\frac{1}{1+\kappa_R}}\mathbf{h}_{\mathrm{RU}}^{\mathrm{NLoS}} -  \widehat{\mathbf{h}}_{\mathrm{RU}}^{\mathrm{LoS}} \\
=&\left( \sqrt{\frac{\kappa_R}{1+\kappa_R}}\mathbf{h}_{\mathrm{RU}}^{\mathrm{LoS}} -  \widehat{\mathbf{h}}_{\mathrm{RU}}^{\mathrm{LoS}} \right) + \sqrt{\frac{1}{1+\kappa_R}}\mathbf{h}_{\mathrm{RU}}^{\mathrm{NLoS}} \\
=& \Delta\mathbf{h}_{\mathrm{RU}}^{\mathrm{LoS}} + \Delta\mathbf{h}_{\mathrm{RU}}^{\mathrm{NLoS}} ,
\end{split}
\end{equation}
where $\widehat{\mathbf{h}}_{\mathrm{RU}}^{\mathrm{LoS}}$ denotes the LoS part of the RIS-UE channel acquired by $\widehat{\mathbf{p}}$, which can be calculated via
\begin{equation}
\widehat{\mathbf{h}}_{\mathrm{RU}}^{\mathrm{LoS}}=\left(\sqrt{\widehat{\rho}_{Loss}^{R_1-U}} e^{j\frac{2\pi}{\lambda}\widehat{\varpi}_1},...,\sqrt{\widehat{\rho}_{Loss}^{R_N-U}} e^{j\frac{2\pi}{\lambda}\widehat{\varpi}_{N}}\right)^{\mathrm{T}},
\end{equation}
where $\widehat{\varpi}_{\ell+(k-1)L}$ is expressed as 
$\widehat{\varpi}_{\ell+(k-1)L}=\|\mathbf{v}_{\ell+(k-1)L}-\widehat{\mathbf{p}}\|_2-\|\mathbf{v}_1-\widehat{\mathbf{p}}\|_2$,
and $\widehat{\rho}_{Loss}^{R_{\ell+(k-1)L}-U}$ is computed by (\ref{pathloss_definition_RM}) using
\begin{equation}\nonumber
\widehat{d}_{R_{\ell+(k-1)L}-U}=\|\mathbf{v}_{\ell+(k-1)L}-\widehat{\mathbf{p}}\|_2.
\end{equation}

From (\ref{Overall_CSI_Error}), it is observed that $\Delta \mathbf{h}_{\mathrm{RU}}^{\mathrm{LoS}}$, representing the CSI error of the LoS component depending on $\Delta\mathbf{p}$, is given by
\begin{equation}\label{LoS CSI error}
\Delta \mathbf{h}_{\mathrm{RU}}^{\mathrm{LoS}}=\sqrt{\frac{\kappa_R}{1+\kappa_R}}\mathbf{h}_{\mathrm{RU}}^{\mathrm{LoS}} - \widehat{\mathbf{h}}_{\mathrm{RU}}^{\mathrm{LoS}},
\end{equation}
while $\Delta \mathbf{h}_{\mathrm{RU}}^{\mathrm{NLoS}}$, standing for the CSI error caused by the non-LoS component, is given by
\begin{equation}\label{NLoS CSI error}
\Delta \mathbf{h}_{\mathrm{RU}}^{\mathrm{NLoS}} = \sqrt{\frac{1}{1+\kappa_R}}\mathbf{h}_{\mathrm{RU}}^{\mathrm{NLoS}}.
\end{equation}

Eq. (\ref{LoS CSI error}) and (\ref{NLoS CSI error}) indicate that in the mmWave communication with a large $\kappa_R$, $\Delta \mathbf{h}_{\mathrm{RU}}^{\mathrm{LoS}}$ generally dominates the entire CSI error $\Delta \mathbf{h}_{\mathrm{RU}}$.
In the presence of $\Delta \mathbf{h}_{\mathrm{RU}}$, we aim to optimize $\mathbf{w}$ and $\mathbf{\Theta}$ based on {\color{black}$\widehat{\mathbf{h}}_{\mathrm{RU}}^{\mathrm{LoS}}$} and $\mathbf{H}_{\mathrm{BR}}$, by maximizing the worst-case SNR at the receiver}. 




\subsection{Problem Formulation}

For the aforementioned system model, the worst-case robust beamforming optimization problem is formulated {\color{black}as
\begin{subequations}\label{Initial_Problem}
\begin{align}
\mathop{\max}\limits_{\mathbf{w},\mathbf{\Theta}}\mathop{\min}\limits_{\Delta\mathbf{h}_{\mathrm{RU}},  \mathbf{h}_{\mathrm{BU}}}&\!\!
\ \left|\left[(\widehat{\mathbf{h}}_{\mathrm{RU}}^{\mathrm{LoS}}+\Delta\mathbf{h}_{\mathrm{RU}})^{\mathrm{H}} \mathbf{\Theta}\mathbf{H}_{\mathrm{BR}}+\mathbf{h}_{\mathrm{BU}}^{\mathrm{H}}\right]\mathbf{w}\right|,
\\ \mathrm{s.t.} &\ \|\mathbf{w}\|_2=1,
\\&\ \|\Delta \mathbf{h}_{\mathrm{RU}}\|_2\leq \epsilon_{\Delta \mathbf{h}_{\mathrm{RU}}},
\\&\ \|\mathbf{h}_{\mathrm{BU}}\|_2=\delta_{\mathrm{BU}},
\\&\ |[\mathbf{\Theta}]_{(i,i)}|=\beta,\ \ i=1,2,...,N,
\\&\ \arg\left\{[\mathbf{\Theta}]_{(i,i)}\right\} \in \mathcal{S},\ i=1,2,...,N,
\end{align}
\end{subequations}
where} the constant terms of $P_T$ and $\sigma_n^2$ are omitted for conciseness, without changing the optimization results. Constraint (\ref{Initial_Problem}b) comes from the unit-norm transmit beamforming vector. Constraint (\ref{Initial_Problem}c) means that the overall CSI error is bounded in a spherical uncertainty region. {\color{black} Constraint (\ref{Initial_Problem}d) comes from $\|\mathbf{h}_{\mathrm{BU}}\|_2$, which is fixed and known as illustrated in Section II-B. Constraint (\ref{Initial_Problem}e) comes from the constant-modulus phase shifts of the RIS. Constraint (\ref{Initial_Problem}f) represents the phase shift argument constraint}.

It is noted that because the overall CSI error bound $\epsilon_{\Delta \mathbf{h}_{\mathrm{RU}}}$ in constraint (\ref{Initial_Problem}c) is still undetermined here, it should first be derived on the basis of $\|\Delta\mathbf{p}\|_2\leq \epsilon_{\Delta\mathbf{p}}$ and $\|\mathbf{h}_{\mathrm{RU}}^{\mathrm{NLoS}}\|_2=\delta_{\mathrm{RU}}^{\mathrm{NLoS}}$ before solving problem (\ref{Initial_Problem}). The derivation of $\epsilon_{\Delta \mathbf{h}_{\mathrm{RU}}}$ will be detailed in the next section.

\section{Derivation of the CSI Error Bound}

This section is dedicated to the derivation of the overall CSI error bound in (\ref{Initial_Problem}c) according to $\|\Delta\mathbf{p}\|_2 \leq \epsilon_{\Delta\mathbf{p}}$ and $\|\mathbf{h}_{\mathrm{RU}}^{\mathrm{NLoS}}\|_2=\delta_{\mathrm{RU}}^{\mathrm{NLoS}}$. \footnote{\color{black} The paper is dedicated to a single-user communication scenario. Potential extension of the derivation process in this section to a multi-user case is feasible, since the location information based CSI acquisition and the derivation of the CSI error bound can be performed independently for each single user.} First, when retrospecting (\ref{Overall_CSI_Error}), we have (\ref{First_Expansion_CSI_Error_Norm}) on the top of the next page,
\begin{figure*}[!t]
\normalsize
\begin{scriptsize}
\begin{equation}\label{First_Expansion_CSI_Error_Norm}
\begin{split}
\|\Delta \mathbf{h}_{\mathrm{RU}}\|_2&\leq
\left\|\Delta\mathbf{h}_{\mathrm{RU}}^{\mathrm{LoS}}\right\|_2 + \left\|\Delta\mathbf{h}_{\mathrm{RU}}^{\mathrm{NLoS}}\right\|_2
= \left\|\sqrt{\frac{\kappa_R}{1+\kappa_R}}\left(\mathbf{h}_{\mathrm{RU}}^{\mathrm{LoS}} - \widehat{\mathbf{h}}_{\mathrm{RU}}^{\mathrm{LoS}}\right) + \left(\sqrt{\frac{\kappa_R}{1+\kappa_R}}-1\right)\widehat{\mathbf{h}}_{\mathrm{RU}}^{\mathrm{LoS}}\right\|_2 + \sqrt{\frac{1}{1+\kappa_R}} \delta_{\mathrm{RU}}^{\mathrm{NLoS}}\\
&\leq \sqrt{\frac{\kappa_R}{1+\kappa_R}}\left\|\mathbf{h}_{\mathrm{RU}}^{\mathrm{LoS}} - \widehat{\mathbf{h}}_{\mathrm{RU}}^{\mathrm{LoS}} \right\|_2 + \left(1-\sqrt{\frac{\kappa_R}{1+\kappa_R}}\right)\left\|\widehat{\mathbf{h}}_{\mathrm{RU}}^{\mathrm{LoS}} \right\|_2 + \sqrt{\frac{1}{1+\kappa_R}} \delta_{\mathrm{RU}}^{\mathrm{NLoS}}\\
&\leq \sqrt{\frac{\kappa_R}{1+\kappa_R}}\epsilon_{\mathrm{RU}}^{\mathrm{LoS}} + \left(1-\sqrt{\frac{\kappa_R}{1+\kappa_R}}\right)\left\|\widehat{\mathbf{h}}_{\mathrm{RU}}^{\mathrm{LoS}} \right\|_2 + \sqrt{\frac{1}{1+\kappa_R}} \delta_{\mathrm{RU}}^{\mathrm{NLoS}}.
\end{split}
\end{equation}
\end{scriptsize}
\hrulefill
\vspace*{-15pt}
\end{figure*}
where $\epsilon_{\mathrm{RU}}^{\mathrm{LoS}}$ denotes the upper bound of $\|\mathbf{h}_{\mathrm{RU}}^{\mathrm{LoS}} - \widehat{\mathbf{h}}_{\mathrm{RU}}^{\mathrm{LoS}} \|_2$, which is unknown and should be derived according to $\|\Delta\mathbf{p}\|_2 \leq \epsilon_{\Delta\mathbf{p}}$.

\subsection{Derivation of $\epsilon_{\mathrm{RU}}^{\mathrm{LoS}}$}

Here we begin the derivation of $\epsilon_{\mathrm{RU}}^{\mathrm{LoS}}$ from $\|\mathbf{h}_{\mathrm{RU}}^{\mathrm{LoS}} - \widehat{\mathbf{h}}_{\mathrm{RU}}^{\mathrm{LoS}} \|_2$, which can be expanded into
\begin{equation}\label{CSI_Error_Norm_Expansion}
\begin{split}
\left\|\mathbf{h}_{\mathrm{RU}}^{\mathrm{LoS}} - \widehat{\mathbf{h}}_{\mathrm{RU}}^{\mathrm{LoS}} \right\|_2=\sqrt{\zeta_0 \left(\frac{1}{d_0}\right)^{-\alpha} \Omega\left(\Delta\mathbf{p}\right)},
\end{split}
\end{equation}
where $\Omega\left(\Delta\mathbf{p}\right)$ is a function of $\Delta\mathbf{p}$, expressed as (\ref{Omega_Delta_p}) on the top of the next page,
\begin{figure*}[!t]
\normalsize
\begin{scriptsize}
\begin{equation}\label{Omega_Delta_p}
\begin{split}
\Omega\left(\Delta\mathbf{p}\right)= &
\sum_{k=1}^{L}\sum_{\ell=1}^{L}\left\{\left\|\mathbf{v}_{\ell+(k-1)L}-\widehat{\mathbf{p}}-\Delta\mathbf{p}\right\|_2^{-\alpha} + \left\|\mathbf{v}_{\ell+(k-1)L}-\widehat{\mathbf{p}}\right\|_2^{-\alpha} \right\} 
 - 2\sum_{k=1}^{L}\sum_{\ell=1}^{L}\left\{ \left\|\mathbf{v}_{\ell+(k-1)L}-\widehat{\mathbf{p}}-\Delta\mathbf{p}\right\|_2^{-\frac{\alpha}{2}} \right.
\left. \left\|\mathbf{v}_{\ell+(k-1)L}-\widehat{\mathbf{p}}\right\|_2^{-\frac{\alpha}{2}} \cos{\left(\frac{2\pi}{\lambda} \mathfrak{U}_{\ell+(k-1)L}\right)}
\right\}.
\end{split}
\end{equation}
\end{scriptsize}
\hrulefill
\vspace*{-15pt}
\end{figure*}
with $\mathfrak{U}_{\ell+(k-1)L}$ being detailed as
\begin{equation}\nonumber
\begin{split}
\mathfrak{U}_{\ell+(k-1)L}=\ &\varpi_{\ell+(k-1)L}-\widehat{\varpi}_{\ell+(k-1)L}\\
=\ &\left(\left\|\mathbf{v}_{\ell+(k-1)L}-\widehat{\mathbf{p}}-\Delta\mathbf{p}\right\|_2-\left\|\mathbf{v}_1-\widehat{\mathbf{p}}-\Delta\mathbf{p}\right\|_2\right)\\
&-\left(\left\|\mathbf{v}_{\ell+(k-1)L}-\widehat{\mathbf{p}}\right\|_2-\left\|\mathbf{v}_1-\widehat{\mathbf{p}}\right\|_2\right)\\
\overset{(a)}{\approx}\ &\left(
\frac{(\mathbf{v}_1-\widehat{\mathbf{p}})^{\mathrm{T}}}{\|\mathbf{v}_1-\widehat{\mathbf{p}}\|_2}-\frac{(\mathbf{v}_{\ell+(k-1)L}-\widehat{\mathbf{p}})^{\mathrm{T}}}{\|\mathbf{v}_{\ell+(k-1)L}-\widehat{\mathbf{p}}\|_2}
\right)\Delta\mathbf{p}\\
=\ &\bm{\eta}_{\ell+(k-1)L}^{\mathrm{T}}\Delta\mathbf{p},
\end{split}
\end{equation}
where $\bm{\eta}_{\ell+(k-1)L}=\left(\frac{\mathbf{v}_1-\widehat{\mathbf{p}}}{\|\mathbf{v}_1-\widehat{\mathbf{p}}\|_2}-\frac{\mathbf{v}_{\ell+(k-1)L}-\widehat{\mathbf{p}}}{\|\mathbf{v}_{\ell+(k-1)L}-\widehat{\mathbf{p}}\|_2}\right)$ is independent of $\Delta\mathbf{p}$; derivation $(a)$ uses the property that for a fixed real-valued vector $\mathbf{b}$ and a variable $\mathbf{x}$, $\|\mathbf{b}-\mathbf{x}\|_2$ can be well approximated by its first-order approximation of $\|\mathbf{b}-\mathbf{x}\|_2\approx \|\mathbf{b}\|_2+\left\langle\nabla_{\mathbf{x}=\mathbf{0}} \|\mathbf{b}-\mathbf{x}\|_2,\mathbf{x}\right\rangle=\|\mathbf{b}\|_2-\frac{\mathbf{b}^{\mathrm{T}}}{\|\mathbf{b}\|_2}\mathbf{x}$.

It is remarkable that when $\mathbf{v}_{\ell+(k-1)L}$ and $\widehat{\mathbf{p}}$ are given, deriving $\epsilon_{\mathrm{RU}}^{\mathrm{LoS}}$ is equivalent to deriving the maximum of $\Omega\left(\Delta\mathbf{p}\right)$
under the constraint of $\|\Delta\mathbf{p}\|_2\leq \epsilon_{\Delta\mathbf{p}}$, 
which however, is a challenging task, as $\Delta\mathbf{p}$ appears in both the cosine function and the $\ell_2$-norm. In view of this issue, we will further transform $\Omega\left(\Delta\mathbf{p}\right)$ into a simpler form by means of approximation.

Nevertheless, the approximation of $\Omega\left(\Delta\mathbf{p}\right)$ may lead to an inaccurate $\epsilon_{\mathrm{RU}}^{\mathrm{LoS}}$, implying that the theoretically derived $\epsilon_{\mathrm{RU}}^{\mathrm{LoS}}$ will be either higher or lower than the practical $\epsilon_{\mathrm{RU}}^{\mathrm{LoS}}$. Both of the two outcomes will influence the optimal solution of problem (\ref{Initial_Problem}) to some extent. Fortunately, if the theoretically derived $\epsilon_{\mathrm{RU}}^{\mathrm{LoS}}$ is higher than the practical $\epsilon_{\mathrm{RU}}^{\mathrm{LoS}}$, the worst-case CSI experienced during the optimization process will become even worse than the practical worst-case CSI, hence resulting in a more robust solution for problem (\ref{Initial_Problem}). Consequently, in order to preserve the robustness of the transmit and passive beamforming, we need to derive an approximate \textit{upper bound} of the maximum of $\Omega\left(\Delta\mathbf{p}\right)$ under $\|\Delta\mathbf{p}\|_2\leq \epsilon_{\Delta\mathbf{p}}$. The result is provided in the following Lemma 1.

\begin{lemma}
When $\|\Delta\mathbf{p}\|_2\leq \epsilon_{\Delta\mathbf{p}}$, the approximate upper bound of the maximum of $\Omega\left(\Delta\mathbf{p}\right)$, denoted by $\Omega^{\mathrm{Upp}}_{\mathrm{max}}$, can be derived as
\begin{equation}\label{Omega_upp_max}
\Omega^{\mathrm{Upp}}_{\mathrm{max}}=\max_{\mathbf{P}\succeq\bm{0}}\left\{-\frac{4\pi^4}{3\lambda^4 N}tr^2\left(\mathbf{S}\mathbf{P}\right) + tr\left(\mathbf{R}\mathbf{P}\right)\right\}
\end{equation}
from the solution of the following convex problem:
\begin{subequations}\label{Obtain_CSI_Error_Bound_Problem}
\begin{align}
\mathop{\max}\limits_{\mathbf{P}\succeq\bm{0}}&
\ -\frac{4\pi^4}{3\lambda^4 N}tr^2\left(\mathbf{S}\mathbf{P}\right) + tr\left(\mathbf{R}\mathbf{P}\right),
\\ \mathrm{s.t.} &\ \ tr\left(\mathbf{P}\right)\leq \epsilon_{\Delta\mathbf{p}}^2,
\end{align}
\end{subequations}
where $\mathbf{S}$ and $\mathbf{R}$ are expressed as
\begin{equation}\label{mathbf_S}
\begin{split}
\mathbf{S}=\sum_{k=1}^{L}\sum_{\ell=1}^{L}&\left\{
\left(\left\|\mathbf{v}_{\ell+(k-1)L}-\widehat{\mathbf{p}}\right\|_2-\epsilon_{\Delta\mathbf{p}}\right)^{-\frac{\alpha}{4}}\right.\\
&\left.  \times\left\|\mathbf{v}_{\ell+(k-1)L}-\widehat{\mathbf{p}}\right\|_2^{-\frac{\alpha}{4}}
\mathbf{\Xi}_{\ell+(k-1)L}
\right\},
\end{split}
\end{equation}
\begin{equation}\label{mathbf_R}
\begin{split}
\mathbf{R}=\sum_{k=1}^{L}\sum_{\ell=1}^{L}&\left\{
\frac{1}{2}\mathbf{G}_{\alpha}
- \left\|\mathbf{v}_{\ell+(k-1)L}-\widehat{\mathbf{p}}\right\|_2^{-\frac{\alpha}{2}}\mathbf{G}_{\frac{\alpha}{2}}\right.\\
&\left.  +\frac{4\pi^2}{\lambda^2} \left(\left\|\mathbf{v}_{\ell+(k-1)L}-\widehat{\mathbf{p}}\right\|_2-\epsilon_{\Delta\mathbf{p}}\right)^{-\frac{\alpha}{2}}\right.\\
&\left.  \times\left\|\mathbf{v}_{\ell+(k-1)L}-\widehat{\mathbf{p}}\right\|_2^{-\frac{\alpha}{2}}
\mathbf{\Xi}_{\ell+(k-1)L}
\right\},
\end{split}
\end{equation}
in which $\mathbf{G}_{\alpha}$, $\mathbf{G}_{\frac{\alpha}{2}}$ and $\mathbf{\Xi}_{\ell+(k-1)L}$ are given by
\begin{equation}
\begin{split}
\mathbf{G}_{\alpha}
=\ &\alpha(\alpha+2)\left\|\mathbf{v}_{\ell+(k-1)L}-\widehat{\mathbf{p}}\right\|_2^{-\alpha-4}\\
&\times \left(\widehat{\mathbf{p}}-\mathbf{v}_{\ell+(k-1)L}\right) \left(\widehat{\mathbf{p}}-\mathbf{v}_{\ell+(k-1)L}\right)^{\mathrm{T}}\\
&- \alpha \left\|\mathbf{v}_{\ell+(k-1)L}-\widehat{\mathbf{p}}\right\|_2^{-\alpha-2} \mathbf{I}_{3\times 3},
\end{split}
\end{equation}
\begin{equation}
\begin{split}
\mathbf{G}_{\frac{\alpha}{2}}
=\ &\frac{\alpha}{2}\left(\frac{\alpha}{2}+2\right)\left\|\mathbf{v}_{\ell+(k-1)L}-\widehat{\mathbf{p}}\right\|_2^{-\frac{\alpha}{2}-4}\\
&\times \left(\widehat{\mathbf{p}}-\mathbf{v}_{\ell+(k-1)L}\right) \left(\widehat{\mathbf{p}}-\mathbf{v}_{\ell+(k-1)L}\right)^{\mathrm{T}}\\
&- \frac{\alpha}{2} \left\|\mathbf{v}_{\ell+(k-1)L}-\widehat{\mathbf{p}}\right\|_2^{-\frac{\alpha}{2}-2} \mathbf{I}_{3\times 3},
\end{split}
\end{equation}
\begin{equation}
\mathbf{\Xi}_{\ell+(k-1)L} = \bm{\eta}_{\ell+(k-1)L} \bm{\eta}_{\ell+(k-1)L}^{\mathrm{T}}.
\end{equation}
\end{lemma}
\begin{proof}
The proof is given in Appendix A.
\end{proof}

After obtaining $\Omega^{\mathrm{Upp}}_{\mathrm{max}}$ from (\ref{Omega_upp_max}), according to (\ref{CSI_Error_Norm_Expansion}), we have
\begin{equation}\label{CSI_Error_Bound_Upp}
\begin{split}
\epsilon_{\mathrm{RU}}^{\mathrm{LoS}}
\approx \sqrt{\zeta_0 \left(\frac{1}{d_0}\right)^{-\alpha} \Omega^{\mathrm{Upp}}_{\mathrm{max}}}.
\end{split}
\end{equation}

\vspace*{-15pt}



\subsection{Determination of $\epsilon_{\Delta \mathbf{h}_{\mathrm{RU}}}$}

After deriving $\epsilon_{\mathrm{RU}}^{\mathrm{LoS}}$, based on (\ref{First_Expansion_CSI_Error_Norm}), we consequently obtain
\begin{equation}\label{Finally_Derived_CSIerrorbound}
\begin{split}
\epsilon_{\Delta \mathbf{h}_{\mathrm{RU}}}=&\sqrt{\frac{\kappa_R}{1+\kappa_R}}\epsilon_{\mathrm{RU}}^{\mathrm{LoS}} + \left(1-\sqrt{\frac{\kappa_R}{1+\kappa_R}}\right)\left\|\widehat{\mathbf{h}}_{\mathrm{RU}}^{\mathrm{LoS}} \right\|_2 \\
&+ \sqrt{\frac{1}{1+\kappa_R}} \delta_{\mathrm{RU}}^{\mathrm{NLoS}},
\end{split}
\end{equation}
which completes the derivation of the CSI error bound.


\section{Robust Transmit and Passive Beamforming}

Since $\epsilon_{\Delta \mathbf{h}_{\mathrm{RU}}}$ has been derived and the constraint of (\ref{Initial_Problem}c) has been determined, we are now ready to solve problem (\ref{Initial_Problem}), by proposing a robust beamforming optimization approach using the CSI acquired completely from the location information.

According to (\ref{Initial_Problem}a) and (\ref{Initial_Problem}b), the optimal transmit beamforming vector, denoted by $\overline{\mathbf{w}}$, is readily formed {\color{black} by
\begin{equation}\label{Opt_w}
\overline{\mathbf{w}}=
\frac{\mathbf{H}_{\mathrm{BR}}^{\mathrm{H}} \mathbf{\Theta}^{\mathrm{H}} (\widehat{\mathbf{h}}_{\mathrm{RU}}^{\mathrm{LoS}}+\Delta\mathbf{h}_{\mathrm{RU}}) + \mathbf{h}_{\mathrm{BU}}}
{\|(\widehat{\mathbf{h}}_{\mathrm{RU}}^{\mathrm{LoS}}+\Delta\mathbf{h}_{\mathrm{RU}})^{\mathrm{H}} \mathbf{\Theta}\mathbf{H}_{\mathrm{BR}} + \mathbf{h}_{\mathrm{BU}}^{\mathrm{H}}\|_2}.
\end{equation}
Note that in (\ref{Opt_w}), $\mathbf{\Theta}$, $\Delta\mathbf{h}_{\mathrm{RU}}$ and $\mathbf{h}_{\mathrm{BU}}$ are currently all undetermined variables. Hence, in order to settle $\overline{\mathbf{w}}$, the optimal $\mathbf{\Theta}$ as well as the worst-case $\Delta\mathbf{h}_{\mathrm{RU}}$ and $\mathbf{h}_{\mathrm{BU}}$ should be determined first}.
By substituting (\ref{Opt_w}) into (\ref{Initial_Problem}a) and (\ref{Initial_Problem}b), problem (\ref{Initial_Problem}) is recast {\color{black}as
\begin{subequations}\label{P2-0}
\begin{align}
\mathop{\max}\limits_{\mathbf{\Theta}}\mathop{\min}\limits_{\Delta\mathbf{h}_{\mathrm{RU}},  \mathbf{h}_{\mathrm{BU}}}&
\ \left\|(\widehat{\mathbf{h}}_{\mathrm{RU}}^{\mathrm{LoS}}+\Delta\mathbf{h}_{\mathrm{RU}})^{\mathrm{H}} \mathbf{\Theta}\mathbf{H}_{\mathrm{BR}}+\mathbf{h}_{\mathrm{BU}}^{\mathrm{H}}\right\|_2,
\\&\ \|\Delta \mathbf{h}_{\mathrm{RU}}\|_2\leq \epsilon_{\Delta \mathbf{h}_{\mathrm{RU}}},
\\&\ \|\mathbf{h}_{\mathrm{BU}}\|_2=\delta_{\mathrm{BU}},
\\&\ |[\mathbf{\Theta}]_{(i,i)}|=\beta,\ \ i=1,2,...,N,
\\&\ \arg\left\{[\mathbf{\Theta}]_{(i,i)}\right\} \in \mathcal{S},\ i=1,2,...,N,
\end{align}
\end{subequations}
which is a max-min problem with respect to $\mathbf{\Theta}$, $\Delta\mathbf{h}_{\mathrm{RU}}$ and $\mathbf{h}_{\mathrm{BU}}$. Note that in accordance with the triangle inequality, the objective function in (\ref{P2-0}a) satisfies
\begin{equation}\label{add1}
\begin{split}
&\left\|(\widehat{\mathbf{h}}_{\mathrm{RU}}^{\mathrm{LoS}}+\Delta\mathbf{h}_{\mathrm{RU}})^{\mathrm{H}} \mathbf{\Theta}\mathbf{H}_{\mathrm{BR}}+\mathbf{h}_{\mathrm{BU}}^{\mathrm{H}}\right\|_2\\
\geq &
\left\|(\widehat{\mathbf{h}}_{\mathrm{RU}}^{\mathrm{LoS}}+\Delta\mathbf{h}_{\mathrm{RU}})^{\mathrm{H}} \mathbf{\Theta}\mathbf{H}_{\mathrm{BR}}\right\|_2 - \delta_{\mathrm{BU}},
\end{split}
\end{equation}
where the equality holds when 
\begin{equation}\label{worst-case h_BU}
\mathbf{h}_{\mathrm{BU}}=\overline{\mathbf{h}}_{\mathrm{BU}}=- \frac{\mathbf{H}_{\mathrm{BR}}^{\mathrm{H}} \mathbf{\Theta}^{\mathrm{H}}(\widehat{\mathbf{h}}_{\mathrm{RU}}^{\mathrm{LoS}}+\Delta\mathbf{h}_{\mathrm{RU}}) } { \left\|(\widehat{\mathbf{h}}_{\mathrm{RU}}^{\mathrm{LoS}}+\Delta\mathbf{h}_{\mathrm{RU}})^{\mathrm{H}} \mathbf{\Theta}\mathbf{H}_{\mathrm{BR}}\right\|_2}  \delta_{\mathrm{BU}},
\end{equation}
with $\overline{\mathbf{h}}_{\mathrm{BU}}$ representing the worst-case $\mathbf{h}_{\mathrm{BU}}$ of (\ref{add1}). 
Then, the objective of problem (\ref{P2-0}) is equivalent to $\mathop{\max}\limits_{\mathbf{\Theta}}\mathop{\min}\limits_{\Delta\mathbf{h}_{\mathrm{RU}}}
\ \|(\widehat{\mathbf{h}}_{\mathrm{RU}}^{\mathrm{LoS}}+\Delta\mathbf{h}_{\mathrm{RU}})^{\mathrm{H}} \mathbf{\Theta}\mathbf{H}_{\mathrm{BR}}\|_2$ or $\mathop{\max}\limits_{\mathbf{\Theta}}\mathop{\min}\limits_{\Delta\mathbf{h}_{\mathrm{RU}}}
\ \|(\widehat{\mathbf{h}}_{\mathrm{RU}}^{\mathrm{LoS}}+\Delta\mathbf{h}_{\mathrm{RU}})^{\mathrm{H}} \mathbf{\Theta}\mathbf{H}_{\mathrm{BR}}\|_2^2$. 
As a result, problem (\ref{P2-0}) is transformed into}
\begin{subequations}\label{P2}
\begin{align}
\mathop{\max}\limits_{\mathbf{\Theta}}\mathop{\min}\limits_{\Delta\mathbf{h}_{\mathrm{RU}}}&
\ \|(\widehat{\mathbf{h}}_{\mathrm{RU}}^{\mathrm{LoS}}+\Delta\mathbf{h}_{\mathrm{RU}})^\mathrm{H} \mathbf{\Theta}\mathbf{H}_{\mathrm{BR}}\|_2^2,
\\ \mathrm{s.t.} &\ \|\Delta \mathbf{h}_{\mathrm{RU}}\|_2\leq \epsilon_{\Delta \mathbf{h}_{\mathrm{RU}}},
\\&\ |[\mathbf{\Theta}]_{(i,i)}|=\beta,\ \ i=1,2,...,N,
\\&\ \arg\left\{[\mathbf{\Theta}]_{(i,i)}\right\} \in \mathcal{S},\ i=1,2,...,N,
\end{align}
\end{subequations}
which is now a max-min problem with respect to $\mathbf{\Theta}$ and $\Delta\mathbf{h}_{\mathrm{RU}}$.

To solve a max-min problem, existing methods such as S-procedure \cite{S-Procedure2009} or saddle point theory \cite{Saddle2013} based techniques, can be applied to eliminate the disturbance caused by the nondeterminacy of the CSI error, or to find the worst-case CSI from an equivalent min-max problem. However, in view of 1) the existence of the {\color{black}constant-modulus} constraint in (\ref{P2}c) and the argument constraint in (\ref{P2}d), and 2) the difficulty in determining the closed-form optimal $\mathbf{\Theta}$, using these conventional methods to solve problem (\ref{P2}) is a hard nut to crack. Therefore, in this section, we develop an iterative optimization approach to solve problem (\ref{P2}), and design the overall optimization algorithm.

\subsection{Proposed Iterative Optimization Approach}

{\color{black}This subsection focuses on finding the solution of problem (\ref{P2}), by firstly solving the inner minimization problem and then solving the outer maximization problem. In particular, unlike the existing iterative algorithms in e.g. \cite{G.Zhou-TSP2020} or \cite{G.Zhou-WCL2020} which provide approximate solutions of the subproblems in each iteration, here we propose a novel alternating optimization process through the instrumentality of Lagrange multiplier and matrix inverse lemma. In this process, the optimal solutions of the Lagrange dual variable in the inner minimization and the RIS phase shifts in the outer maximization are alternatively obtained, via the bisection search and the SDR/BnB, respectively. Results will finally demonstrate that the proposed approach can outperform the ones in \cite{G.Zhou-TSP2020} or \cite{G.Zhou-WCL2020}, in terms of both the eventual optimization result and the overall convergence rate.}

$\ $

\textbf{\textit{1) Solve the Inner Minimization Problem:}}

To solve problem (\ref{P2}), let us first consider the following inner minimization problem when treating $\mathbf{\Theta}$ as a fixed term:
\begin{subequations}\label{P3}
\begin{align}
\mathop{\min}\limits_{\Delta\mathbf{h}_{\mathrm{RU}}}&
\ \|(\widehat{\mathbf{h}}_{\mathrm{RU}}^{\mathrm{LoS}}+\Delta\mathbf{h}_{\mathrm{RU}})^{\mathrm{H}} \mathbf{\Theta}\mathbf{H}_{\mathrm{BR}}\|_2^2,
\\ \mathrm{s.t.} &\ \|\Delta \mathbf{h}_{\mathrm{RU}}\|_2\leq \epsilon_{\Delta \mathbf{h}_{\mathrm{RU}}}.
\end{align}
\end{subequations}
Here, we can rewrite constraint (\ref{P3}b) as
\begin{equation}
\Delta \mathbf{h}_{\mathrm{RU}}^{\mathrm{H}} \Delta \mathbf{h}_{\mathrm{RU}} - \epsilon_{\Delta \mathbf{h}_{\mathrm{RU}}}^2 \leq 0,
\end{equation}
so that problem (\ref{P3}) is a standard QCQP problem with respect to $\Delta \mathbf{h}_{\mathrm{RU}}$. To acquire the optimal solution of $\Delta \mathbf{h}_{\mathrm{RU}}$ {\color{black}(e.g. the worst-case $\Delta \mathbf{h}_{\mathrm{RU}}$)}, we construct the Lagrange function:
\begin{equation}\label{Lagrange}
\begin{split}
\mathcal{L}(\Delta \mathbf{h}_{\mathrm{RU}},\mu)
=\ &\|(\widehat{\mathbf{h}}_{\mathrm{RU}}^{\mathrm{LoS}}+\Delta\mathbf{h}_{\mathrm{RU}})^{\mathrm{H}} \mathbf{\Theta}\mathbf{H}_{\mathrm{BR}}\|_2^2 \\
&+ \mu\left(\Delta \mathbf{h}_{\mathrm{RU}}^{\mathrm{H}} \Delta \mathbf{h}_{\mathrm{RU}} - \epsilon_{\Delta \mathbf{h}_{\mathrm{RU}}}^2\right)\\
=\ &\Delta \mathbf{h}_{\mathrm{RU}}^{\mathrm{H}} \left(\mathbf{\Theta}\mathbf{H}_{\mathrm{BR}}\mathbf{H}_{\mathrm{BR}}^{\mathrm{H}}\mathbf{\Theta}^{\mathrm{H}}+\mu\mathbf{I}\right)\Delta \mathbf{h}_{\mathrm{RU}}\\
&+2\mathfrak{Re}\left\{(\widehat{\mathbf{h}}_{\mathrm{RU}}^{\mathrm{LoS}})^{\mathrm{H}}\mathbf{\Theta}\mathbf{H}_{\mathrm{BR}}\mathbf{H}_{\mathrm{BR}}^{\mathrm{H}}\mathbf{\Theta}^{\mathrm{H}} \Delta \mathbf{h}_{\mathrm{RU}}\right\}\\
&+(\widehat{\mathbf{h}}_{\mathrm{RU}}^{\mathrm{LoS}})^{\mathrm{H}}\mathbf{\Theta}\mathbf{H}_{\mathrm{BR}}\mathbf{H}_{\mathrm{BR}}^{\mathrm{H}}\mathbf{\Theta}^{\mathrm{H}} \widehat{\mathbf{h}}_{\mathrm{RU}}^{\mathrm{LoS}}
- \mu \epsilon_{\Delta \mathbf{h}_{\mathrm{RU}}}^2,
\end{split}
\end{equation}
where $\mu$ is the Lagrange dual variable.

The optimal solutions of $\mu$ and $\Delta \mathbf{h}_{\mathrm{RU}}$ should satisfy the Karush-Kuhn-Tucker (KKT) conditions, listed as
\begin{subequations}\label{KKT}
\begin{align}
\frac{\partial\mathcal{L}(\Delta \mathbf{h}_{\mathrm{RU}},\mu)}{\partial\Delta \mathbf{h}_{\mathrm{RU}}}&=\bm{0},\ \ \ \ \\
\Delta \mathbf{h}_{\mathrm{RU}}^{\mathrm{H}} \Delta \mathbf{h}_{\mathrm{RU}} - \epsilon_{\Delta \mathbf{h}_{\mathrm{RU}}}^2 &\leq 0,\\
\mu &\geq 0,\\
\mu\left(\Delta \mathbf{h}_{\mathrm{RU}}^{\mathrm{H}} \Delta \mathbf{h}_{\mathrm{RU}} - \epsilon_{\Delta \mathbf{h}_{\mathrm{RU}}}^2\right) &= 0.
\end{align}
\end{subequations}

By calculating $\frac{\partial\mathcal{L}(\Delta \mathbf{h}_{\mathrm{RU}},\mu)}{\partial\Delta \mathbf{h}_{\mathrm{RU}}}=\bm{0}$, we obtain
\begin{equation}
\begin{split}
&\left(\mathbf{\Theta}\mathbf{H}_{\mathrm{BR}}\mathbf{H}_{\mathrm{BR}}^{\mathrm{H}}\mathbf{\Theta}^{\mathrm{H}}+\mu\mathbf{I}\right)^{\mathrm{T}} \Delta \mathbf{h}_{\mathrm{RU}}^* \\
&\ \ \ \ \ \ \ \ +\left(\mathbf{\Theta}\mathbf{H}_{\mathrm{BR}}\mathbf{H}_{\mathrm{BR}}^{\mathrm{H}}\mathbf{\Theta}^{\mathrm{H}} \widehat{\mathbf{h}}_{\mathrm{RU}}^{\mathrm{LoS}}\right)^*=\bm{0},
\end{split}
\end{equation}
from which we derive the closed-form worst-case $\Delta \mathbf{h}_{\mathrm{RU}}$ with respect to $\mu$ and $\mathbf{\Theta}$, denoted by $\overline{\Delta \mathbf{h}}_{\mathrm{RU}}(\mu,\mathbf{\Theta})$, as
\begin{equation}\label{Delta_h_Opt}
\begin{split}
\overline{\Delta \mathbf{h}}_{\mathrm{RU}}(\mu,\mathbf{\Theta})=-&\left(\mathbf{\Theta}\mathbf{H}_{\mathrm{BR}}\mathbf{H}_{\mathrm{BR}}^{\mathrm{H}}\mathbf{\Theta}^{\mathrm{H}}+\mu\mathbf{I}\right)^{-1}\\
&\times
\mathbf{\Theta}\mathbf{H}_{\mathrm{BR}}\mathbf{H}_{\mathrm{BR}}^{\mathrm{H}}\mathbf{\Theta}^{\mathrm{H}} \widehat{\mathbf{h}}_{\mathrm{RU}}^{\mathrm{LoS}}.
\end{split}
\end{equation}

According to (\ref{KKT}d), the optimal solutions of $\Delta \mathbf{h}_{\mathrm{RU}}$ and $\mu$ should either satisfy $\mu=0$ or $\Delta \mathbf{h}_{\mathrm{RU}}^{\mathrm{H}} \Delta \mathbf{h}_{\mathrm{RU}} - \epsilon_{\Delta \mathbf{h}_{\mathrm{RU}}}^2 = 0$. If $\mu=0$, $\overline{\Delta \mathbf{h}}_{\mathrm{RU}}(\mu,\mathbf{\Theta})$ in (\ref{Delta_h_Opt}) degenerates into $\overline{\Delta \mathbf{h}}_{\mathrm{RU}}(\mu,\mathbf{\Theta})=\overline{\Delta \mathbf{h}}_{\mathrm{RU}}(0,\mathbf{\Theta})=-\widehat{\mathbf{h}}_{\mathrm{RU}}^{\mathrm{LoS}}$, which makes the objective function in (\ref{P3}a) reduce to zero. In this case, the optimization process fails to find {\color{black}the} optimal $\mathbf{\Theta}$. Therefore, the optimization process is feasible only if $\mu>0$ and $\Delta \mathbf{h}_{\mathrm{RU}}^{\mathrm{H}} \Delta \mathbf{h}_{\mathrm{RU}} - \epsilon_{\Delta \mathbf{h}_{\mathrm{RU}}}^2 = 0$. 

{\color{black}By substituting $\overline{\Delta \mathbf{h}}_{\mathrm{RU}}(\mu,\mathbf{\Theta})$ into $\Delta \mathbf{h}_{\mathrm{RU}}^{\mathrm{H}} \Delta \mathbf{h}_{\mathrm{RU}} = \epsilon_{\Delta \mathbf{h}_{\mathrm{RU}}}^2$ and $\|(\widehat{\mathbf{h}}_{\mathrm{RU}}^{\mathrm{LoS}}+\Delta\mathbf{h}_{\mathrm{RU}})^{\mathrm{H}} \mathbf{\Theta}\mathbf{H}_{\mathrm{BR}}\|_2^2$, respectively, we obtain (\ref{bisection}) and (\ref{mathcal_F}) on the top of the next page. 
\begin{figure*}[!t]
\normalsize
\begin{scriptsize}
\begin{equation}\label{bisection}
\begin{split}
\mathcal{C}(\mu,\mathbf{\Theta})=
(\widehat{\mathbf{h}}_{\mathrm{RU}}^{\mathrm{LoS}})^{\mathrm{H}}\mathbf{\Theta}\mathbf{H}_{\mathrm{BR}}\mathbf{H}_{\mathrm{BR}}^{\mathrm{H}}\mathbf{\Theta}^{\mathrm{H}}
\left(\mathbf{\Theta}\mathbf{H}_{\mathrm{BR}}\mathbf{H}_{\mathrm{BR}}^{\mathrm{H}}\mathbf{\Theta}^{\mathrm{H}}+\mu\mathbf{I}\right)^{-2}
\mathbf{\Theta}\mathbf{H}_{\mathrm{BR}}\mathbf{H}_{\mathrm{BR}}^{\mathrm{H}}\mathbf{\Theta}^{\mathrm{H}} \widehat{\mathbf{h}}_{\mathrm{RU}}^{\mathrm{LoS}}
-\epsilon_{\Delta \mathbf{h}_{\mathrm{RU}}}^2
=0.
\end{split}
\end{equation}
\end{scriptsize}
\vspace*{-15pt}
\end{figure*}
\begin{figure*}[!t]
\normalsize
\begin{scriptsize}
\begin{equation}\label{mathcal_F}
\mathcal{F}(\mu,\mathbf{\Theta})=
\left\|
\left[\widehat{\mathbf{h}}_{\mathrm{RU}}^{\mathrm{LoS}}-\left(\mathbf{\Theta}\mathbf{H}_{\mathrm{BR}}\mathbf{H}_{\mathrm{BR}}^{\mathrm{H}}\mathbf{\Theta}^{\mathrm{H}}+\mu\mathbf{I}\right)^{-1}\mathbf{\Theta}\mathbf{H}_{\mathrm{BR}}\mathbf{H}_{\mathrm{BR}}^{\mathrm{H}}\mathbf{\Theta}^{\mathrm{H}} \widehat{\mathbf{h}}_{\mathrm{RU}}^{\mathrm{LoS}}\right]^{\mathrm{H}} \mathbf{\Theta}\mathbf{H}_{\mathrm{BR}}
\right\|_2^2.
\end{equation}
\end{scriptsize}
\hrulefill
\vspace*{-15pt}
\end{figure*}
It is noted that if (\ref{bisection}) is satisfied, (\ref{mathcal_F}) is the lower bound of $\|(\widehat{\mathbf{h}}_{\mathrm{RU}}^{\mathrm{LoS}}+\Delta\mathbf{h}_{\mathrm{RU}})^{\mathrm{H}} \mathbf{\Theta}\mathbf{H}_{\mathrm{BR}}\|_2^2$ for any $\|\Delta \mathbf{h}_{\mathrm{RU}}\|_2\leq \epsilon_{\Delta \mathbf{h}_{\mathrm{RU}}}$, and the minimum objective value of the inner minimization problem is achieved. Therefore, based on (\ref{bisection}) and (\ref{mathcal_F}), problem (\ref{P2}) is transformed into
\begin{subequations}\label{P4}
\begin{align}
\mathop{\max}\limits_{\mu,\mathbf{\Theta}}&
\ \mathcal{F}(\mu,\mathbf{\Theta}),
\\ \mathrm{s.t.} &\ \mathcal{C}(\mu,\mathbf{\Theta})= 0,
\\&\ |[\mathbf{\Theta}]_{(i,i)}|=\beta,\ \ i=1,2,...,N,
\\&\ \arg\left\{[\mathbf{\Theta}]_{(i,i)}\right\} \in \mathcal{S},\ i=1,2,...,N,
\end{align}
\end{subequations}
where $\mathcal{C}(\mu,\mathbf{\Theta})$ and $\mathcal{F}(\mu,\mathbf{\Theta})$ are specified in (\ref{bisection}) and (\ref{mathcal_F})}.

$\ $

\textbf{\textit{2) Solve the Outer Maximization Problem (\ref{P4}):}}

Subsequently, we focus on solving problem (\ref{P4}). It is {\color{black}remarkable} that
problem (\ref{P4}) is a complicated non-convex problem containing the inverse of $\mathbf{\Theta}\mathbf{H}_{\mathrm{BR}}\mathbf{H}_{\mathrm{BR}}^{\mathrm{H}}\mathbf{\Theta}^{\mathrm{H}}+\mu\mathbf{I}$, which makes problem (\ref{P4}) difficult to be solved efficiently. To deal with $\left(\mathbf{\Theta}\mathbf{H}_{\mathrm{BR}}\mathbf{H}_{\mathrm{BR}}^{\mathrm{H}}\mathbf{\Theta}^{\mathrm{H}}+\mu\mathbf{I}\right)^{-1}$, one feasible choice is to define a new variable as $\mathbf{Q}=\left(\mathbf{\Theta}\mathbf{H}_{\mathrm{BR}}\mathbf{H}_{\mathrm{BR}}^{\mathrm{H}}\mathbf{\Theta}^{\mathrm{H}}+\mu\mathbf{I}\right)^{-1}$, and optimize $\mathbf{\Theta}$ and $\mathbf{Q}$ iteratively until the optimization process converges \cite{Rui Wang}. However, this choice has one major drawback, i.e. additional iterations are included in the optimization process, leading to an increase of the overall computational complexity. In view of this defect, we focus on pursuing another effective way to simplify the objective function and the constraints by adopting matrix inversion lemma in the following Proposition 1.

\begin{proposition}
With the aid of the matrix inversion lemma, problem (\ref{P4}) can be transformed {\color{black}into 
\begin{subequations}\label{P5}
\begin{align}
\mathop{\max}\limits_{\mu,\bm{\theta}}&
\ \bm{\theta}^{\mathrm{T}} \mathbf{\Upsilon}(\mu) \bm{\theta}^*,
\\ \mathrm{s.t.} &\ \bm{\theta}^{\mathrm{T}}\mathbf{\Gamma}(\mu)\bm{\theta}^*=\epsilon_{\Delta \mathbf{h}_{\mathrm{RU}}}^2,
\\&\ |[\bm{\theta}]_{i}|=\beta,\ \ i=1,2,...,N,
\\&\ \arg\left\{[\bm{\theta}]_{i}\right\} \in \mathcal{S},\ i=1,2,...,N,
\end{align}
\end{subequations}
where $\bm{\theta}=\beta\left(e^{j\theta_1},e^{j\theta_2},...,e^{j\theta_N}\right)^{\mathrm{T}}$ is a column vector containing the diagonal elements of $\mathbf{\Theta}$; $\mathbf{\Upsilon}(\mu)$ and $\mathbf{\Gamma}(\mu)$, detailed in (\ref{mathbf_Y}) and (\ref{Gamma_mu}) in Appendix B, are positive semidefinite matrices with respect to $\mu$ and are independent of $\bm{\theta}$}.
\end{proposition}

\begin{proof}
The proof is given in Appendix B.
\end{proof}

{\color{black}Problem (\ref{P5}) is still non-convex and difficult to be solved directly, due to the coupling of $\mu$ and $\bm{\theta}$. In order to decouple the two variables and solve (\ref{P5}) efficiently, here we propose a relaxed alternating optimization process (RAOP) to optimize $\mu$ and $\bm{\theta}$ iteratively. In specific, with a given $\bm{\theta}$, we find the optimal $\mu$ by solving the following problem:
\begin{subequations}\label{Optimize_mu}
\begin{align}
\mathop{\mathrm{find}}\limits_{\mu}&
\ \mu,
\\ \mathrm{s.t.} &\ \bm{\theta}^{\mathrm{T}}\mathbf{\Gamma}(\mu)\bm{\theta}^*=\epsilon_{\Delta \mathbf{h}_{\mathrm{RU}}}^2,
\end{align}
\end{subequations}
which can be solved by the well-known bisection search \cite{Bisection Ref}.
After the solution of $\mu$ is obtained, we optimize $\bm{\theta}$ by solving the following relaxed problem:
\begin{subequations}\label{Optimize_theta}
\begin{align}
\mathop{\max}\limits_{\bm{\theta}}&
\ \bm{\theta}^{\mathrm{T}} \mathbf{\Upsilon}(\mu) \bm{\theta}^*,
\\ \mathrm{s.t.} &\ \bm{\theta}^{\mathrm{T}}\mathbf{\Gamma}(\mu)\bm{\theta}^*\geq\epsilon_{\Delta \mathbf{h}_{\mathrm{RU}}}^2,
\\&\ |[\bm{\theta}]_{i}|=\beta,\ \ i=1,2,...,N,
\\&\ \arg\left\{[\bm{\theta}]_{i}\right\} \in \mathcal{S},\ i=1,2,...,N.
\end{align}
\end{subequations}
The above procedure repeats until some convergence criteria (e.g. the difference between the objective values of two adjacent iterations becomes smaller than a threshold) are met.
In (\ref{Optimize_theta}), the constraint of $\bm{\theta}^{\mathrm{T}}\mathbf{\Gamma}(\mu)\bm{\theta}^*=\epsilon_{\Delta \mathbf{h}_{\mathrm{RU}}}^2$ is relaxed into $\bm{\theta}^{\mathrm{T}}\mathbf{\Gamma}(\mu)\bm{\theta}^*\geq\epsilon_{\Delta \mathbf{h}_{\mathrm{RU}}}^2$, in order to guarantee the convergence of the alternating optimization, by ensuring that the objective value of $\bm{\theta}^{\mathrm{T}} \mathbf{\Upsilon}(\mu) \bm{\theta}^*$ is monotonically increasing during the iteration process (For more details, please refer to the convergence analysis in Section IV-C). The relaxation will become tight when the alternating optimization converges, owing to the equality of (44b).

Problem (\ref{Optimize_theta})} can be solved by various existing techniques. Here, we briefly introduce two state-of-the-art algorithms for solving problem (\ref{Optimize_theta}), which are semidefinite relaxation (SDR) algorithm and branch-and-bound (BnB) algorithm. 
\footnote{\color{black} Problem (\ref{Optimize_theta}) is a non-convex constant-modulus and argument constrained QCQP problem, which is NP-hard in general. To solve such a problem, this paper introduces SDR since it has been validated to be a powerful, computationally efficient technique that can transform non-convex objectives/constraints into convex trace forms, and can provide accurate or sometimes near-optimal solutions \cite{Z.-Q. Luo-SPM2010}.
It is noted that the SDR simply drops the argument constraint, so that it only performs well when $\mathcal{S}=[0,2\pi]$. Therefore, this paper also introduces the recently developed BnB to handle other argument sets and solve problem (\ref{Optimize_theta}) globally \cite{BnB1}.}
 In specific, the BnB is slightly modified to deal with our problem. 

1) \textbf{SDR algorithm}: When $\mathcal{S}=[0,2\pi]$, problem (\ref{Optimize_theta}) can be solved by the SDR algorithm, which has been widely utilized in wireless communication and signal processing fields to transform non-convex problems into semidefinite programming (SDP) problems. Considering problem (\ref{Optimize_theta}) as an example, the main idea of the SDR is to introduce a new variable $\mathbf{C}=\bm{\theta}^*\bm{\theta}^{\mathrm{T}}$, and transform problem (\ref{Optimize_theta}) by dropping the rank-one constraint {\color{black}into
\begin{subequations}\label{P6-SDR}
\begin{align}
\mathop{\max}\limits_{\mathbf{C}\succeq\mathbf{0}}&
\ tr\left(\mathbf{\Upsilon}(\mu)\mathbf{C}\right),
\\ \mathrm{s.t.} &\ tr\left(\mathbf{\Gamma}(\mu)\mathbf{C}\right)\geq\epsilon_{\Delta \mathbf{h}_{\mathrm{RU}}}^2,
\\&\ tr(\mathbf{E}_{i}\mathbf{C})=\beta^2,\ \ i=1,2,...,N,
\end{align}
\end{subequations}
where} the matrix $\mathbf{E}_{i}$ satisfies
\begin{equation}
\left[\mathbf{E}_i\right]_{(m,n)}=\left\{\begin{matrix}1,\ \ \ m=n=i;\\0,\ \ \ \ \mathrm{otherwise.}\\\end{matrix}\right.
\end{equation}

{\color{black}Problem (\ref{P6-SDR}) can be solved by existing methods such as interior-point method. To address the omitted rank-one constraint after solving problem (\ref{P6-SDR}) and obtaining the solution of $\mathbf{C}$, denoted by $\overline{\mathbf{C}}$, one can perform eigenvalue decomposition for $\overline{\mathbf{C}}$ to acquire the optimal $\bm{\theta}$ if $\overline{\mathbf{C}}$ is rank-one \cite{Z.-Q. Luo-SPM2010}, or use some convex relaxation techniques for phase-only beamforming, e.g., the SDP concave-convex procedure \cite{rank-one solution}, to iteratively acquire an approximate rank-one solution if the SDR for problem (\ref{P6-SDR}) is not intrinsically guaranteed to be tight.}
{\color{black}Fortunately, it can be proved that the SDR for problem (\ref{P6-SDR}) is sometimes tight. For instance, regarding problem (\ref{P6-SDR}), here we theoretically derive two specific conditions making $rank(\overline{\mathbf{C}})=1$ hold in the following Proposition 2:
\begin{proposition}
$rank(\overline{\mathbf{C}})=1$ holds if the conditions below are concurrently satisfied:

Condition 1: $\mathbf{\Upsilon}(\mu)$ is rank-one, and each element in its eigenvector is non-zero.

Condition 2: The solution $\overline{\mathbf{C}}$ satisfies $tr\left(\mathbf{\Gamma}(\mu)\overline{\mathbf{C}}\right)>\epsilon_{\Delta \mathbf{h}_{\mathrm{RU}}}^2$.
\end{proposition}
\begin{proof}
The proof is given in Appendix C.
\end{proof}

\begin{remark}
For Condition 1 in Proposition 2, according to (\ref{mathbf_Y}) and (\ref{Z_mu}), if $rank(\mathbf{H}_{\mathrm{BR}})=1$, we obtain $rank(\mathbf{\Upsilon}(\mu))=1$. Since the elements in $\mathbf{H}_{\mathrm{BR}}$ are non-zero, the elements in the eigenvector of $\mathbf{\Upsilon}(\mu)$ are non-zero as well according to (87) in Appendix D. Therefore, Condition 1 can hold when $rank(\mathbf{H}_{\mathrm{BR}})=1$, i.e. $\mathbf{H}_{\mathrm{BR}}$ is a far-field LoS channel \cite{Our-Previous}. For Condition 2 in Proposition 2, let $\overline{\mathbf{C}}_k$ and $\overline{\mathbf{C}}_{k+1}$ denote the solutions of $\mathbf{C}$ in the $k$-th iteration and the $(k+1)$-th iteration, respectively (similar definitions are given for $\mu_k$ and $\mu_{k+1}$). If the iteration process is not terminated in the $(k+1)$-th iteration, we have $tr\left(\mathbf{\Upsilon}(\mu_{k+1})\overline{\mathbf{C}}_{k}\right)\neq tr\left(\mathbf{\Upsilon}(\mu_{k+1})\overline{\mathbf{C}}_{k+1}\right)$ and $tr\left(\mathbf{\Gamma}(\mu_{k+1})\overline{\mathbf{C}}_{k}\right)\neq tr\left(\mathbf{\Gamma}(\mu_{k+1})\overline{\mathbf{C}}_{k+1}\right)$ according to the structures in (\ref{mathbf_Y}) and (\ref{Gamma_mu}). Moreover, because $tr\left(\mathbf{\Gamma}(\mu_{k+1})\overline{\mathbf{C}}_{k}\right)=\epsilon_{\Delta \mathbf{h}_{\mathrm{RU}}}^2$ holds in accordance with problem (\ref{Optimize_mu}), we obtain $tr\left(\mathbf{\Gamma}(\mu_{k+1})\overline{\mathbf{C}}_{k+1}\right)\neq \epsilon_{\Delta \mathbf{h}_{\mathrm{RU}}}^2$. Finally, because $\overline{\mathbf{C}}_{k+1}$ should be in the feasible region of (\ref{P6-SDR}b), we obtain $tr\left(\mathbf{\Gamma}(\mu_{k+1})\overline{\mathbf{C}}_{k+1}\right)> \epsilon_{\Delta \mathbf{h}_{\mathrm{RU}}}^2$. Therefore, Condition 2 can hold if the iteration process is not terminated.
\end{remark}

Although $rank(\overline{\mathbf{C}})=1$ can possibly hold in some cases, a non-zero gap between $\overline{\mathbf{C}}$ and the true optimal $\mathbf{C}$ generally exists when $rank(\mathbf{\Upsilon}(\mu))\neq 1$}. In addition, if $\mathcal{S}$ is an arbitrary argument set of $\mathcal{S}=[\ell_l,\ell_u]$ instead of $\mathcal{S}=[0,2\pi]$, the SDR fails to solve problem (\ref{Optimize_theta}).
Forasmuch as these drawbacks of SDR, the BnB has been proposed, with the details stated below.

2) \textbf{BnB algorithm}: The BnB algorithm was proposed in \cite{BnB1,BnB2} to find the global optimal solution of complex quadratic programming problems, by branching on the argument sets. Compared with the SDR, the BnB has two major advantages. First, the BnB can deal with arbitrary argument sets of $\mathcal{S}=[\ell_l,\ell_u]$ with $\ell_u-\ell_l\leq \pi$, such as $\mathcal{S}=[0,\pi]$, $\mathcal{S}=[0,\frac{\pi}{2}]$, or discrete argument values, whereas the SDR is only able to handle $\mathcal{S}=[0,2\pi]$. Second, the solution of BnB can be very close to the true optimum after sufficient times of iterations. 

Since the BnB requires $\ell_u-\ell_l\leq \pi$, when dealing with $\mathcal{S}=[\ell_l,\ell_u]$ in which $\ell_u-\ell_l> \pi$, we first divide $\mathcal{S}$ into two subsets of $\mathcal{S}_1=[\ell_l,\ell_m]$ and $\mathcal{S}_2=[\ell_m,\ell_u]$, where $\mathcal{S}_1+\mathcal{S}_2=\mathcal{S}$ and $\ell_m-\ell_l=\pi$ \footnote{For example, if $\mathcal{S}=[\ell_l,\ell_u]=[0,2\pi]$, we divide $\mathcal{S}$ into $\mathcal{S}_1=[0,\pi]$ and $\mathcal{S}_2=[\pi,2\pi]$, with $\ell_l=0$, $\ell_m=\pi$ and $\ell_u=2\pi$.}.
For each subset, we construct a convex envelope, denoted by $\mathcal{E}_i=\{[\bm{\theta}]_i|\mathfrak{Re}\{a_i^*[\bm{\theta}]_i\}\geq \cos{\frac{\ell_{\mathrm{right}}-\ell_{\mathrm{left}}}{2}}\}$, where $a_i=\cos{\frac{\ell_{\mathrm{right}}+\ell_{\mathrm{left}}}{2}}+j\sin{\frac{\ell_{\mathrm{right}}+\ell_{\mathrm{left}}}{2}}$; $\ell_{\mathrm{left}}$ and $\ell_{\mathrm{right}}$ satisfy $\ell_{\mathrm{left}}=\ell_l$ and $\ell_{\mathrm{right}}=\ell_m$ for subset $\mathcal{S}_1$, or satisfy $\ell_{\mathrm{left}}=\ell_m$ and $\ell_{\mathrm{right}}=\ell_u$ for subset $\mathcal{S}_2$. For more visual details, one can refer to [Fig. 1, 45] which depicts the convex envelope. Then, based on $\mathcal{E}_i$, we reformulate problem (\ref{Optimize_theta}) {\color{black} as
\begin{subequations}\label{P7-BnB}
\begin{align}
\mathop{\max}\limits_{\mathbf{C}\succeq\mathbf{0},\bm{\theta}}&
\ tr\left(\mathbf{\Upsilon}(\mu)\mathbf{C}\right),
\\ \mathrm{s.t.} &\ (\ref{P6-SDR}\mathrm{b}),(\ref{P6-SDR}\mathrm{c}),
\\&\ [\bm{\theta}]_i\in \mathcal{E}_i, \ \ i=1,2,...,N,
\\&\ \mathbf{C}\succeq\bm{\theta}^*\bm{\theta}^{\mathrm{T}}.
\end{align}
\end{subequations}
which} is convex. Afterwards, considering problem (\ref{P7-BnB}), we adopt [Algorithm 1, 45] to find the optimal solutions for the two subsets of $\mathcal{S}_1$ and $\mathcal{S}_2$ separately, which are represented by $\overline{\bm{\theta}}_{\mathcal{S}_1}$ and $\overline{\bm{\theta}}_{\mathcal{S}_2}$. The main idea of [Algorithm 1, 45] is to continuously execute the following procedure:

1) Cut $\mathcal{S}_1$ or $\mathcal{S}_2$ into smaller sets and solve problem (\ref{P7-BnB}).

2) Use the rounding operation to acquire a projection of the solution of $\bm{\theta}$ on the feasible domain of $|[\bm{\theta}]_{i}|=\beta$.

The procedure stops when the gap between the objective values with respect to the projection and the solution of problem (\ref{P7-BnB}) is smaller than a predetermined error tolerance $\epsilon_{BnB}$.
After obtaining $\overline{\bm{\theta}}_{\mathcal{S}_1}$ and $\overline{\bm{\theta}}_{\mathcal{S}_2}$, we substitute them into (\ref{Optimize_theta}a) to calculate the objective values. Then, from $\overline{\bm{\theta}}_{\mathcal{S}_1}$ and $\overline{\bm{\theta}}_{\mathcal{S}_2}$, we choose the one corresponding to the higher objective value as the final solution. 
The flowchart of the modified BnB in our algorithm is described in Fig. \ref{BnB_Flowchart}, which shows the above process.

Moreover, when dealing with the discrete phase shift argument sets, the BnB algorithm in \cite{BnB1} can be readily adopted by constructing a polyhedral convex hull. One can refer to \cite{BnB1} for more details.

\begin{figure}[!t]
\includegraphics[width=2.9in]{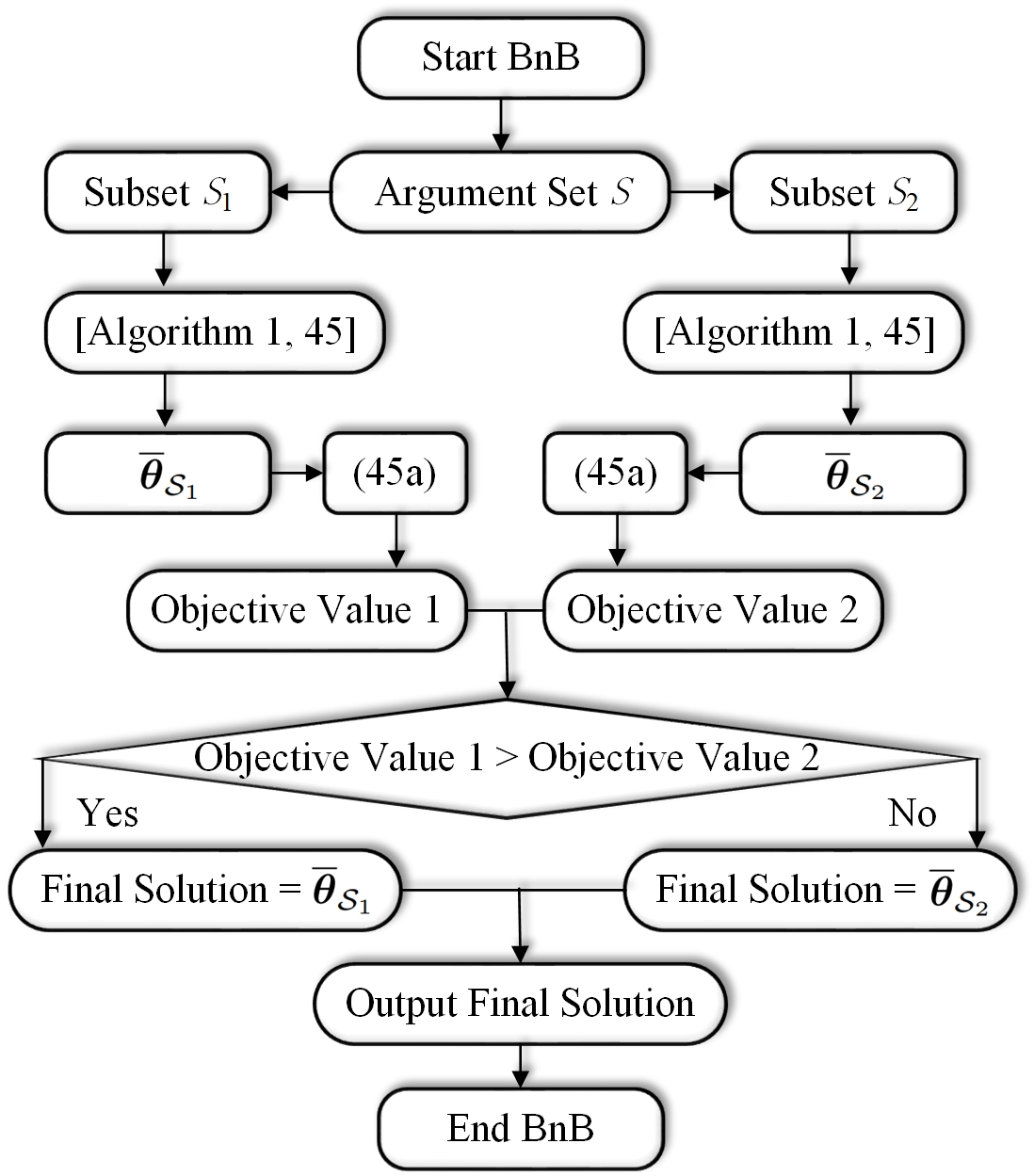}
\hfil
\centering
\caption{\color{black}The flowchart of the modified BnB in our algorithm. }
\label{BnB_Flowchart}
\end{figure}


Based on the analysis in this subsection, the overall robust beamforming optimization algorithm is designed as follows.

\subsection{Overall Algorithm Design}

Our robust beamforming optimization algorithm is designed as \textbf{Algorithm 1}, which can be summarized as 4 steps.

\begin{itemize}
\item[•] \textbf{Step 1: Input parameters:} the parameters including locations, number of antennas and reflecting elements, etc., are input to the algorithm.

\item[•] \textbf{Step 2: Compute CSI error bound:} the CSI error bound $\epsilon_{\Delta \mathbf{h}_{\mathrm{RU}}}$ is calculated using (\ref{Finally_Derived_CSIerrorbound}).

\item[•] \textbf{Step 3: Robust beamforming optimization:} first, the phase-shift matrix, total iteration times, iteration index, terminating condition and objective value are initialized. Then, problem (\ref{Optimize_mu}) and problem (\ref{Optimize_theta}) are iteratively solved by bisection method and SDR/BnB to optimize $\mu$ and $\mathbf{\Theta}$, until the optimization procedure converges.

\item[•] \textbf{Step 4: Output solutions:} the optimal Lagrange dual variable, worst-case CSI error, {\color{black}worst-case BS-UE channel}, and optimal transmit and passive beamforming are obtained and output.
\end{itemize}

\begin{algorithm}
\begin{small}
\caption{\begin{small} Proposed Overall Optimization Algorithm.\end{small}}
\LinesNumbered
{\bf Input:} $\mathbf{q}_1$, $\mathbf{v}_1$, $\widehat{\mathbf{p}}$, $d_{BS}$, $d_{RIS}$, $\lambda$, $M$, $L$, $N$, $\epsilon_{\Delta\mathbf{p}}$\;
Compute $\mathbf{H}_{\mathrm{BR}}$ and $\widehat{\mathbf{h}}_{\mathrm{RU}}^{\mathrm{LoS}}$ according to (\ref{H_B-R}) and (\ref{h_R-M})\;

$\%$ \textit{Begin computation for the CSI error bound}\;
Compute CSI error bound $\epsilon_{\Delta \mathbf{h}_{\mathrm{RU}}}$ using (\ref{Finally_Derived_CSIerrorbound})\;

$\%$ \textit{Begin the proposed RAOP}\;
Initialize: $\theta_1=\theta_2=\dots=\theta_N=0$, total iteration times $T$, iteration index $t=1$, terminating condition $\epsilon_{R}$, initial objective value $R_{Record}^{obj}=0$\;
  \While{$t\leq T$}{
  Record the $t$-th $\mathbf{\Theta}$: $\mathbf{\Theta}_{Record} \leftarrow \mathbf{\Theta}$\;
  Solve (\ref{Optimize_mu}) by using bisection method with given $\mathbf{\Theta}$, and then obtain $\mu$\;
  Compute {\color{black}$\mathbf{\Upsilon}(\mu)$ and $\mathbf{\Gamma}(\mu)$}\;
  Solve problem (\ref{Optimize_theta}) using BnB/SDR method, obtain $\mathbf{\Theta}$ and the value of the objective function $R^{obj}$\;
  \If{$|R^{obj}-R^{obj}_{Record}| \leq \epsilon_{R}$,}{
  \textbf{break}\;
  }
  $t\leftarrow t+1$\;
  Record the $t$-th $R^{obj}$: $R^{obj}_{Record} \leftarrow R^{obj}$\;
  }
 { \color{black}$\%$ \textit{Output the optimization results}\;}
  Obtain the optimal $\mu$: $\overline{\mu} \leftarrow \mu$\;
  Obtain the optimal $\mathbf{\Theta}$: $\overline{\mathbf{\Theta}} \leftarrow \mathbf{\Theta}$\;
Compute $\overline{\Delta \mathbf{h}}_{\mathrm{RU}}(\overline{\mu},\overline{\mathbf{\Theta}})$ using (\ref{Delta_h_Opt})\;
{\color{black}Compute $\overline{\mathbf{h}}_{\mathrm{BU}}$ using (\ref{worst-case h_BU})}\;
Compute $\overline{\mathbf{w}}$ using (\ref{Opt_w})\;
{\bf Output:} $\overline{\mathbf{\Theta}}$, $\overline{\mathbf{w}}$\;

\end{small}
\end{algorithm}

{\color{black}
\subsection{Convergence Analysis}

Since \textbf{Algorithm 1} includes an iteration procedure, it is necessary to investigate its convergence behaviour. To facilitate the analysis, we first prove the following property:

\begin{proposition}
With a given $\bm{\theta}$, if $\mu$ grows, the value of $\bm{\theta}^{\mathrm{T}} \mathbf{\Upsilon}(\mu) \bm{\theta}^*$ increases, whereas the value of $\bm{\theta}^{\mathrm{T}} \mathbf{\Gamma}(\mu) \bm{\theta}^*$ decreases.
\end{proposition}
\begin{proof}
The proof is given in Appendix D.
\end{proof}

Based on Proposition 3, we are now ready to prove the convergence of \textbf{Algorithm 1}.

\begin{proposition}
The proposed \textbf{Algorithm 1} is convergent.
\end{proposition}

\begin{proof}
Let $\bm{\theta}_k$ ($\bm{\theta}_{k+1}$) and $\mu_k$ ($\mu_{k+1}$) be the solutions of $\bm{\theta}$ and $\mu$ in the $k$-th ($(k+1)$-th) iteration, respectively. To prove the algorithm convergence, we need to prove that
\begin{equation}
\bm{\theta}_{k+1}^{\mathrm{T}} \mathbf{\Upsilon}(\mu_{k+1}) \bm{\theta}_{k+1}^* \geq \bm{\theta}_{k}^{\mathrm{T}} \mathbf{\Upsilon}(\mu_{k}) \bm{\theta}_{k}^*
\end{equation}
strictly holds for $\forall k>0$. 

First, when $\mu_{k+1}$ is given, we can readily obtain
\begin{equation}
\bm{\theta}_{k+1}^{\mathrm{T}} \mathbf{\Upsilon}(\mu_{k+1}) \bm{\theta}_{k+1}^* \geq \bm{\theta}_{k}^{\mathrm{T}} \mathbf{\Upsilon}(\mu_{k+1}) \bm{\theta}_{k}^*,
\end{equation}
owing to the maximization of the objective function in problem (\ref{Optimize_theta}). Afterwards, when $\bm{\theta}_{k}$ is given, we need to compare $\bm{\theta}_{k}^{\mathrm{T}} \mathbf{\Upsilon}(\mu_{k+1}) \bm{\theta}_{k}^*$ and $\bm{\theta}_{k}^{\mathrm{T}} \mathbf{\Upsilon}(\mu_{k}) \bm{\theta}_{k}^*$. 

Since $\bm{\theta}_{k}$ is the solution of problem (\ref{Optimize_theta}) with a given $\mu_{k}$, according to (\ref{Optimize_theta}b), we have
\begin{equation}\label{53}
\bm{\theta}_k^{\mathrm{T}}\mathbf{\Gamma}(\mu_k)\bm{\theta}_k^*\geq\epsilon_{\Delta \mathbf{h}_{\mathrm{RU}}}^2.
\end{equation}
Subsequently, by solving problem (\ref{Optimize_mu}), we obtain $\mu_{k+1}$ with a given $\bm{\theta}_k$, such that
\begin{equation}\label{54}
\bm{\theta}_k^{\mathrm{T}}\mathbf{\Gamma}(\mu_{k+1})\bm{\theta}_k^*=\epsilon_{\Delta \mathbf{h}_{\mathrm{RU}}}^2.
\end{equation}
Thus, based on (\ref{53}) and (\ref{54}), we have
\begin{equation}\label{55}
\bm{\theta}_k^{\mathrm{T}}\mathbf{\Gamma}(\mu_k)\bm{\theta}_k^*\geq\bm{\theta}_k^{\mathrm{T}}\mathbf{\Gamma}(\mu_{k+1})\bm{\theta}_k^*.
\end{equation}
Combining (\ref{55}) with Proposition 3, we obtain $\mu_k\leq \mu_{k+1}$, resulting in $\bm{\theta}_k^{\mathrm{T}}\mathbf{\Upsilon}(\mu_k)\bm{\theta}_k^*\leq\bm{\theta}_k^{\mathrm{T}}\mathbf{\Upsilon}(\mu_{k+1})\bm{\theta}_k^*$. Consequently, we obtain
\begin{equation}
\bm{\theta}_{k+1}^{\mathrm{T}} \mathbf{\Upsilon}(\mu_{k+1}) \bm{\theta}_{k+1}^* \geq \bm{\theta}_{k}^{\mathrm{T}} \mathbf{\Upsilon}(\mu_{k+1}) \bm{\theta}_{k}^*\geq \bm{\theta}_{k}^{\mathrm{T}} \mathbf{\Upsilon}(\mu_{k}) \bm{\theta}_{k}^*,
\end{equation}
implying that the solution of \textbf{Algorithm 1} in the $(k+1)$-th iteration is better than that in the $k$-th iteration. In addition, the objective value of $\bm{\theta}^{\mathrm{T}} \mathbf{\Upsilon}(\mu) \bm{\theta}^*$ is upper-bounded owing to the constraint of $|[\bm{\theta}]_{i}|=\beta$ for $i=1,2,...,N$. Therefore, we prove that \textbf{Algorithm 1} is convergent.
\end{proof}

}

\subsection{Computational Complexity Analysis}

Here we analyse the approximate computational complexity of \textbf{Algorithm 1} through the comparisons with the following benchmarks: 

\begin{itemize}
\item[1)] \textbf{Benchmark 1 (B1)}: the non-robust approach in \cite{Zhong-TCOM2020}, which directly uses the estimated CSI depending on angle-of-departure/arrival (AOD/AOA) to design the beams, without considering the CSI errors.

\item[2)] \textbf{Benchmark 2 (B2)}: the worst-case robust beamforming optimization algorithms in \cite{G.Zhou-TSP2020} or \cite{G.Zhou-WCL2020} based on S-procedure and penalty CCP. 
\end{itemize}

\begin{table*}
\caption{Approximate computational complexities.}
\label{Table_Complexity}
\centering
\begin{small}
\begin{tabular}{ccc}
\hline
Algorithms & Approximate complexities per iteration & Overall complexities\\
\hline
Proposed algorithm with SDR & $o_{\mu}+o_{(SDR)}$ & $\left(o_{\mu}+o_{(SDR)}\right)\times T_P^{Con}$\\
Proposed algorithm with BnB & $o_{\mu}+o_{(BnB)}$ & $\left(o_{\mu}+o_{(BnB)}\right)\times T_P^{Con}$\\
Robust beamforming in \cite{G.Zhou-TSP2020} or \cite{G.Zhou-WCL2020} (\textbf{B2}) & $o_{\mathbf{w}}+o_{\mathbf{\Theta}}$ & $\left(o_{\mathbf{w}}+o_{\mathbf{\Theta}}\right)\times T_R^{Con}$\\
Non-robust beamforming in \cite{Zhong-TCOM2020} (\textbf{B1}) & $o_{(SDR)}$ & $o_{(SDR)}$\\
\hline
\end{tabular}
\end{small}
\end{table*}

Table \ref{Table_Complexity} lists the approximate computational complexities of our algorithm and \textbf{B1}, \textbf{B2}, where

\begin{itemize}
\item[1)] $o_{\mu}\approx\mathcal{O}\left(\log_2 W_{\mu}\right)$ is the complexity of the bisection search, which is used to find the optimal $\mu$ from (\ref{bisection}), where $W_{\mu}$ denotes the length of the search interval.

\item[2)] $o_{(SDR)}\approx\mathcal{O}\left((N+1)^4 N^{\frac{1}{2}} \log_2\frac{1}{\epsilon_{(SDR)}}\right)$ is the complexity of SDR \cite{Z.-Q. Luo-SPM2010}, which is used to solve problem (\ref{P5}), where $\epsilon_{(SDR)}>0$ represents the solution accuracy.

\item[3)] $o_{(BnB)}\approx\left(\prod_{i=1}^{N}\lceil \frac{2(\ell_u-\ell_l)}{\delta} \rceil\right) o_{(SDR)}$ is the complexity of BnB, which is used to solve problem (\ref{P5}), where $\delta$ is specified in [Lemma 4, 45].

\item[4)] $o_{\mathbf{w}}\approx\mathcal{O}((NM+M+2)^{\frac{1}{2}}M[M^2 + M ((NM+1)^2 + (1+M)^2) + ((NM+1)^3 + (1+M)^3)
])$ and $o_{\mathbf{\Theta}}\approx\mathcal{O}(
(NM+2+2N)^{\frac{1}{2}} N [N^2 + N((NM+1)^2+1) + ((NM+1)^3 + 1) +N^2]
)$ are the complexities of the transmit beamforming optimization and passive beamforming optimization of \textbf{B2}.

\item[5)] $T_P^{Con}$ and $T_R^{Con}$ stand for the total iteration times required for the convergence of \textbf{Algorithm 1} and \textbf{B2}.
\end{itemize}

Table \ref{Table_Complexity} indicates that: 1) \textbf{B1} requires only one time of SDR optimization. It is the simplest approach but may suffer from serious performance degradation in the presence of CSI errors.
2) In each iteration, \textbf{B2} is more computationally complex than the proposed algorithm with SDR, whereas their overall complexities depend on the total iteration times for convergence, {\color{black}which will be numerically investigated in Section V-C}.
3) Although the proposed algorithm with BnB is also computationally complicated in one iteration, it has the advantages of providing near-optimal solutions and dealing with arbitrary phase shift argument sets. These advantages will be justified in Section V-C as well.

\section{Simulation Results}

In this section, the system parameters are configured and the performance metric is defined. Then, the CSI error bound, derived in Section III, is numerically investigated. Finally, the performance of \textbf{Algorithm 1} is evaluated and compared with those of \textbf{B1} and \textbf{B2}, in terms of the worst-case SNR at the {\color{black}UE} and the algorithm efficiency.

\subsection{System Parameters and Performance Metric}

\subsubsection{System Parameters}

In the simulations, the system parameters are set as follows. $\mathbf{q}_1$, $\mathbf{v}_1$ and $\widehat{\mathbf{p}}$ are set as $\mathbf{q}_1=(0,0,25)^{\mathrm{T}}$, $\mathbf{v}_1=(2,-2,26)^{\mathrm{T}}$ and $\widehat{\mathbf{p}}=(10,5,18)^{\mathrm{T}}$ in meters. The antenna/element spacing is $d_{BS}=d_{RIS}=0.5$ cm. The number of transmit antennas is $M=32$.
The total transmit power is $P_T=27$ dBm. The noise power at the receiver side is $\sigma_n^2=-80$ dBm. The carrier frequency is $f_c=60$ GHz. The speed of light is $c\approx 2.99792458\times 10^8$ m/s. The signal wavelength is calculated by $\lambda=\frac{c}{f_c}$. The reflection amplitude is $\beta=1$.
For the large-scale path loss, we set $\zeta_0=-30$ dB, $d_0=1$ m, and $\alpha=2.2$ \cite{C.Pan-JSAC2020, K.Zhi-WCL2021}. The Rician factor is $\kappa_R=20$. The $\ell_2$-norm of $\mathbf{h}_{\mathrm{RU}}^{\mathrm{NLoS}}$ and $\mathbf{h}_{\mathrm{BU}}$ are set as $\delta_{\mathrm{RU}}^{\mathrm{NLoS}}=10^{-4}$ and $\delta_{\mathrm{BU}}=0$, respectively.
When performing \textbf{Algorithm 1}, we set $\epsilon_R=0.0001$, and configure the error tolerance of BnB and the maximum of total iteration times of BnB to be $0.0002$ and $1000$, respectively.
In addition, the user location error bound $\epsilon_{\Delta\mathbf{p}}$ and the number of reflecting elements $N$ (or $L$) will vary for diverse observations.

\subsubsection{Performance Metric}

For evaluating the performance of the proposed algorithm, the worst-case SNR at the {\color{black}UE} with respect to the optimized transmit and passive beamforming is considered as a performance metric:
\begin{equation}\label{S_w}
S_w(\overline{\mathbf{\Theta}},\overline{\mathbf{w}})
\!=\!\frac{P_T \left|\left[(\mathbf{h}_{\mathrm{RU}}^{\mathrm{worst}})^{\mathrm{H}}\overline{\mathbf{\Theta}}\mathbf{H}_{\mathrm{BR}} + \overline{\mathbf{h}}_{\mathrm{BU}}^{\mathrm{H}}\right]\overline{\mathbf{w}}\right|^2}{\sigma_n^2},
\end{equation}
which is a function of $\overline{\mathbf{\Theta}}$ and $\overline{\mathbf{w}}$, where $\mathbf{h}_{\mathrm{RU}}^{\mathrm{worst}}=\widehat{\mathbf{h}}_{\mathrm{RU}}^{\mathrm{LoS}}+\overline{\Delta \mathbf{h}}_{\mathrm{RU}}(\overline{\mu},\overline{\mathbf{\Theta}})$ is the worst-case channel from RIS to {\color{black}UE}. In the simulation results, $S_w(\overline{\mathbf{\Theta}},\overline{\mathbf{w}})$ is presented as original ratio value instead of decibel (dB).

\subsection{CSI Error Bound}

Based on the system setup, we first numerically verify the theoretical derivation of the CSI error bound in Section III.
Specifically, we compare the following two results to examine the tightness of the CSI error bound: 

1) $\epsilon_{\Delta \mathbf{h}_{\mathrm{RU}}}$, which is theoretically derived from (\ref{Finally_Derived_CSIerrorbound}).

2) The actual CSI error bound, which is obtained by finding the maximum of $\|\Delta \mathbf{h}_{\mathrm{RU}}\|_2$ from $50000$ computations of $\|\Delta \mathbf{h}_{\mathrm{RU}}\|_2=\|\mathbf{h}_{\mathrm{RU}}-\widehat{\mathbf{h}}_{\mathrm{RU}}^{\mathrm{LoS}}\|_2$, where each computation corresponds to one channel realization of $\mathbf{h}_{\mathrm{RU}}$ with a random $\mathbf{p}$ in the spherical uncertainty region of $\|\Delta\mathbf{p}\|_2\leq \epsilon_{\Delta\mathbf{p}}$.

\begin{figure}[!t]
\includegraphics[width=3.2in]{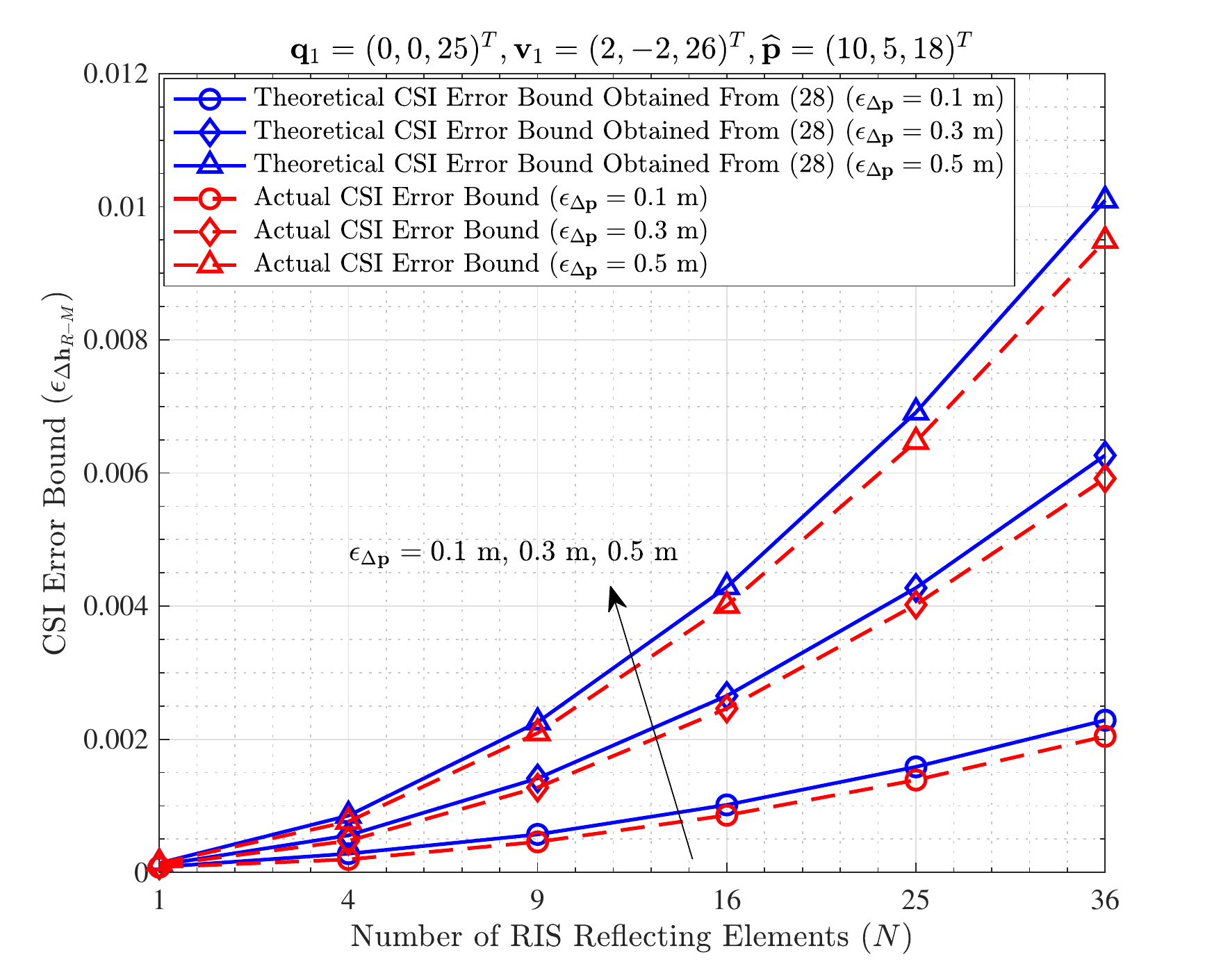}
\hfil
\centering
\caption{The CSI error bounds as functions of $N$ with different $\epsilon_{\Delta\mathbf{p}}$. }
\label{CSI_Error_Bound_N}
\end{figure}

Fig. \ref{CSI_Error_Bound_N} depicts the CSI error bounds with respect to $N$ and $\epsilon_{\Delta\mathbf{p}}$. It indicates that first, the overall theoretical results fit well with the actual ones, implying that the derivation of the CSI error bound is correct. Second, the CSI error bounds increase as $N$ or $\epsilon_{\Delta\mathbf{p}}$ grows, illustrating that more reflecting elements or higher level of user location uncertainty will lead to graver CSI error. Third, the red dashed curves are always below the blue solid curves, manifesting that the theoretical results are upper bounds of the corresponding actual results, which become tight when the level of the user location uncertainty is {\color{black}moderate}.

\begin{figure*}[!t]
\centering
\subfloat[$\widehat{p}_y=5$ m and $\widehat{p}_z=18$ m.]{\includegraphics[width=2.2in]{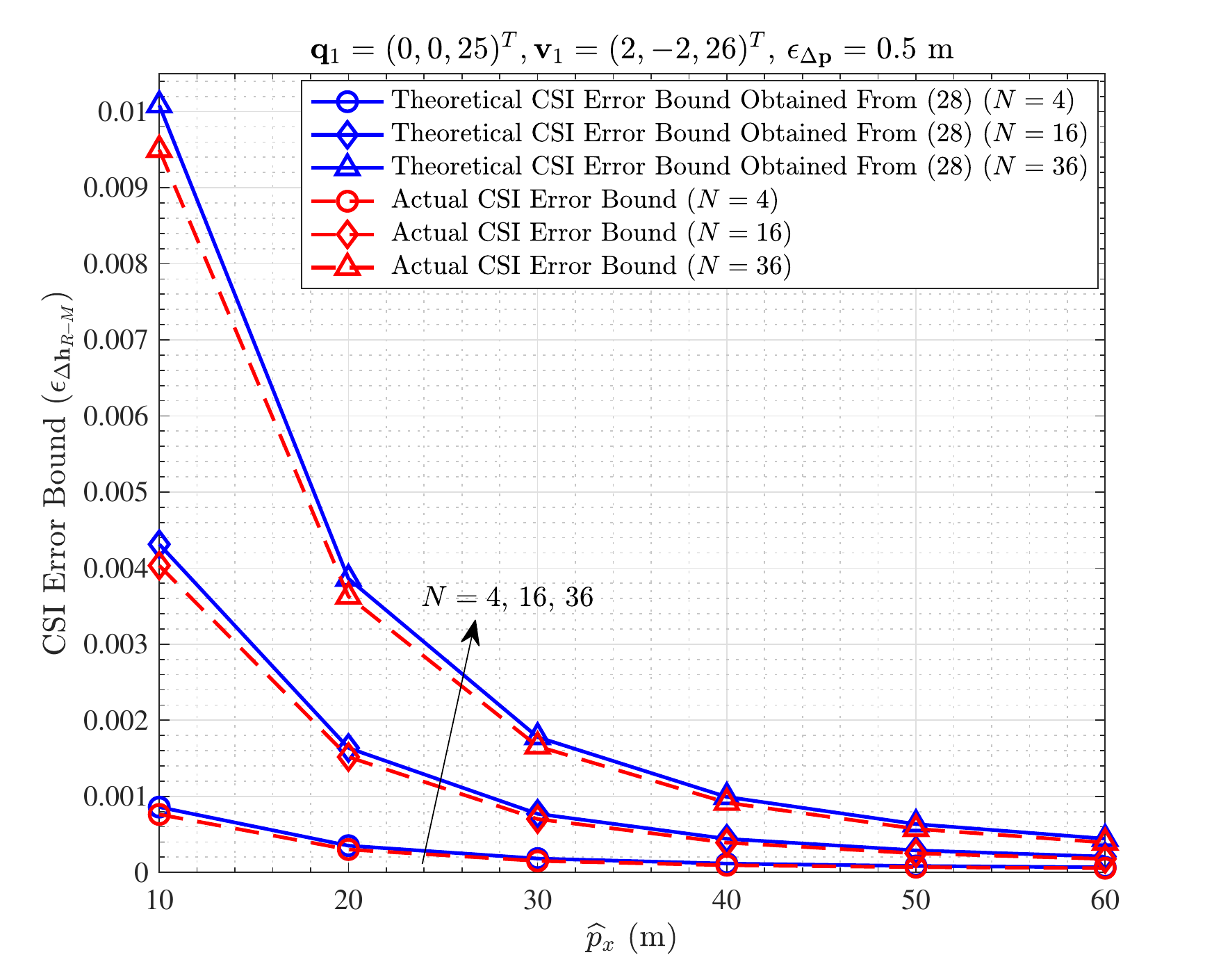}}
\label{CSI_Error_Bound_Final_Nearfield_x}
\subfloat[$\widehat{p}_x=10$ m and $\widehat{p}_z=18$ m.]{\includegraphics[width=2.2in]{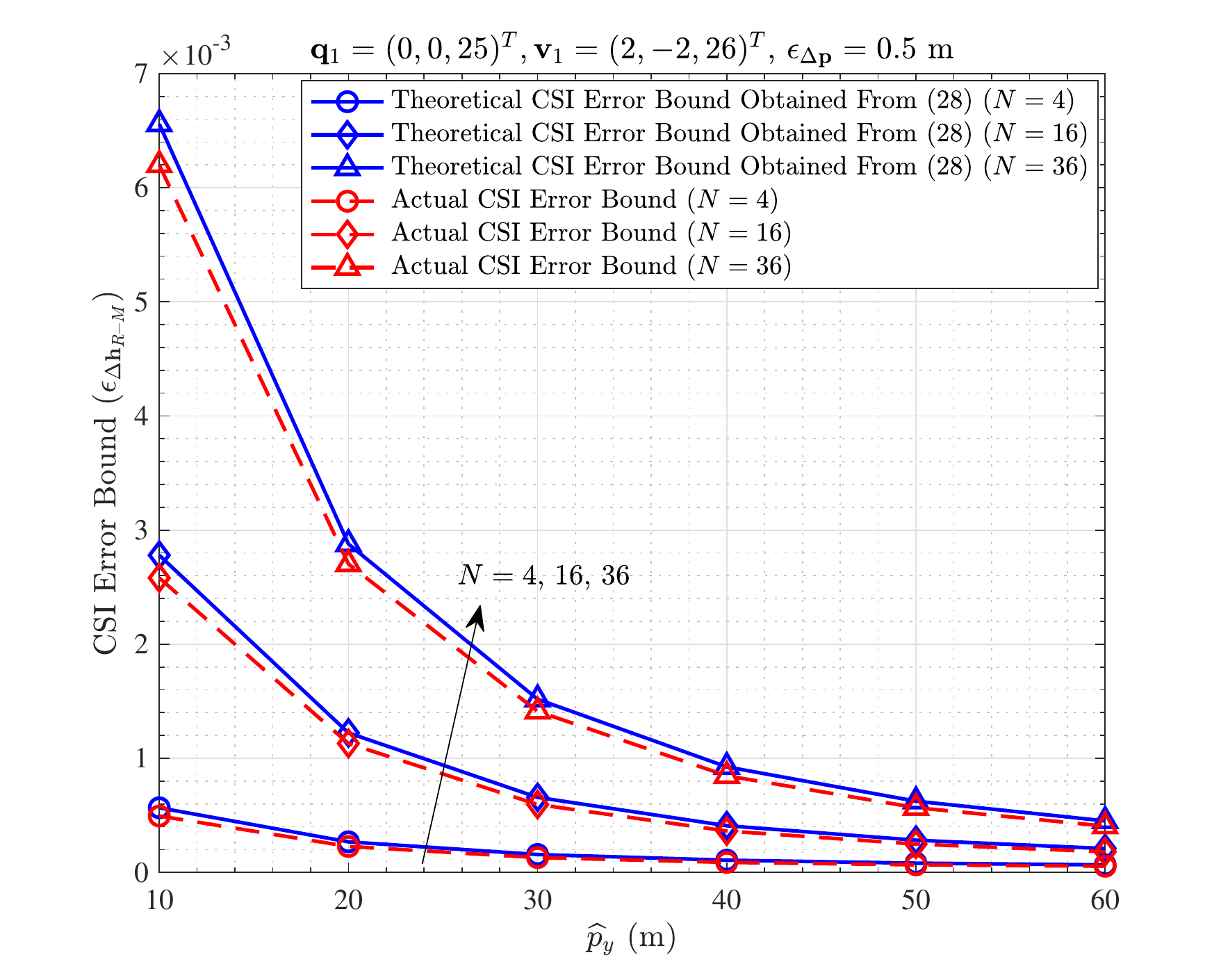}}
\label{CSI_Error_Bound_Final_Nearfield_y}
\subfloat[$\widehat{p}_x=10$ m and $\widehat{p}_y=5$ m.]{\includegraphics[width=2.2in]{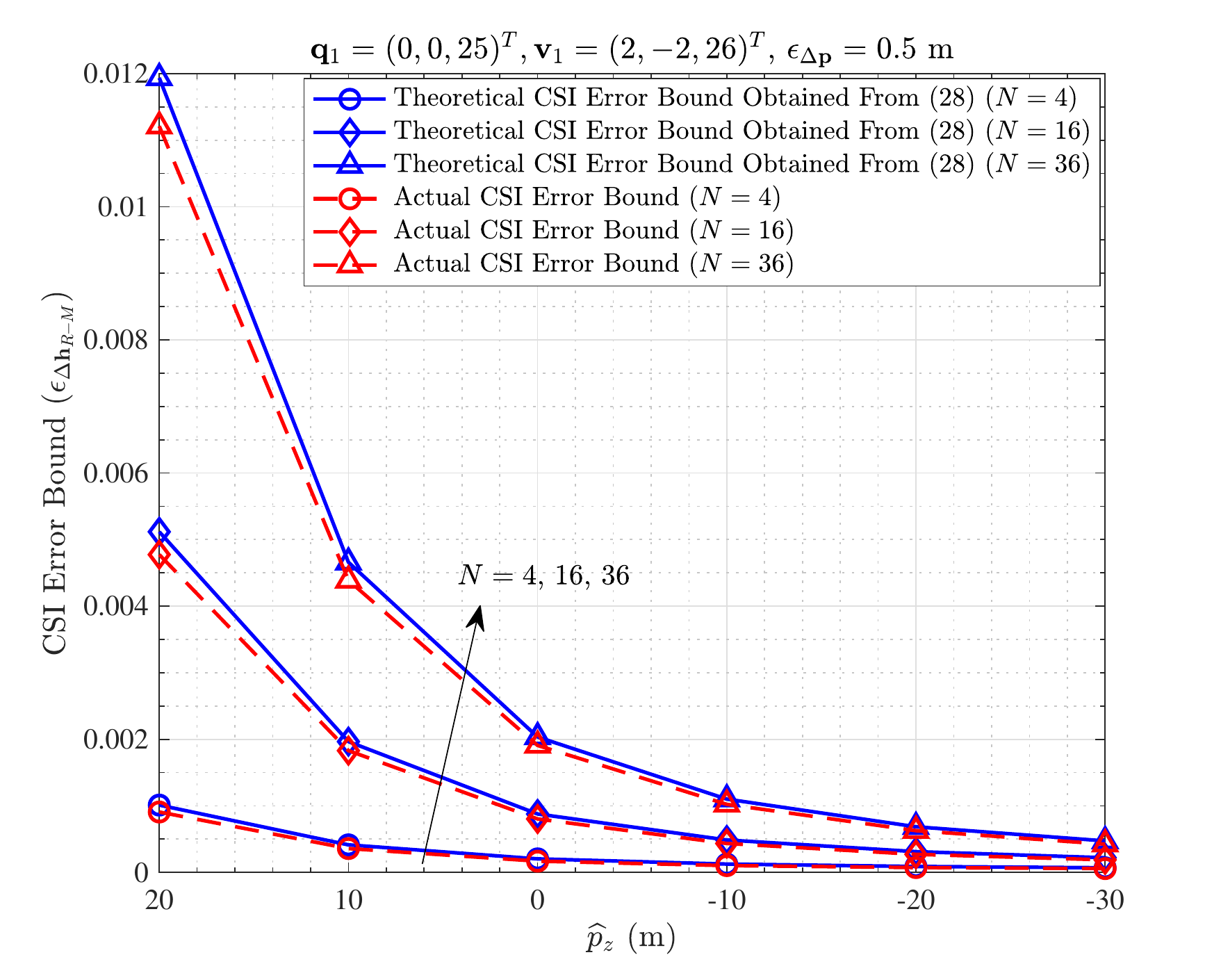}}
\label{CSI_Error_Bound_Final_Nearfield_z}
\hfil
\caption{The CSI error bounds as functions of user coordinates with $\epsilon_{\Delta\mathbf{p}}=0.5$ m and different $N$.}
\label{CSI_Error_Bound_xyz}
\end{figure*}

Fig. \ref{CSI_Error_Bound_xyz} displays the CSI error bounds with respect to user coordinates, with $\epsilon_{\Delta\mathbf{p}}=0.5$ m and different $N$. It demonstrates that when the distance between the RIS and the {\color{black}UE} increases while the user location error bound is fixed, the theoretical CSI error bounds descend and become tighter. 

\subsection{Algorithm Performance Evaluation}

Subsequently, we investigate the performance of the proposed algorithm through the comparisons with \textbf{B1} and \textbf{B2}, in terms of both the worst-case SNR and the algorithm efficiency. As \textbf{B1} and \textbf{B2} can only apply to the phase shift argument set of $[0,2\pi]$, here we first compare the performances under the condition of $\mathcal{S}=[0,2\pi]$, and then justify the advantages of our algorithm with BnB in the cases of other $\mathcal{S}=[\ell_l,\ell_u]$.


\begin{figure}[!t]
\includegraphics[width=3.2in]{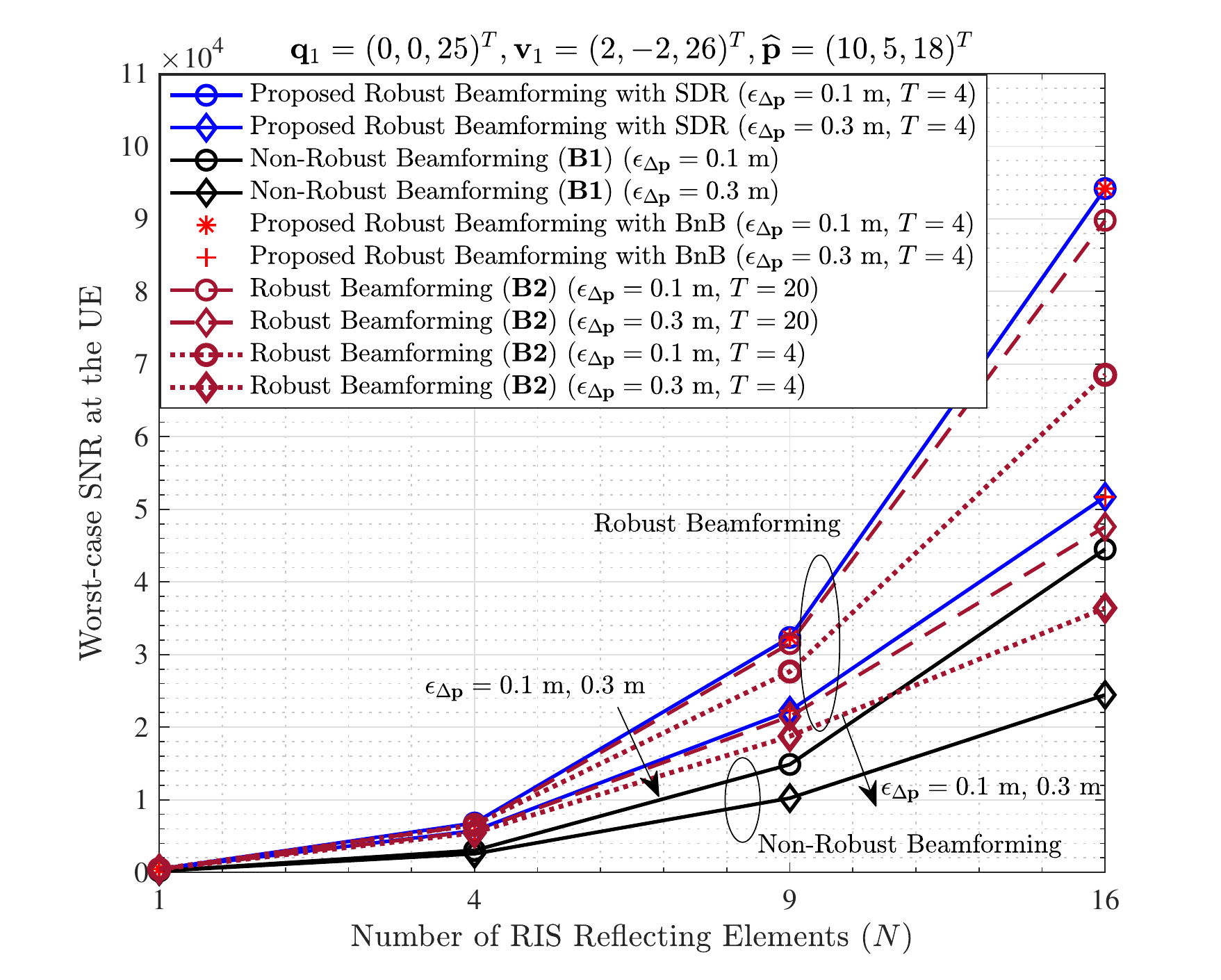}
\hfil
\centering
\caption{\color{black}Worst-case SNRs at the UE with respect to $N$ when $\mathcal{S}=[0,2\pi]$. }
\label{Worstcase_Signal_Strength}
\end{figure}

\begin{figure}[!t]
\includegraphics[width=3.2in]{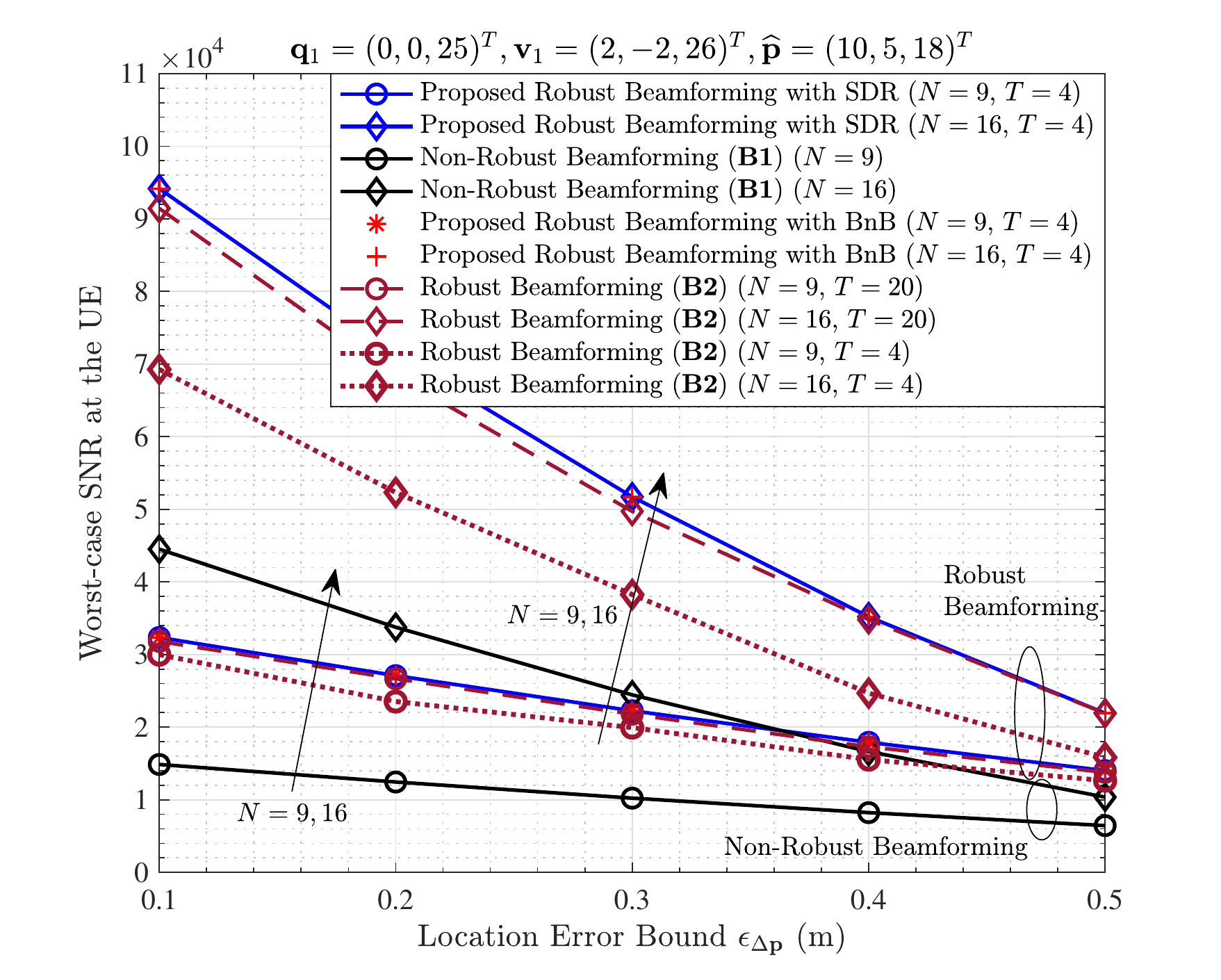}
\hfil
\centering
\caption{\color{black}Worst-case SNRs at the UE with respect to $\epsilon_{\Delta\mathbf{p}}$ when $\mathcal{S}=[0,2\pi]$. }
\label{Worstcase_Signal_Strength_2}
\end{figure}

Fig. \ref{Worstcase_Signal_Strength} and Fig.  \ref{Worstcase_Signal_Strength_2} show the worst-case SNRs at the {\color{black}UE} with respect to $N$ and $\epsilon_{\Delta\mathbf{p}}$ when $\mathcal{S}=[0,2\pi]$. 
Results reveal that: 1) the worst-case SNRs at the {\color{black}UE} are strongly influenced by the number of the RIS reflecting elements and the level of the CSI uncertainty.
2) The proposed approach significantly outperforms the non-robust approach, i.e. \textbf{B1}, thus validating the robustness of our algorithm. 
3) The proposed algorithm is superior to \textbf{B2} in terms of the worst-case SNR performance, when the number of the iterations is small ($T<5$ in specific, as shown in the figures). However, it can also be seen that if $T$ is large enough, the worst-case SNRs of \textbf{B2} are close to those of the proposed algorithm. Thereupon, we will evaluate the performance from the perspective of the convergence rate as well, in order to show a more comprehensive comparison.

\begin{figure}[!t]
\includegraphics[width=3.2in]{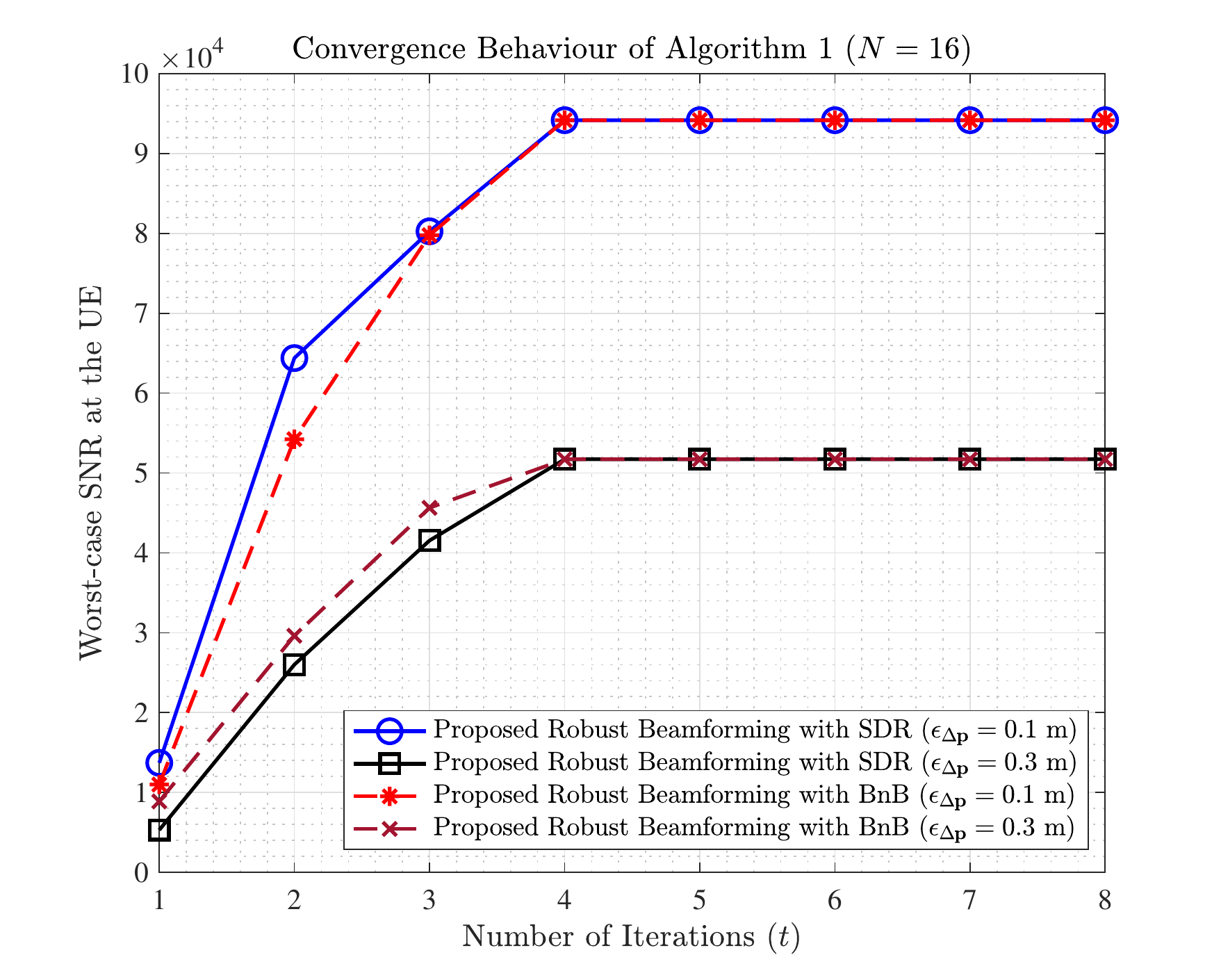}
\hfil
\centering
\caption{\color{black}Convergence behaviour of the proposed algorithm with respect to distinct $\epsilon_{\Delta\mathbf{p}}$ when $N=16$. }
\label{Convergence_Behaviour}
\end{figure}

{\color{black} Fig. \ref{Convergence_Behaviour} depicts the convergence behaviour of the proposed algorithm at $N=16$.  In Fig. \ref{Convergence_Behaviour}, each simulation curve is obtained by averaging on five Monte Carlo trials with randomly initialized RIS phase shifts. It is demonstrated that the proposed algorithm with SDR/BnB is convergent, verifying the theoretical analysis in Proposition 3 and 4. Furthermore, it also indicates that the proposed algorithm can possibly converge to the optimum after around 4 iterations. This could be a potential key advantage facilitating the practical application of the proposed design.}

\begin{figure}[!t]
\includegraphics[width=3.2in]{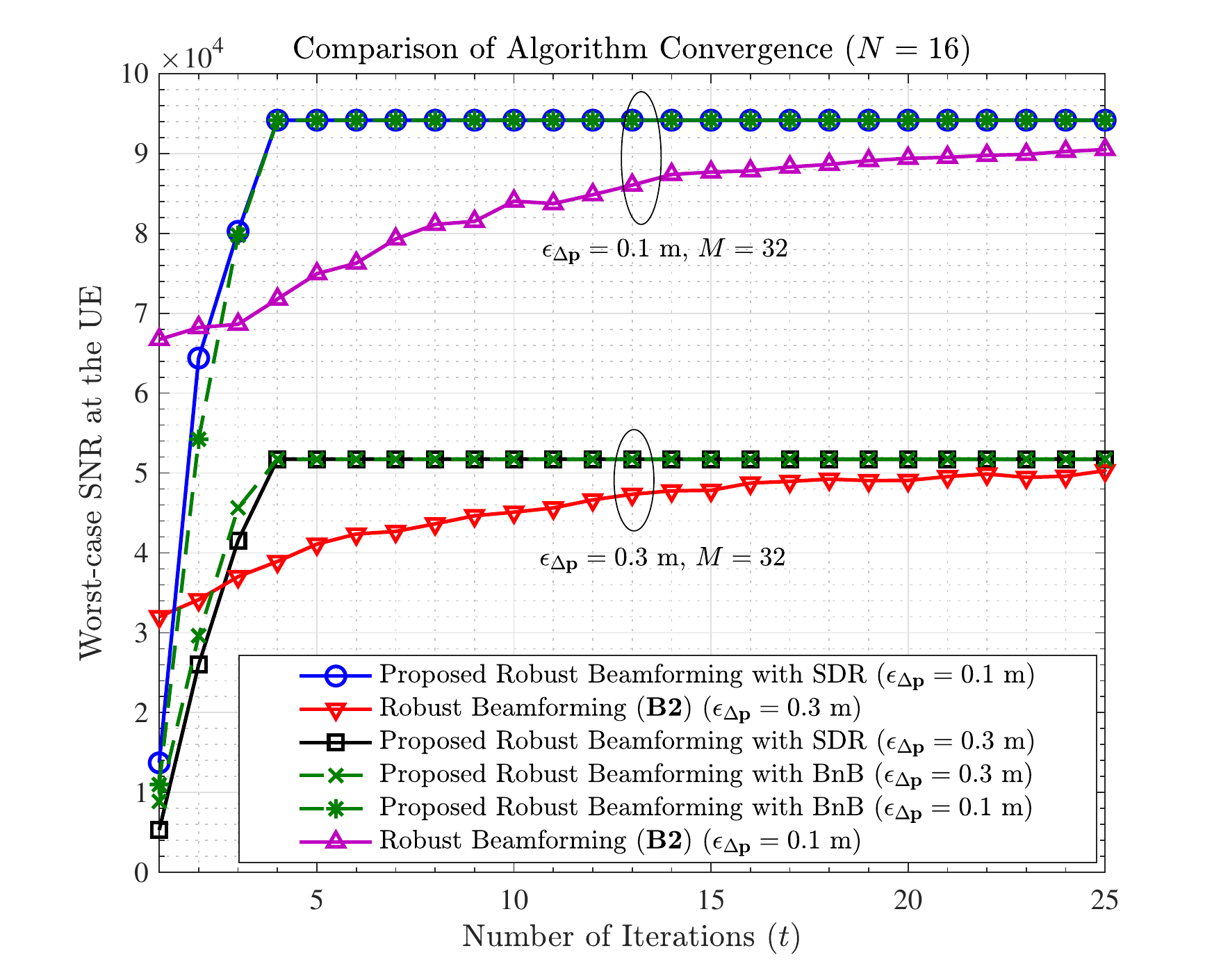}
\hfil
\centering
\caption{\color{black}Comparisons of the convergence rates of the proposed algorithm and \textbf{B2}, when $\mathcal{S}=[0,2\pi]$, $N=16$, and $\epsilon_{\Delta\mathbf{p}}=0.1$ m, $0.3$ m. }
\label{Convergence_Comparison}
\end{figure}

Fig. \ref{Convergence_Comparison} compares the convergence rates of the proposed algorithm and \textbf{B2} when $\mathcal{S}=[0,2\pi]$, showing that the proposed algorithm converges after around $T_P^{Con}=4$ iterations, while \textbf{B2} converges after around $T_R^{Con}=20$ iterations. This demonstrates that our algorithm can outperform \textbf{B2} in terms of the convergence rate, illustrating meanwhile that our proposed design would be more advantageous in terms of the algorithm efficiency.

\begin{figure}[!t]
\includegraphics[width=3.2in]{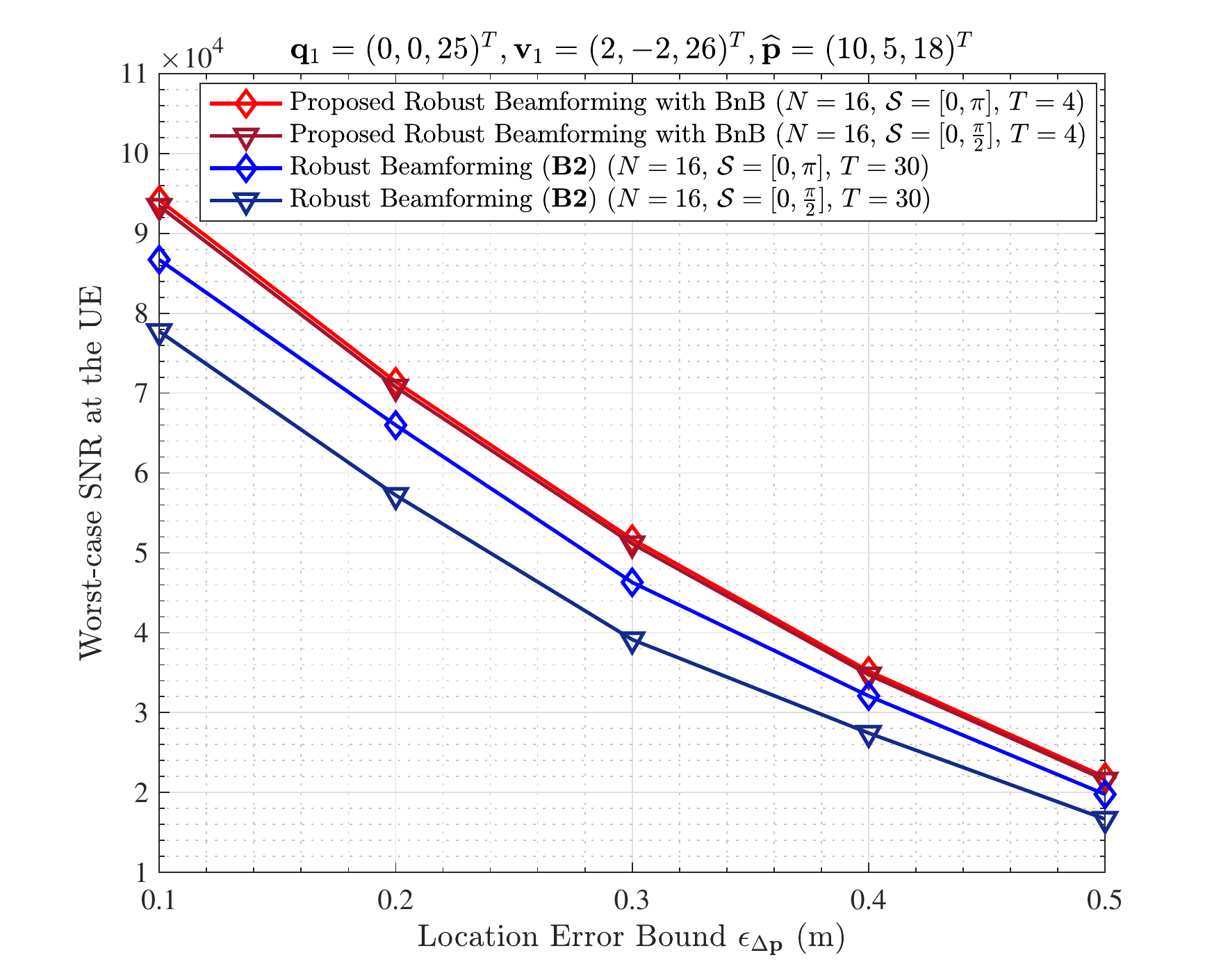}
\hfil
\centering
\caption{\color{black}Worst-case SNRs at the UE when the phase shift argument sets are $\mathcal{S}=[0,\pi]$ and $\mathcal{S}=[0,\frac{\pi}{2}]$.}
\label{Worstcase_Signal_Strength_0topi}
\end{figure}

Finally, in order to justify that the proposed algorithm with BnB possesses the advantage of handling arbitrary phase shift argument sets, we take $\mathcal{S}=[\ell_l,\ell_u]$ as an example, and investigate the performances of \textbf{B2} and the proposed algorithm with BnB in the cases of $\mathcal{S}=[0,\pi]$ and $\mathcal{S}=[0,\frac{\pi}{2}]$. Since the solutions of \textbf{B2} intrinsically belong to $[0,2\pi]$, we perform the rounding operations for $\theta_{\ell+(k-1)L}$ to map the arguments into $\theta_{\ell+(k-1)L}\in[0,\pi]$ and $\theta_{\ell+(k-1)L}\in[0,\frac{\pi}{2}]$ after running \textbf{B2}. Specifically, under the argument constraint of $\mathcal{S}=[0,\pi]$, we perform
\begin{equation}
\left\{\begin{matrix}\theta_{\ell+(k-1)L} = \pi,\ \ \ \mathrm{if}\ \theta_{\ell+(k-1)L}\in(\pi,\frac{3\pi}{2}) ;\\\theta_{\ell+(k-1)L} = 0,\ \ \ \mathrm{if}\ \theta_{\ell+(k-1)L}\in[\frac{3\pi}{2},2\pi],\\\end{matrix}\right.
\end{equation}
and under the argument constraint of $\mathcal{S}=[0,\frac{\pi}{2}]$, we perform
\begin{equation}
\left\{\begin{matrix}\theta_{\ell+(k-1)L} = \frac{\pi}{2},\ \ \ \mathrm{if}\ \theta_{\ell+(k-1)L}\in(\frac{\pi}{2},\frac{5\pi}{4}) ;\\\theta_{\ell+(k-1)L} = 0,\ \ \ \mathrm{if}\ \theta_{\ell+(k-1)L}\in[\frac{5\pi}{4},2\pi].\\\end{matrix}\right.
\end{equation}

Fig. \ref{Worstcase_Signal_Strength_0topi} depicts the worst-case SNRs at the {\color{black}UE} with the two phase shift argument sets, demonstrating that when $\mathcal{S}\neq[0,2\pi]$, the BnB can still maintain good SNR performance, whereas \textbf{B2} performs worse as \textbf{B2} only applies to the phase shift argument constraint of $\mathcal{S}=[0,2\pi]$.

\vspace*{-8pt}

\section{Conclusions and Prospects}

In this paper, when adopting the user location to acquire the CSI in the RIS-aided wireless communication system, we theoretically derived an approximate CSI error bound in the presence of user location uncertainty, with the aid of Taylor approximation, triangle and power mean inequalities, and SDR method. Then, in order to resist the impact of location error on the beamforming design, we formulated a worst-case robust beamforming optimization problem to jointly optimize the transmit and passive beamforming. To solve this non-convex problem efficiently, we proposed an iterative optimization approach based on Lagrange multiplier and matrix inverse lemma, and evaluated its performance in terms of the worst-case SNR at the UE, convergence behaviour, and algorithm efficiency.
Results verified the theoretical derivation of the CSI error bound, and validated the robustness of the proposed approach. Compared with the existing S-procedure and penalty CCP based robust beamforming techniques, our algorithm could perform better and converge more quickly to the optimum. Consequently, the proposed design in this work would be promising in contributing to the development of location information assisted RIS-aided communications.

{\color{black}Finally, a brief discussion on the potential applicability and extensibility of the proposed beamforming design is presented, from the following two aspects:

\textbf{1) UE mobility}: If the UE is moving, the system performance could be maintained only when the actual UE location is still within the uncertainty region of $\|\Delta\mathbf{p}\|_2\leq \epsilon_{\Delta\mathbf{p}}$ before the transmit/passive beamforming is updated. Nevertheless, the high-speed mobility of the UE may generally not be neglected in several cases, e.g. when the UE is inside a moving vehicle on the high-way. If the UE location changes significantly within a short period, the location information assisted beamforming optimization may provide a hysteretic result, which degrades the current system performance to some extent. Hence, it is worth developing a high-efficiency and low-latency approach to combat the potential negative effect caused by the high-speed UE movement in the future.
\textbf{2) RIS calibration}: This work considers the constant-modulus reflection model, implying that the inherent element responses of the RIS reflectors are approximately invariant and same toward arbitrary incident/reflective directions, which however, may not exactly be the case in practice. As such, the mismatch between the reconstructed and the actual RIS-UE channels may generally be non-negligible, on account of the differences of the RIS element responses toward distinct UE locations. To compensate for such a mismatch, developing an efficient RIS calibration procedure deserves further effort.

}


\appendices

\section{Proof of Lemma 1}

In Appendix A, we prove the results in Lemma 1. Since the cosine term in (\ref{Omega_Delta_p}) is the dominant factor that makes the derivation for the maximum of $\Omega\left(\Delta\mathbf{p}\right)$ become challenging, we first apply the fourth-order Taylor expansion of $\cos{x}\approx 1-\frac{1}{2!}x^2+\frac{1}{4!}x^4$ at $x=0$ to approximate $\cos{\left(\frac{2\pi}{\lambda} \mathfrak{U}_{\ell+(k-1)L}\right)}$, and obtain
\begin{equation}\label{Cosine_Taylor}
\cos{\left(\frac{2\pi}{\lambda} \mathfrak{U}_{\ell+(k-1)L}\right)}\approx
1 - \frac{2\pi^2}{\lambda^2}\mathfrak{U}_{\ell+(k-1)L}^2
+ \frac{2\pi^4}{3\lambda^4}\mathfrak{U}_{\ell+(k-1)L}^4.
\end{equation}
Substituting (\ref{Cosine_Taylor}) into (\ref{Omega_Delta_p}), after some manipulations, we have 
\begin{equation}\label{Omega_Delta_p_fourth_Taylor}
\Omega\left(\Delta\mathbf{p}\right) \approx T_1\left(\Delta\mathbf{p}\right) + T_2\left(\Delta\mathbf{p}\right),
\end{equation}
where $T_1\left(\Delta\mathbf{p}\right)$ and $T_2\left(\Delta\mathbf{p}\right)$ are expressed as (\ref{T_1_D_p}) and (\ref{T_2_D_p}) on the top of the next page.
\begin{figure*}[!t]
\normalsize
\begin{scriptsize}
\begin{equation}\label{T_1_D_p}
\begin{split}
T_1\left(\Delta\mathbf{p}\right)=&
\sum_{k=1}^{L}\sum_{\ell=1}^{L}\left\{\left\|\mathbf{v}_{\ell+(k-1)L}-\widehat{\mathbf{p}}-\Delta\mathbf{p}\right\|_2^{-\alpha} + \left\|\mathbf{v}_{\ell+(k-1)L}-\widehat{\mathbf{p}}\right\|_2^{-\alpha} - 2\left\|\mathbf{v}_{\ell+(k-1)L}-\widehat{\mathbf{p}}-\Delta\mathbf{p}\right\|_2^{-\frac{\alpha}{2}} \left\|\mathbf{v}_{\ell+(k-1)L}-\widehat{\mathbf{p}}\right\|_2^{-\frac{\alpha}{2}} \right\},
\end{split}
\end{equation}
\begin{equation}\label{T_2_D_p}
\begin{split}
T_2\left(\Delta\mathbf{p}\right)= \sum_{k=1}^{L}\sum_{\ell=1}^{L}\left\{ \left\|\mathbf{v}_{\ell+(k-1)L}-\widehat{\mathbf{p}}-\Delta\mathbf{p}\right\|_2^{-\frac{\alpha}{2}} 
 \left\|\mathbf{v}_{\ell+(k-1)L}-\widehat{\mathbf{p}}\right\|_2^{-\frac{\alpha}{2}} 
\left(\frac{4\pi^2}{\lambda^2}\mathfrak{U}_{\ell+(k-1)L}^2
- \frac{4\pi^4}{3\lambda^4}\mathfrak{U}_{\ell+(k-1)L}^4\right)
\right\}.
\end{split}
\end{equation}
\end{scriptsize}
\vspace*{-18pt}
\end{figure*}

Note that because $\cos{\left(\frac{2\pi}{\lambda} \mathfrak{U}_{\ell+(k-1)L}\right)} \leq 1$, we have
$1 - \frac{2\pi^2}{\lambda^2}\mathfrak{U}_{\ell+(k-1)L}^2
+ \frac{2\pi^4}{3\lambda^4}\mathfrak{U}_{\ell+(k-1)L}^4
\leq 1$,
yielding
$\frac{4\pi^2}{\lambda^2}\mathfrak{U}_{\ell+(k-1)L}^2
- \frac{4\pi^4}{3\lambda^4}\mathfrak{U}_{\ell+(k-1)L}^4
\geq 0$.
Therefore, by utilizing the triangle inequality of
$\left\|\mathbf{v}_{\ell+(k-1)L}-\widehat{\mathbf{p}}-\Delta\mathbf{p}\right\|_2\geq
\left\|\mathbf{v}_{\ell+(k-1)L}-\widehat{\mathbf{p}}\right\|_2 - \left\|\Delta\mathbf{p}\right\|_2
\geq
\left\|\mathbf{v}_{\ell+(k-1)L}-\widehat{\mathbf{p}}\right\|_2 -\epsilon_{\Delta\mathbf{p}}$,
we have 
$\left\|\mathbf{v}_{\ell+(k-1)L}-\widehat{\mathbf{p}}-\Delta\mathbf{p}\right\|_2^{-\frac{\alpha}{2}}
\leq
\left(\left\|\mathbf{v}_{\ell+(k-1)L}-\widehat{\mathbf{p}}\right\|_2 -\epsilon_{\Delta\mathbf{p}}\right)^{-\frac{\alpha}{2}}$
because $\alpha$ is positive, and obtain an upper bound of $\Omega\left(\Delta\mathbf{p}\right)$, denoted by $\Omega_1\left(\Delta\mathbf{p}\right)$, as
\begin{equation}\label{Omega_1_Delta_p}
\Omega\left(\Delta\mathbf{p}\right)\leq \Omega_1\left(\Delta\mathbf{p}\right)=T_1\left(\Delta\mathbf{p}\right) + T_2^{\mathrm{Upp}}\left(\Delta\mathbf{p}\right),
\end{equation}
where $T_2^{\mathrm{Upp}}\left(\Delta\mathbf{p}\right)$ is specified in (\ref{T_2_Upp_D_p}) on the top of the next page, which is an upper bound of $T_2\left(\Delta\mathbf{p}\right)$.
\begin{figure*}[!t]
\normalsize
\begin{scriptsize}
\begin{equation}\label{T_2_Upp_D_p}
\begin{split}
T_2^{\mathrm{Upp}}\left(\Delta\mathbf{p}\right)= \sum_{k=1}^{L}\sum_{\ell=1}^{L}\left\{ \left(\left\|\mathbf{v}_{\ell+(k-1)L}-\widehat{\mathbf{p}}\right\|_2 -\epsilon_{\Delta\mathbf{p}}\right)^{-\frac{\alpha}{2}}
 \left\|\mathbf{v}_{\ell+(k-1)L}-\widehat{\mathbf{p}}\right\|_2^{-\frac{\alpha}{2}} 
\left(\frac{4\pi^2}{\lambda^2}\mathfrak{U}_{\ell+(k-1)L}^2
- \frac{4\pi^4}{3\lambda^4}\mathfrak{U}_{\ell+(k-1)L}^4\right)
\right\}.
\end{split}
\end{equation}
\end{scriptsize}
\vspace*{-18pt}
\end{figure*}

Subsequently, as it is still difficult to derive the maximum of $\Omega_1\left(\Delta\mathbf{p}\right)$ under the constraint of $\|\Delta\mathbf{p}\|_2\leq \epsilon_{\Delta\mathbf{p}}$, we focus on finding an approximation of $T_1\left(\Delta\mathbf{p}\right)$. As the value of $\alpha$ is generally around 2 in the sparse geometry channel model \cite{K.Zhi-WCL2021}, the second-order Taylor approximation for $T_1\left(\Delta\mathbf{p}\right)$ can be sufficiently accurate. Therefore, by using the second-order Taylor expansion:
\begin{equation}\nonumber
T_1\left(\Delta\mathbf{p}\right)\approx
T_1\left(\mathbf{0}\right) +\left( \nabla T_1\left(\mathbf{0}\right)\right)^{\mathrm{T}} \Delta\mathbf{p} + \frac{1}{2} \Delta\mathbf{p}^{\mathrm{T}} \nabla^2 T_1\left(\mathbf{0}\right) \Delta\mathbf{p}
\end{equation}
at $\Delta\mathbf{p} = \mathbf{0}$ to approximate $T_1\left(\Delta\mathbf{p}\right)$, and after a few manipulations, we have
\begin{equation}
\Omega_1\left(\Delta\mathbf{p}\right)\approx
\Delta\mathbf{p}^{\mathrm{T}}\mathbf{R}\Delta\mathbf{p} + T_3\left(\Delta\mathbf{p}\right),
\end{equation}
where $\mathbf{R}$ is given by (\ref{mathbf_R}) in Lemma 1, while $T_3\left(\Delta\mathbf{p}\right)$ is expressed as (\ref{T_3_D_p}) on the top of the next page.
\begin{figure*}[!t]
\normalsize
\begin{scriptsize}
\begin{equation}\label{T_3_D_p}
\begin{split}
T_3\left(\Delta\mathbf{p}\right)= 
-\frac{4\pi^4}{3\lambda^4}\sum_{k=1}^{L}\sum_{\ell=1}^{L}
\left\{ \left[\Delta\mathbf{p}^{\mathrm{T}}\left(\left\|\mathbf{v}_{\ell+(k-1)L}-\widehat{\mathbf{p}}\right\|_2 -\epsilon_{\Delta\mathbf{p}}\right)^{-\frac{\alpha}{4}}
 \left\|\mathbf{v}_{\ell+(k-1)L}-\widehat{\mathbf{p}}\right\|_2^{-\frac{\alpha}{4}} 
 \mathbf{\Xi}_{\ell+(k-1)L} \Delta\mathbf{p}\right]^2
\right\}.
\end{split}
\end{equation}
\end{scriptsize}
\vspace*{-15pt}
\end{figure*}
In (\ref{T_3_D_p}), because a quadratic function with respect to $\Delta\mathbf{p}$ appears as its square form inside the summation operator in $T_3\left(\Delta\mathbf{p}\right)$, we use the power mean inequality, i.e. 
$\frac{1}{N}\sum_{i=1}^{N}x_i\leq\sqrt{\frac{1}{N}\sum_{i=1}^{N}x_i^2}$
for $x_i\geq0$ to derive an upper bound of $T_3\left(\Delta\mathbf{p}\right)$, denoted by $T_3^{\mathrm{Upp}}\left(\Delta\mathbf{p}\right)$, and obtain (\ref{T_3_Upp_D_p}) on the top of the next page, 
\begin{figure*}[!t]
\normalsize
\begin{scriptsize}
\begin{equation}\label{T_3_Upp_D_p}
\begin{split}
T_3\left(\Delta\mathbf{p}\right)\leq T_3^{\mathrm{Upp}}\left(\Delta\mathbf{p}\right) &= 
-\frac{4\pi^4}{3\lambda^4 N}\left(\sum_{k=1}^{L}\sum_{\ell=1}^{L}
\left\{ \Delta\mathbf{p}^{\mathrm{T}}\left(\left\|\mathbf{v}_{\ell+(k-1)L}-\widehat{\mathbf{p}}\right\|_2 -\epsilon_{\Delta\mathbf{p}}\right)^{-\frac{\alpha}{4}}
 \left\|\mathbf{v}_{\ell+(k-1)L}-\widehat{\mathbf{p}}\right\|_2^{-\frac{\alpha}{4}} 
 \mathbf{\Xi}_{\ell+(k-1)L} \Delta\mathbf{p}
\right\}\right)^2 
=  -\frac{4\pi^4}{3\lambda^4 N}\left(\Delta\mathbf{p}^{\mathrm{T}}\mathbf{S}\Delta\mathbf{p}\right)^2.
\end{split}
\end{equation}
\end{scriptsize}
\vspace*{-22pt}
\end{figure*}
where $\mathbf{S}$ is detailed in (\ref{mathbf_S}) in Lemma 1.
Thus, $\Omega_1\left(\Delta\mathbf{p}\right)$ is upper bounded by
\begin{equation}
\Omega_1\left(\Delta\mathbf{p}\right) \leq \Omega_2\left(\Delta\mathbf{p}\right)=\Delta\mathbf{p}^{\mathrm{T}}\mathbf{R}\Delta\mathbf{p} -\frac{4\pi^4}{3\lambda^4 N}\left(\Delta\mathbf{p}^{\mathrm{T}}\mathbf{S}\Delta\mathbf{p}\right)^2.
\end{equation}

Consequently, we have $\Omega\left(\Delta\mathbf{p}\right)\leq \Omega_1\left(\Delta\mathbf{p}\right) \leq \Omega_2\left(\Delta\mathbf{p}\right)$, indicating that $\Omega_2\left(\Delta\mathbf{p}\right)$ is an approximate upper bound of $\Omega\left(\Delta\mathbf{p}\right)$. In order to acquire the maximum of $\Omega_2\left(\Delta\mathbf{p}\right)$ under the constraint of $\|\Delta\mathbf{p}\|_2\leq \epsilon_{\Delta\mathbf{p}}$, we can formulate the following maximization problem:
\begin{subequations}\label{Omega_2_maximization_problem}
\begin{align}
\mathop{\max}\limits_{\Delta\mathbf{p}}&
\ \Omega_2\left(\Delta\mathbf{p}\right)=\Delta\mathbf{p}^{\mathrm{T}}\mathbf{R}\Delta\mathbf{p} -\frac{4\pi^4}{3\lambda^4 N}\left(\Delta\mathbf{p}^{\mathrm{T}}\mathbf{S}\Delta\mathbf{p}\right)^2,
\\ \mathrm{s.t.} &\ \|\Delta\mathbf{p}\|_2\leq \epsilon_{\Delta\mathbf{p}}.
\end{align}
\end{subequations}

By defining $\mathbf{P}=\Delta\mathbf{p}\Delta\mathbf{p}^{\mathrm{T}}$ and using SDR, problem (\ref{Omega_2_maximization_problem}) can be further transformed by dropping the rank-one constraint into problem (\ref{Obtain_CSI_Error_Bound_Problem}) in Lemma 1. Problem (\ref{Obtain_CSI_Error_Bound_Problem}) is convex, because: 1) the $tr(\cdot)$ is linear, while the $-(\cdot)^2 + (\cdot)$ is concave, making the objective function in (\ref{Obtain_CSI_Error_Bound_Problem}a) concave with respect to $\Delta\mathbf{p}$; 2) the constraint (\ref{Obtain_CSI_Error_Bound_Problem}b) is convex. Hence, it can be solved by CVX, after which the maximum of $\Omega_2\left(\Delta\mathbf{p}\right)$ under $\|\Delta\mathbf{p}\|_2\leq \epsilon_{\Delta\mathbf{p}}$ can be derived. Because $\Omega_2\left(\Delta\mathbf{p}\right)$ is an approximate upper bound of $\Omega\left(\Delta\mathbf{p}\right)$, the maximum of $\Omega_2\left(\Delta\mathbf{p}\right)$ is also an approximate upper bound of the maximum of $\Omega\left(\Delta\mathbf{p}\right)$, which completes the proof.



\section{Proof of Proposition 1}


In Appendix B, we prove that problem (\ref{P4}) can be transformed into {\color{black}problem (\ref{P5})}.
Here, the existence of $\left(\mathbf{\Theta}\mathbf{H}_{\mathrm{BR}}\mathbf{H}_{\mathrm{BR}}^{\mathrm{H}}\mathbf{\Theta}^{\mathrm{H}}+\mu\mathbf{I}\right)^{-1}$ is the primary factor that makes problem (\ref{P4}) tough to compute. Fortunately, as $\mu\mathbf{I}$ is invertible, $\left(\mathbf{\Theta}\mathbf{H}_{\mathrm{BR}}\mathbf{H}_{\mathrm{BR}}^{\mathrm{H}}\mathbf{\Theta}^{\mathrm{H}}+\mu\mathbf{I}\right)^{-1}$ can be further expanded with the aided of the matrix inversion lemma. Although the matrix inversion lemma generally leads to an expansion in a more complicated form, it is remarkable that in problem (\ref{P4}), we have {\color{black}$\mathbf{\Theta}^{\mathrm{H}}\mathbf{\Theta}=\beta^2\mathbf{I}$}, which can potentially make the expansion become rather simple.

According to the matrix inversion lemma, we have
\begin{equation}\nonumber
(\mathbf{A} + \mathbf{U}\mathbf{B}\mathbf{V})^{-1} =
\mathbf{A}^{-1} - \mathbf{A}^{-1}\mathbf{U}\mathbf{B}(\mathbf{I}+\mathbf{V}\mathbf{A}^{-1}\mathbf{U}\mathbf{B})^{-1}\mathbf{V}\mathbf{A}^{-1},
\end{equation}
if $\mathbf{A}$ is invertible. Hence, let $\mathbf{A}=\mu\mathbf{I}$, $\mathbf{U}\mathbf{B}=\mathbf{\Theta}\mathbf{H}_{\mathrm{BR}}$ and $\mathbf{V}=\mathbf{H}_{\mathrm{BR}}^{\mathrm{H}}\mathbf{\Theta}^{\mathrm{H}}$. Then, we obtain (\ref{Inverse_Expansion}) on the top of the next page. 
\begin{figure*}[!t]
\normalsize
\begin{scriptsize}
{\color{black}
\begin{equation}\label{Inverse_Expansion}
\begin{split}
\left(\mathbf{\Theta}\mathbf{H}_{\mathrm{BR}}\mathbf{H}_{\mathrm{BR}}^{\mathrm{H}}\mathbf{\Theta}^{\mathrm{H}}+\mu\mathbf{I}\right)^{-1}
&=\mu^{-1}\mathbf{I}-\mu^{-2}\mathbf{\Theta}\mathbf{H}_{\mathrm{BR}} \left(\mathbf{I}+\mu^{-1}\mathbf{H}_{\mathrm{BR}}^{\mathrm{H}}\underbrace{\mathbf{\Theta}^{\mathrm{H}}\mathbf{\Theta}}_{\beta^2\mathbf{I}}\mathbf{H}_{\mathrm{BR}}\right)^{-1}\mathbf{H}_{\mathrm{BR}}^{\mathrm{H}}\mathbf{\Theta}^{\mathrm{H}}
=\mu^{-1}\mathbf{I}-\mu^{-2}\mathbf{\Theta}\mathbf{H}_{\mathrm{BR}} \left(\mathbf{I}+\beta^2\mu^{-1}\mathbf{H}_{\mathrm{BR}}^{\mathrm{H}}\mathbf{H}_{\mathrm{BR}}\right)^{-1}\mathbf{H}_{\mathrm{BR}}^{\mathrm{H}}\mathbf{\Theta}^{\mathrm{H}}.
\end{split}
\end{equation}
}
\end{scriptsize}
\hrulefill
\vspace*{-15pt}
\end{figure*}
Eq. (\ref{Inverse_Expansion}) indicates that $\left(\mathbf{\Theta}\mathbf{H}_{\mathrm{BR}}\mathbf{H}_{\mathrm{BR}}^{\mathrm{H}}\mathbf{\Theta}^{\mathrm{H}}+\mu\mathbf{I}\right)^{-1}$ can be expanded into an expression, in which $\mathbf{\Theta}$ is not included in the inverse operator.
Subsequently, using (\ref{Inverse_Expansion}), we further simplify {\color{black} $\mathcal{F}(\mu,\mathbf{\Theta})$ and $\mathcal{C}(\mu,\mathbf{\Theta})$} in problem (\ref{P4}). First, we simplify $\mathcal{F}(\mu,\mathbf{\Theta})$. By substituting (\ref{Inverse_Expansion}) into (\ref{mathcal_F}), we obtain (\ref{Objective_Fun_Simplify}) on the top of the next page, 
\begin{figure*}[!t]
\normalsize
\begin{scriptsize}
{\color{black}
\begin{equation}\label{Objective_Fun_Simplify}
\begin{split}
&\!\!\!\!\!\!\!\!\!\!\!\!\!\!\!\!\!\!\!\!\!\!\!\!\!\!\mathcal{F}(\mu,\mathbf{\Theta})
=\left\|
\left\{\widehat{\mathbf{h}}_{\mathrm{RU}}^{\mathrm{LoS}}-\left[\mu^{-1}\mathbf{I}-\mu^{-2}\mathbf{\Theta}\mathbf{H}_{\mathrm{BR}} \left(\mathbf{I}+\beta^2\mu^{-1}\mathbf{H}_{\mathrm{BR}}^{\mathrm{H}}\mathbf{H}_{\mathrm{BR}}\right)^{-1}\mathbf{H}_{\mathrm{BR}}^{\mathrm{H}}\mathbf{\Theta}^{\mathrm{H}}\right]\mathbf{\Theta}\mathbf{H}_{\mathrm{BR}}\mathbf{H}_{\mathrm{BR}}^{\mathrm{H}}\mathbf{\Theta}^{\mathrm{H}} \widehat{\mathbf{h}}_{\mathrm{RU}}^{\mathrm{LoS}}\right\}^{\mathrm{H}} \mathbf{\Theta}\mathbf{H}_{\mathrm{BR}}
\right\|_2^2\\
=&\left\|
\left\{\widehat{\mathbf{h}}_{\mathrm{RU}}^{\mathrm{LoS}}-\left[\mu^{-1}\mathbf{\Theta}\mathbf{H}_{\mathrm{BR}}\mathbf{H}_{\mathrm{BR}}^{\mathrm{H}}\mathbf{\Theta}^{\mathrm{H}} \widehat{\mathbf{h}}_{\mathrm{RU}}^{\mathrm{LoS}}-\mu^{-2}\mathbf{\Theta}\mathbf{H}_{\mathrm{BR}} \left(\mathbf{I}+\beta^2\mu^{-1}\mathbf{H}_{\mathrm{BR}}^{\mathrm{H}}\mathbf{H}_{\mathrm{BR}}\right)^{-1}\mathbf{H}_{\mathrm{BR}}^{\mathrm{H}}\underbrace{\mathbf{\Theta}^{\mathrm{H}}\mathbf{\Theta}}_{\beta^2\mathbf{I}}\mathbf{H}_{\mathrm{BR}}\mathbf{H}_{\mathrm{BR}}^{\mathrm{H}}\mathbf{\Theta}^{\mathrm{H}} \widehat{\mathbf{h}}_{\mathrm{RU}}^{\mathrm{LoS}}\right]\right\}^{\mathrm{H}} \mathbf{\Theta}\mathbf{H}_{\mathrm{BR}}
\right\|_2^2\\
=&\left\|
(\widehat{\mathbf{h}}_{\mathrm{RU}}^{\mathrm{LoS}})^{\mathrm{H}} \mathbf{\Theta}\mathbf{H}_{\mathrm{BR}}
-\mu^{-1}(\widehat{\mathbf{h}}_{\mathrm{RU}}^{\mathrm{LoS}})^{\mathrm{H}}\mathbf{\Theta}\mathbf{H}_{\mathrm{BR}}\mathbf{H}_{\mathrm{BR}}^{\mathrm{H}}\underbrace{\mathbf{\Theta}^{\mathrm{H}} \mathbf{\Theta}}_{\beta^2\mathbf{I}}\mathbf{H}_{\mathrm{BR}}
+\beta^2\mu^{-2}(\widehat{\mathbf{h}}_{\mathrm{RU}}^{\mathrm{LoS}})^{\mathrm{H}}\mathbf{\Theta}\mathbf{H}_{\mathrm{BR}}\mathbf{H}_{\mathrm{BR}}^{\mathrm{H}}\mathbf{H}_{\mathrm{BR}} \left(\mathbf{I}+\beta^2\mu^{-1}\mathbf{H}_{\mathrm{BR}}^{\mathrm{H}}\mathbf{H}_{\mathrm{BR}}\right)^{-1}\mathbf{H}_{\mathrm{BR}}^{\mathrm{H}}\underbrace{\mathbf{\Theta}^{\mathrm{H}} \mathbf{\Theta}}_{\beta^2\mathbf{I}}\mathbf{H}_{\mathrm{BR}}
\right\|_2^2\\
=&\left\|
(\widehat{\mathbf{h}}_{\mathrm{RU}}^{\mathrm{LoS}})^{\mathrm{H}} \mathbf{\Theta}\left[
\mathbf{H}_{\mathrm{BR}}
-\beta^2\mu^{-1}\mathbf{H}_{\mathrm{BR}}\mathbf{H}_{\mathrm{BR}}^{\mathrm{H}}\mathbf{H}_{\mathrm{BR}}
+\beta^4\mu^{-2}\mathbf{H}_{\mathrm{BR}}\mathbf{H}_{\mathrm{BR}}^{\mathrm{H}}\mathbf{H}_{\mathrm{BR}} \left(\mathbf{I}+\beta^2\mu^{-1}\mathbf{H}_{\mathrm{BR}}^{\mathrm{H}}\mathbf{H}_{\mathrm{BR}}\right)^{-1}\mathbf{H}_{\mathrm{BR}}^{\mathrm{H}}\mathbf{H}_{\mathrm{BR}}\right]
\right\|_2^2\\
&\!\!\!\!\!\!\!\!\overset{(\mathrm{A1})}{=}
\bm{\theta}^{\mathrm{T}}\underbrace{\mathbf{\Upsilon}(\mu)}_{\mathrm{Independent\ of}\ \mathbf{\Theta}}\bm{\theta}^{\mathrm{*}}.
\end{split}
\end{equation}
}
\end{scriptsize}
\hrulefill
\vspace*{-15pt}
\end{figure*}
where {\color{black} the derivation (A1) uses the properties of $(\widehat{\mathbf{h}}_{\mathrm{RU}}^{\mathrm{LoS}})^{\mathrm{H}} \mathbf{\Theta}=\bm{\theta}^{\mathrm{T}} \mathrm{diag}\{(\widehat{\mathbf{h}}_{\mathrm{RU}}^{\mathrm{LoS}})^{\mathrm{H}}\}$ and $\mathbf{\Theta}^{\mathrm{H}} \widehat{\mathbf{h}}_{\mathrm{RU}}^{\mathrm{LoS}} =\mathrm{diag}\{\widehat{\mathbf{h}}_{\mathrm{RU}}^{\mathrm{LoS}}\} \bm{\theta}^{\mathrm{*}} $}, and $\mathbf{\Upsilon}(\mu)$ is expressed {\color{black}as
\begin{equation}\label{mathbf_Y}
\begin{split}
\!\!\mathbf{\Upsilon}(\mu)\!\!=&
\mathrm{diag}\{(\widehat{\mathbf{h}}_{\mathrm{RU}}^{\mathrm{LoS}})^{\mathrm{H}}\} \mathbf{Z}(\mu) (\mathbf{Z}(\mu))^{\mathrm{H}} \mathrm{diag}\{\widehat{\mathbf{h}}_{\mathrm{RU}}^{\mathrm{LoS}}\},
\end{split}
\end{equation}
with $\mathbf{Z}(\mu)$ given by
\begin{equation}\label{Z_mu}
\begin{split}
\mathbf{Z}(\mu) = & 
\mathbf{H}_{\mathrm{BR}}
-\beta^2\mu^{-1}\mathbf{H}_{\mathrm{BR}}\mathbf{H}_{\mathrm{BR}}^{\mathrm{H}}\mathbf{H}_{\mathrm{BR}}\\
&+\beta^4\mu^{-2}\mathbf{H}_{\mathrm{BR}}\mathbf{H}_{\mathrm{BR}}^{\mathrm{H}}\mathbf{H}_{\mathrm{BR}} \\
&\times\left(\mathbf{I}+\beta^2\mu^{-1}\mathbf{H}_{\mathrm{BR}}^{\mathrm{H}}\mathbf{H}_{\mathrm{BR}}\right)^{-1}
\mathbf{H}_{\mathrm{BR}}^{\mathrm{H}}\mathbf{H}_{\mathrm{BR}}  \\
=& \left[ \mathbf{I} -  \mathbf{H}_{\mathrm{BR}}\mathbf{H}_{\mathrm{BR}}^{\mathrm{H}} (\beta^{-2}\mu \mathbf{I} + \mathbf{H}_{\mathrm{BR}}\mathbf{H}_{\mathrm{BR}}^{\mathrm{H}})^{-1} \right] \mathbf{H}_{\mathrm{BR}}.
\end{split}
\end{equation}
Then, we simplify $\mathcal{C}(\mu,\mathbf{\Theta})$}. Substituting (\ref{Inverse_Expansion}) into (\ref{bisection}), we obtain (\ref{Constraint_Simplify}) on the top of the next page,
\begin{figure*}[!t]
\normalsize
\begin{scriptsize}
{\color{black}
\begin{equation}\label{Constraint_Simplify}
\begin{split}
\mathcal{C}(\mu,\mathbf{\Theta})=&
\left[\left(\mathbf{\Theta}\mathbf{H}_{\mathrm{BR}}\mathbf{H}_{\mathrm{BR}}^{\mathrm{H}}\mathbf{\Theta}^{\mathrm{H}}+\mu\mathbf{I}\right)^{-1}\mathbf{\Theta}\mathbf{H}_{\mathrm{BR}}\mathbf{H}_{\mathrm{BR}}^{\mathrm{H}}\mathbf{\Theta}^{\mathrm{H}} \widehat{\mathbf{h}}_{\mathrm{RU}}^{\mathrm{LoS}}\right]^{\mathrm{H}}
\left[\left(\mathbf{\Theta}\mathbf{H}_{\mathrm{BR}}\mathbf{H}_{\mathrm{BR}}^{\mathrm{H}}\mathbf{\Theta}^{\mathrm{H}}+\mu\mathbf{I}\right)^{-1}\mathbf{\Theta}\mathbf{H}_{\mathrm{BR}}\mathbf{H}_{\mathrm{BR}}^{\mathrm{H}}\mathbf{\Theta}^{\mathrm{H}} \widehat{\mathbf{h}}_{\mathrm{RU}}^{\mathrm{LoS}}\right] - \epsilon_{\Delta \mathbf{h}_{\mathrm{RU}}}^2\\
=&
\left[\left(\mu^{-1}\mathbf{I}-\mu^{-2}\mathbf{\Theta}\mathbf{H}_{\mathrm{BR}} \left(\mathbf{I}+\beta^2\mu^{-1}\mathbf{H}_{\mathrm{BR}}^{\mathrm{H}}\mathbf{H}_{\mathrm{BR}}\right)^{-1}\mathbf{H}_{\mathrm{BR}}^{\mathrm{H}}\mathbf{\Theta}^{\mathrm{H}}\right)\mathbf{\Theta}\mathbf{H}_{\mathrm{BR}}\mathbf{H}_{\mathrm{BR}}^{\mathrm{H}}\mathbf{\Theta}^{\mathrm{H}} \widehat{\mathbf{h}}_{\mathrm{RU}}^{\mathrm{LoS}}\right]^{\mathrm{H}}\\
&\times \left[\left(\mu^{-1}\mathbf{I}-\mu^{-2}\mathbf{\Theta}\mathbf{H}_{\mathrm{BR}} \left(\mathbf{I}+\beta^2\mu^{-1}\mathbf{H}_{\mathrm{BR}}^{\mathrm{H}}\mathbf{H}_{\mathrm{BR}}\right)^{-1}\mathbf{H}_{\mathrm{BR}}^{\mathrm{H}}\mathbf{\Theta}^{\mathrm{H}}\right)\mathbf{\Theta}\mathbf{H}_{\mathrm{BR}}\mathbf{H}_{\mathrm{BR}}^{\mathrm{H}}\mathbf{\Theta}^{\mathrm{H}} \widehat{\mathbf{h}}_{\mathrm{RU}}^{\mathrm{LoS}}\right] - \epsilon_{\Delta \mathbf{h}_{\mathrm{RU}}}^2\\
=&
\left[\mu^{-1}\mathbf{\Theta}\mathbf{H}_{\mathrm{BR}}\mathbf{H}_{\mathrm{BR}}^{\mathrm{H}}\mathbf{\Theta}^{\mathrm{H}} \widehat{\mathbf{h}}_{\mathrm{RU}}^{\mathrm{LoS}}-\mu^{-2}\mathbf{\Theta}\mathbf{H}_{\mathrm{BR}} \left(\mathbf{I}+\beta^2\mu^{-1}\mathbf{H}_{\mathrm{BR}}^{\mathrm{H}}\mathbf{H}_{\mathrm{BR}}\right)^{-1}\mathbf{H}_{\mathrm{BR}}^{\mathrm{H}}\underbrace{\mathbf{\Theta}^{\mathrm{H}} \mathbf{\Theta}}_{\beta^2\mathbf{I}}\mathbf{H}_{\mathrm{BR}}\mathbf{H}_{\mathrm{BR}}^{\mathrm{H}}\mathbf{\Theta}^{\mathrm{H}} \widehat{\mathbf{h}}_{\mathrm{RU}}^{\mathrm{LoS}}\right]^{\mathrm{H}}\\
&\times \left[\mu^{-1}\mathbf{\Theta}\mathbf{H}_{\mathrm{BR}}\mathbf{H}_{\mathrm{BR}}^{\mathrm{H}}\mathbf{\Theta}^{\mathrm{H}} \widehat{\mathbf{h}}_{\mathrm{RU}}^{\mathrm{LoS}}-\mu^{-2}\mathbf{\Theta}\mathbf{H}_{\mathrm{BR}} \left(\mathbf{I}+\beta^2\mu^{-1}\mathbf{H}_{\mathrm{BR}}^{\mathrm{H}}\mathbf{H}_{\mathrm{BR}}\right)^{-1}\mathbf{H}_{\mathrm{BR}}^{\mathrm{H}}\underbrace{\mathbf{\Theta}^{\mathrm{H}} \mathbf{\Theta}}_{\beta^2\mathbf{I}}\mathbf{H}_{\mathrm{BR}}\mathbf{H}_{\mathrm{BR}}^{\mathrm{H}}\mathbf{\Theta}^{\mathrm{H}} \widehat{\mathbf{h}}_{\mathrm{RU}}^{\mathrm{LoS}}\right] - \epsilon_{\Delta \mathbf{h}_{\mathrm{RU}}}^2\\
=&\ 
\mu^{-2}(\widehat{\mathbf{h}}_{\mathrm{RU}}^{\mathrm{LoS}})^{\mathrm{H}}\mathbf{\Theta}\mathbf{H}_{\mathrm{BR}}\mathbf{H}_{\mathrm{BR}}^{\mathrm{H}}\underbrace{\mathbf{\Theta}^{\mathrm{H}} \mathbf{\Theta}}_{\beta^2\mathbf{I}}\mathbf{H}_{\mathrm{BR}}\mathbf{H}_{\mathrm{BR}}^{\mathrm{H}}\mathbf{\Theta}^{\mathrm{H}}\widehat{\mathbf{h}}_{\mathrm{RU}}^{\mathrm{LoS}}\\
&-
\beta^2\mu^{-3}(\widehat{\mathbf{h}}_{\mathrm{RU}}^{\mathrm{LoS}})^{\mathrm{H}}\mathbf{\Theta}\mathbf{H}_{\mathrm{BR}}\mathbf{H}_{\mathrm{BR}}^{\mathrm{H}}\underbrace{\mathbf{\Theta}^{\mathrm{H}} \mathbf{\Theta}}_{\beta^2\mathbf{I}}\mathbf{H}_{\mathrm{BR}} \left(\mathbf{I}+\beta^2\mu^{-1}\mathbf{H}_{\mathrm{BR}}^{\mathrm{H}}\mathbf{H}_{\mathrm{BR}}\right)^{-1}
\mathbf{H}_{\mathrm{BR}}^{\mathrm{H}}\mathbf{H}_{\mathrm{BR}}\mathbf{H}_{\mathrm{BR}}^{\mathrm{H}} \mathbf{\Theta}^{\mathrm{H}}\widehat{\mathbf{h}}_{\mathrm{RU}}^{\mathrm{LoS}}\\
& - \beta^2\mu^{-3}(\widehat{\mathbf{h}}_{\mathrm{RU}}^{\mathrm{LoS}})^{\mathrm{H}}\mathbf{\Theta}\mathbf{H}_{\mathrm{BR}}\mathbf{H}_{\mathrm{BR}}^{\mathrm{H}}\mathbf{H}_{\mathrm{BR}} \left(\mathbf{I}+\beta^2\mu^{-1}\mathbf{H}_{\mathrm{BR}}^{\mathrm{H}}\mathbf{H}_{\mathrm{BR}}\right)^{-1}
\mathbf{H}_{\mathrm{BR}}^{\mathrm{H}}\underbrace{\mathbf{\Theta}^{\mathrm{H}} \mathbf{\Theta}}_{\beta^2\mathbf{I}}\mathbf{H}_{\mathrm{BR}}\mathbf{H}_{\mathrm{BR}}^{\mathrm{H}} \mathbf{\Theta}^{\mathrm{H}}\widehat{\mathbf{h}}_{\mathrm{RU}}^{\mathrm{LoS}}\\
&+ \beta^4\mu^{-4}(\widehat{\mathbf{h}}_{\mathrm{RU}}^{\mathrm{LoS}})^{\mathrm{H}}\mathbf{\Theta}\mathbf{H}_{\mathrm{BR}}\mathbf{H}_{\mathrm{BR}}^{\mathrm{H}}\mathbf{H}_{\mathrm{BR}} \left(\mathbf{I}+\beta^2\mu^{-1}\mathbf{H}_{\mathrm{BR}}^{\mathrm{H}}\mathbf{H}_{\mathrm{BR}}\right)^{-1}\mathbf{H}_{\mathrm{BR}}^{\mathrm{H}}\underbrace{\mathbf{\Theta}^{\mathrm{H}} \mathbf{\Theta}}_{\beta^2\mathbf{I}}\mathbf{H}_{\mathrm{BR}} \left(\mathbf{I}+\beta^2\mu^{-1}\mathbf{H}_{\mathrm{BR}}^{\mathrm{H}}\mathbf{H}_{\mathrm{BR}}\right)^{-1} \mathbf{H}_{\mathrm{BR}}^{\mathrm{H}}\mathbf{H}_{\mathrm{BR}}\mathbf{H}_{\mathrm{BR}}^{\mathrm{H}}
 \mathbf{\Theta}^{\mathrm{H}}\widehat{\mathbf{h}}_{\mathrm{RU}}^{\mathrm{LoS}} 
 - \epsilon_{\Delta \mathbf{h}_{\mathrm{RU}}}^2 \\
\overset{\mathrm{(A1)}}{=} & \ 
\bm{\theta}^{\mathrm{T}}\underbrace{\mathbf{\Gamma}(\mu)}_{\mathrm{Independent\ of}\ \mathbf{\Theta}}\bm{\theta}^{\mathrm{*}} - \epsilon_{\Delta \mathbf{h}_{\mathrm{RU}}}^2
.
\end{split}
\end{equation}
}
\end{scriptsize}
\hrulefill
\vspace*{-15pt}
\end{figure*}
where {\color{black}$\mathbf{\Gamma}(\mu)$ is expressed as
\begin{equation}\label{Gamma_mu}
\begin{split}
\mathbf{\Gamma}(\mu)=&
\mathrm{diag}\{(\widehat{\mathbf{h}}_{\mathrm{RU}}^{\mathrm{LoS}})^{\mathrm{H}}\} \mathbf{X}(\mu) \mathrm{diag}\{\widehat{\mathbf{h}}_{\mathrm{RU}}^{\mathrm{LoS}}\},
\end{split}
\end{equation}
with $\mathbf{X}(\mu)$ given by
\begin{equation}\label{mathbf_X}
\begin{split}
\mathbf{X}(\mu)=&\ \beta^2\mu^{-2}\mathbf{H}_{\mathrm{BR}}\mathbf{H}_{\mathrm{BR}}^{\mathrm{H}}\mathbf{H}_{\mathrm{BR}}\mathbf{H}_{\mathrm{BR}}^{\mathrm{H}} - 2\beta^4\mu^{-3}\mathbf{H}_{\mathrm{BR}}\mathbf{H}_{\mathrm{BR}}^{\mathrm{H}}\\
&\times \mathbf{H}_{\mathrm{BR}}\left(\mathbf{I}+\beta^2\mu^{-1}\mathbf{H}_{\mathrm{BR}}^{\mathrm{H}}\mathbf{H}_{\mathrm{BR}}\right)^{-1} \mathbf{H}_{\mathrm{BR}}^{\mathrm{H}} \mathbf{H}_{\mathrm{BR}} \mathbf{H}_{\mathrm{BR}}^{\mathrm{H}}\\
&+\beta^6 \mu^{-4}\mathbf{H}_{\mathrm{BR}}\mathbf{H}_{\mathrm{BR}}^{\mathrm{H}}\mathbf{H}_{\mathrm{BR}}\left(\mathbf{I}+\beta^2\mu^{-1}\mathbf{H}_{\mathrm{BR}}^{\mathrm{H}}\mathbf{H}_{\mathrm{BR}}\right)^{-1}\\
&\times \mathbf{H}_{\mathrm{BR}}^{\mathrm{H}}\mathbf{H}_{\mathrm{BR}}\left(\mathbf{I}+\beta^2\mu^{-1}\mathbf{H}_{\mathrm{BR}}^{\mathrm{H}}\mathbf{H}_{\mathrm{BR}}\right)^{-1}
 \mathbf{H}_{\mathrm{BR}}^{\mathrm{H}}\mathbf{H}_{\mathrm{BR}}\mathbf{H}_{\mathrm{BR}}^{\mathrm{H}}\\
 = &\  \beta^2 \mathbf{H}_{\mathrm{BR}}\mathbf{H}_{\mathrm{BR}}^{\mathrm{H}} (\mu \mathbf{I} + \beta^2 \mathbf{H}_{\mathrm{BR}}\mathbf{H}_{\mathrm{BR}}^{\mathrm{H}})^{-1}\\
 &\ \ \ \ \ \ \ \ \ \ \ \  \times (\mu \mathbf{I} + \beta^2 \mathbf{H}_{\mathrm{BR}}\mathbf{H}_{\mathrm{BR}}^{\mathrm{H}})^{-1} \mathbf{H}_{\mathrm{BR}}\mathbf{H}_{\mathrm{BR}}^{\mathrm{H}}.
\end{split}
\end{equation}

Combining the outcomes shown in (\ref{Objective_Fun_Simplify}) and (\ref{Constraint_Simplify}), 
we finally completes the proof of Proposition 1}.

{\color{black}
\section{Proof of Proposition 2}

Here, we prove that we can obtain $rank(\overline{\mathbf{C}})=1$ under the conditions of: \textit{1) $\mathbf{\Upsilon}(\mu)$ is rank-one, and each element in its eigenvector is non-zero; 2) $\overline{\mathbf{C}}$ satisfies $tr\left(\mathbf{\Gamma}(\mu)\overline{\mathbf{C}}\right)>\epsilon_{\Delta \mathbf{h}_{\mathrm{RU}}}^2$.}

First, the Lagrange function of problem (\ref{P6-SDR}) is
\begin{equation}
\begin{split}
\!\!\!\!\!&\mathcal{L}(\mathbf{C},\lambda_{\mathrm{b}},\lambda_{\mathrm{c},1},\cdots,\lambda_{\mathrm{c},N},\mathbf{M}) \\
=&
-tr\left(\mathbf{\Upsilon}(\mu)\mathbf{C}\right)
+ \lambda_{\mathrm{b}}\left[ \epsilon_{\Delta \mathbf{h}_{\mathrm{RU}}}^2 - tr\left(\mathbf{\Gamma}(\mu)\mathbf{C}\right) \right] \\
&+ \sum_{i=1}^N  \lambda_{\mathrm{c},i}\left\{ tr(\mathbf{E}_{i}\mathbf{C})- \beta^2 \right\} - tr(\mathbf{M}\mathbf{C}),
\end{split}
\end{equation}
where $\lambda_{\mathrm{b}}$, $\lambda_{\mathrm{c},1}$ to $\lambda_{\mathrm{c},N}$, and $\mathbf{M}$ are the Lagrange multipliers of the constraint (\ref{P6-SDR}b), constraint (\ref{P6-SDR}c), and constraint $\mathbf{C}\succeq\mathbf{0}$, respectively.

Then, based on the KKT conditions, the optimal solution $\overline{\mathbf{C}}$ should simultaneously satisfy
\begin{equation}\label{Rank-1-Lag}
\frac{\partial \mathcal{L}(\overline{\mathbf{C}},\lambda_{\mathrm{b}},\lambda_{\mathrm{c},1},\cdots,\lambda_{\mathrm{c},N},\mathbf{M})}{\partial \overline{\mathbf{C}}} = \mathbf{0},
\end{equation}
\begin{equation}
\lambda_{\mathrm{b}}\geq 0,  \mathbf{M}\succeq\mathbf{0},    \mathbf{M}\overline{\mathbf{C}} = \mathbf{0},
\end{equation}
\begin{equation}\label{Rank-1-Lag-2}
\lambda_{\mathrm{b}}\left[ \epsilon_{\Delta \mathbf{h}_{\mathrm{RU}}}^2 - tr\left(\mathbf{\Gamma}(\mu)\overline{\mathbf{C}}\right) \right] = 0.
\end{equation}

According to (\ref{Rank-1-Lag}), we have
\begin{equation}\label{M-Expre}
\begin{split}
- \mathbf{\Upsilon}(\mu)
- \lambda_{\mathrm{b}} \mathbf{\Gamma}(\mu)
+ \sum_{i=1}^N  \lambda_{\mathrm{c},i} \mathbf{E}_{i} 
= \mathbf{M}.
\end{split}
\end{equation}
If $tr\left(\mathbf{\Gamma}(\mu)\overline{\mathbf{C}}\right)>\epsilon_{\Delta \mathbf{h}_{\mathrm{RU}}}^2$, we obtain $\lambda_{\mathrm{b}}= 0$ according to the complementary slacking constraint (\ref{Rank-1-Lag-2}). Then, we can simplify (\ref{M-Expre}) into
\begin{equation}\label{M-Expre-2}
\begin{split}
- \mathbf{\Upsilon}(\mu)
+ \sum_{i=1}^N  \lambda_{\mathrm{c},i} \mathbf{E}_{i} 
= \mathbf{M}.
\end{split}
\end{equation}

Note that because $\mathbf{\Upsilon}(\mu)\succeq \mathbf{0}$ and $\mathbf{M}\succeq \mathbf{0}$, there should be $ \sum_{i=1}^N  \lambda_{\mathrm{c},i} \mathbf{E}_{i} \succeq \mathbf{0}$, hinting that $\lambda_{\mathrm{c},i}\geq 0$ for $i=1,2,\cdots,N$.
Then, we use $\mathbf{M}\overline{\mathbf{C}} = \mathbf{0}$ to transform (\ref{M-Expre-2}) into
\begin{equation}
\begin{split}
\left[- \mathbf{\Upsilon}(\mu)
+ \sum_{i=1}^N  \lambda_{\mathrm{c},i} \mathbf{E}_{i} \right] \overline{\mathbf{C}}
= \mathbf{0},
\end{split}
\end{equation}
yielding
\begin{equation}\label{MCzero}
\begin{split}
\left(\sum_{i=1}^N  \lambda_{\mathrm{c},i} \mathbf{E}_{i} \right)   \overline{\mathbf{C}}
= \mathbf{\Upsilon}(\mu) \overline{\mathbf{C}}.
\end{split}
\end{equation}

If $rank\left( \mathbf{\Upsilon}(\mu)\right)=1$ and each element in the eigenvector of $\mathbf{\Upsilon}(\mu)$ is non-zero, we can decompose $\mathbf{\Upsilon}(\mu)$ into $\mathbf{\Upsilon}(\mu)=\lambda_\mathbf{\Upsilon}\mathbf{u}\mathbf{u}^{\mathrm{H}}$, where $\lambda_\mathbf{\Upsilon}$ is the eigenvalue of $\mathbf{\Upsilon}(\mu)$, and $\mathbf{u}$ is the corresponding eigenvector with $[\mathbf{u}]_i \neq 0$ for $\forall i=1,2,\cdots,N$. Therefore, (\ref{MCzero}) can be recast as
\begin{equation}\label{MCzero-2}
\begin{split}
\left(\sum_{i=1}^N  \lambda_{\mathrm{c},i} \mathbf{E}_{i} \right)   \overline{\mathbf{C}}
= \lambda_\mathbf{\Upsilon}\mathbf{u}\mathbf{u}^{\mathrm{H}} \overline{\mathbf{C}}.
\end{split}
\end{equation}

From (\ref{MCzero-2}), one can prove that $\left(\sum_{i=1}^N  \lambda_{\mathrm{c},i} \mathbf{E}_{i} \right)$ is a full-rank diagonal matrix. This is because that if $\left(\sum_{i=1}^N  \lambda_{\mathrm{c},i} \mathbf{E}_{i} \right)$ is not full-rank, there will be at least one $\lambda_{\mathrm{c},i}=0$ such that $\left(\sum_{i=1}^N  \lambda_{\mathrm{c},i} \mathbf{E}_{i} \right)   \overline{\mathbf{C}}$ has at least one row of zeros. Then, $\lambda_\mathbf{\Upsilon}\mathbf{u}\mathbf{u}^{\mathrm{H}} \overline{\mathbf{C}}$ will have at least one row of zeros as well. However, since $\left(\mathbf{u}^{\mathrm{H}} \overline{\mathbf{C}}\right)$ cannot be an all-zero vector, there should exist at least one $ [\mathbf{u}]_i = 0$ generating a row of zeros in $\mathbf{u}(\mathbf{u}^{\mathrm{H}} \overline{\mathbf{C}})$, which contradicts the prerequisite of $[\mathbf{u}]_i \neq 0$ for $\forall i=1,2,\cdots,N$. Therefore, $\left(\sum_{i=1}^N  \lambda_{\mathrm{c},i} \mathbf{E}_{i} \right)$ should be full-rank.

Finally, according to the manipulation rules of rank, we have
\begin{equation}\nonumber
\begin{split}
rank\left(\overline{\mathbf{C}}\right)\!\!&=rank\left( \left(\sum_{i=1}^N  \lambda_{\mathrm{c},i} \mathbf{E}_{i} \right)   \overline{\mathbf{C}}\right)
= rank\left( \lambda_\mathbf{\Upsilon}\mathbf{u}\mathbf{u}^{\mathrm{H}} \overline{\mathbf{C}}\right)\\
& \leq \min\left\{rank\left( \lambda_\mathbf{\Upsilon}\mathbf{u}\mathbf{u}^{\mathrm{H}} \right)
,rank\left(\overline{\mathbf{C}}\right)\right\}=1.
\end{split}
\end{equation}
Since $rank\left(\overline{\mathbf{C}}\right)>0$ must hold, we have $rank\left(\overline{\mathbf{C}}\right)= 1$, which completes the proof of Proposition 2.

\section{Proof of Proposition 3}
We begin this proof by calculating the derivatives of $\bm{\theta}^{\mathrm{T}} \mathbf{\Upsilon}(\mu) \bm{\theta}^*$ and $\bm{\theta}^{\mathrm{T}} \mathbf{\Gamma}(\mu) \bm{\theta}^*$ with respect to $\mu$, i.e. calculating
\begin{equation}
D_{obj}(\mu)=\frac{\partial (\bm{\theta}^{\mathrm{T}} \mathbf{\Upsilon}(\mu) \bm{\theta}^*)}{\partial \mu} =\bm{\theta}^{\mathrm{T}} \frac{\partial  \mathbf{\Upsilon}(\mu) }{\partial \mu} \bm{\theta}^*,
\end{equation}
\begin{equation}
D_{cons}(\mu)=\frac{\partial (\bm{\theta}^{\mathrm{T}} \mathbf{\Gamma}(\mu) \bm{\theta}^*)}{\partial \mu} =\bm{\theta}^{\mathrm{T}} \frac{\partial  \mathbf{\Gamma}(\mu) }{\partial \mu} \bm{\theta}^*.
\end{equation}
To derive $\frac{\partial  \mathbf{\Upsilon}(\mu) }{\partial \mu}$, we first perform eigenvalue decomposition for $\mathbf{H}_{\mathrm{BR}}\mathbf{H}_{\mathrm{BR}}^{\mathrm{H}}$ and obtain $\mathbf{H}_{\mathrm{BR}}\mathbf{H}_{\mathrm{BR}}^{\mathrm{H}} = \mathbf{U}_1 \mathbf{\Sigma}_1 \mathbf{U}_1^{\mathrm{H}}$, where $\mathbf{U}_1$ is a unitary matrix and $\mathbf{\Sigma}_1$ is the diagonal eigenvalue matrix. $\mathbf{\Sigma}_1$ satisfies $\mathbf{\Sigma}_1\succeq \mathbf{0}$ since $\mathbf{H}_{\mathrm{BR}}\mathbf{H}_{\mathrm{BR}}^{\mathrm{H}}$ is positive semidefinite.

Then, based on (\ref{mathbf_Y}), (\ref{Z_mu}), and after a few manipulations, $\mathbf{\Upsilon}(\mu)$ can be recast as
\begin{equation}
\mathbf{\Upsilon}(\mu) = \mathrm{diag}\{(\widehat{\mathbf{h}}_{\mathrm{RU}}^{\mathrm{LoS}})^{\mathrm{H}}\}      \mathbf{U}_1 \mathbf{\Sigma}'_1(\mu) \mathbf{U}_1^{\mathrm{H}}    \mathrm{diag}\{\widehat{\mathbf{h}}_{\mathrm{RU}}^{\mathrm{LoS}}\},
\end{equation}
where $\mathbf{\Sigma}'_1(\mu)$ with respect to $\mu$ is given by
\begin{equation}\label{Sigma'}
\mathbf{\Sigma}'_1(\mu) \!\!=\!\! \left[  \mathbf{I} -  \mathbf{\Sigma}_1 (\beta^{-2} \mu \mathbf{I} + \mathbf{\Sigma}_1)^{-1}  \right] \mathbf{\Sigma}_1 \left[  \mathbf{I} -   (\beta^{-2} \mu \mathbf{I} + \mathbf{\Sigma}_1)^{-1} \mathbf{\Sigma}_1  \right].
\end{equation}

It is indicated in (\ref{Sigma'}) that $\mathbf{\Sigma}'_1(\mu)$ is a diagonal matrix, whose $(i,i)$-th diagonal element is given by
\begin{equation}
\left[\mathbf{\Sigma}'_1(\mu)\right]_{(i,i)} = \left(1-\frac{\left[\mathbf{\Sigma}_1\right]_{(i,i)}}{\beta^{-2} \mu + \left[\mathbf{\Sigma}_1\right]_{(i,i)}}\right)^2  \left[\mathbf{\Sigma}_1\right]_{(i,i)}.
\end{equation}
We calculate the derivative of each $\left[\mathbf{\Sigma}'_1(\mu)\right]_{(i,i)}$ with respect to $\mu$, and obtain
\begin{equation}\nonumber
\begin{split}
\frac{\partial \left[\mathbf{\Sigma}'_1(\mu)\right]_{(i,i)}}{\partial \mu} = &
\{2(\beta^{-8}-\beta^{-6})\mu^3 \left[\mathbf{\Sigma}_1\right]_{(i,i)}\\
&+ (4\beta^{-6}-2\beta^{-4})\mu^2 \left[\mathbf{\Sigma}_1\right]_{(i,i)}^2\\
& + 2 \beta^{-4}\mu \left[\mathbf{\Sigma}_1\right]_{(i,i)}^3\} \times \{\beta^{-2}\mu + \left[\mathbf{\Sigma}_1\right]_{(i,i)}\}^{-4}.
\end{split}
\end{equation}

Because $\left[\mathbf{\Sigma}_1\right]_{(i,i)}\geq 0$ and $0< \beta\leq 1$, we have $\frac{\partial \left[\mathbf{\Sigma}'_1(\mu)\right]_{(i,i)}}{\partial \mu} \geq 0$, resulting in $\frac{\partial \mathbf{\Sigma}'_1(\mu)}{\partial \mu} \succeq \mathbf{0}$. This implies that
\begin{equation}\nonumber
\frac{\partial \mathbf{\Upsilon}(\mu)}{\partial \mu} = \mathrm{diag}\{(\widehat{\mathbf{h}}_{\mathrm{RU}}^{\mathrm{LoS}})^{\mathrm{H}}\}      \mathbf{U}_1 \frac{\partial \mathbf{\Sigma}'_1(\mu)}{\partial \mu} \mathbf{U}_1^{\mathrm{H}}    \mathrm{diag}\{\widehat{\mathbf{h}}_{\mathrm{RU}}^{\mathrm{LoS}}\}  \succeq \mathbf{0}.
\end{equation}

As a result, we obtain that
\begin{equation}\label{D_obj_positive}
D_{obj}(\mu) =\bm{\theta}^{\mathrm{T}} \frac{\partial  \mathbf{\Upsilon}(\mu) }{\partial \mu} \bm{\theta}^* \geq 0
\end{equation}
strictly holds for $\forall \bm{\theta}$.

Based on (\ref{Gamma_mu}) and (\ref{mathbf_X}), we perform the same derivation for $\frac{\partial  \mathbf{\Gamma}(\mu) }{\partial \mu}$ and obtain 
\begin{equation}\label{D_cons_negative}
D_{cons}(\mu) =\bm{\theta}^{\mathrm{T}} \frac{\partial  \mathbf{\Gamma}(\mu) }{\partial \mu} \bm{\theta}^* \leq 0.
\end{equation}

Finally, combining the results in (\ref{D_obj_positive}) and (\ref{D_cons_negative}), we complete the proof of Proposition 3.

}


\end{document}